\documentclass[11pt]{article}
\usepackage{amsmath,amsthm,amsfonts,amssymb}
\usepackage{cancel}
\usepackage{fullpage}
\usepackage{liyang}
\usepackage{framed}
\usepackage{verbatim}
\usepackage{enumitem}
\usepackage{array}
\usepackage{multirow}
\usepackage{afterpage}

\usepackage{caption} 

\usepackage[usenames,dvipsnames]{xcolor}

\usepackage{todonotes}

\makeatletter
\newtheorem*{rep@theorem}{\rep@title}
\newcommand{\newreptheorem}[2]{
\newenvironment{rep#1}[1]{
 \def\rep@title{#2 \ref{##1}}
 \begin{rep@theorem}\itshape}
 {\end{rep@theorem}}}
\makeatother
\theoremstyle{plain}



\def\colorful{1}

\ifnum\colorful=1

\newcommand{\gray}[1]{{\color{gray}{#1}}}

\fi
\ifnum\colorful=0

\newcommand{\gray}[1]{{{#1}}}

\fi

\usepackage{boxedminipage}

\newcommand{\ignore}[1]{}
\newcommand{\lnote}[1]{\footnote{{\bf \color{blue}Li-Yang}: {#1}}}
\newcommand{\rnote}[1]{\footnote{{\bf \color{orange}Rocco:} {#1}}}

\newreptheorem{theorem}{Theorem}
\newtheorem*{theorem*}{Theorem}
\newreptheorem{lemma}{Lemma}
\newreptheorem{proposition}{Proposition}
\newtheorem*{noclaim*}{Claim}

\newcommand{\uhr}{\upharpoonright}

\newcommand{\Tribes}{{\sf{Tribes}}}

\newcommand{\encode}{\mathrm{encode}}

\newcommand{\proj}{\mathrm{proj}}

\newcommand{\Tree}{\mathsf{ProjDT}}

\newcommand{\Parity}{\mathsf{Parity}}
\newcommand{\gradual}{\beta}

\newcommand{\BalancedSipser}{\mathsf{Sipser}}

\newcommand{\init}{\mathrm{init}}

\newcommand{\acz}{\mathsf{AC^0}}

\newcommand{\ds}{\displaystyle}

\title{An average-case depth hierarchy theorem for Boolean circuits}

\author{Benjamin Rossman \\ NII, Simons Institute \\ {\tt rossman@nii.ac.jp} \and Rocco A.~Servedio\thanks{Supported by NSF grants CCF-1319788 and CCF-1420349.}\\ 
Columbia University \\{\tt rocco@cs.columbia.edu} \and Li-Yang Tan\thanks{Part of this research was done while visiting Columbia University.} \\ Simons Institute \\ {\tt liyang@cs.columbia.edu}}
\begin{document}

 \maketitle

\begin{abstract}
We prove an average-case depth hierarchy theorem for Boolean circuits over the standard basis of $\AND$, $\OR$, and $\mathsf{NOT}$ gates.  Our hierarchy theorem says that for every $d \geq 2$, there is an explicit $n$-variable Boolean function $f$, computed by a linear-size depth-$d$ formula, which is such that any depth-$(d-1)$ circuit that agrees with $f$ on $(1/2 + o_n(1))$ fraction of all inputs must have size $\exp({n^{\Omega(1/d)}}).$  This answers an open question posed by H\aa stad in his Ph.D.~thesis \cite{Hastad:86}.

Our average-case depth hierarchy theorem implies that the polynomial hierarchy is infinite relative to a random oracle with probability 1, confirming a conjecture of H{\aa}stad~\cite{Hastad86}, Cai~\cite{Cai86}, and Babai~\cite{babai1987}.  We also use our result to show that there is no ``approximate converse'' to the results of Linial, Mansour, Nisan~\cite{linmannis93} and Boppana~\cite{Boppana97} on the total influence of small-depth circuits, thus answering a question posed by O'Donnell~\cite{ODonnell07open}, Kalai~\cite{Kalaioverflow12}, and Hatami~\cite{Hatami14}.

A key ingredient in our proof is a notion of \emph{random projections} which generalize random restrictions.

\end{abstract}
%
%\rnote{Todo:  finish standardizing our hierarchy:  Theorem $>$ Proposition $>$ Lemma $>$ Fact (``Claim'' is deprecated)}

 \thispagestyle{empty}

\newpage

 \thispagestyle{empty}

\hypersetup{linkcolor=magenta}
\hypersetup{linktocpage}
\setcounter{tocdepth}{2}

\tableofcontents
 \thispagestyle{empty}

\newpage 
\setcounter{page}{1}

\section{Introduction} \label{sec:intro}

The study of small-depth Boolean circuits is one of the great success stories of complexity theory.  The exponential lower bounds against constant-depth {\sf AND}-{\sf OR}-{\sf NOT} circuits~\cite{yao1985,Hastad86,Raz87,Smol87} remain among our strongest unconditional lower bounds against concrete models of computation, and the techniques developed to prove these results have led to significant advances in computational learning theory~\cite{linmannis93,Mansour:95}, pseudorandomness~\cite{Nisan91,Baz09,Raz09,Bra10}, proof complexity~\cite{PBI93,Ajt94,KPW94}, structural complexity~\cite{yao1985,Hastad86,Cai86}, and even algorithm design~\cite{Wil14,Williams14,AWY15}.

In addition to \emph{worst-case} lower bounds against small-depth circuits, \emph{average-case} lower bounds, or \emph{correlation bounds}, have also received significant attention\ignore{ over the past thirty years}. As one recent example, Impagliazzo, Matthews, Paturi~\cite{IMP12} and H{\aa}stad~\cite{Has14} independently obtained optimal bounds on the correlation of the parity function with small-depth circuits, capping off a long line of work on the problem~\cite{ajtai1983,yao1985,Hastad86,Cai86,babai1987,BeameIS12}. These results establish strong limits on the computational power of constant-depth circuits, showing that their agreement with the parity function can only be an exponentially small fraction better than that of a constant function.

In this paper we will be concerned with average-case complexity \emph{within} the class of small-depth circuits: our goal is to understand the computational power of depth-$d$ circuits relative to those of strictly smaller depth.  Our main result is an \emph{average-case depth hierarchy theorem} for small-depth circuits:

\begin{theorem}
\label{main-theorem}
Let $2 \leq d \leq {\frac  {c\sqrt{\log n}} {\log \log n}}$, where $c>0$ is an absolute constant, and $\BalancedSipser_d$ be the explicit $n$-variable read-once monotone depth-$d$
formula described in Section \ref{sec:sipser}.  Then any circuit $C$ of depth at most $d-1$ and size at most $S = 2^{n^{{\frac 1 {6(d-1)}}}}$ over $\{0,1\}^n$ agrees with $\BalancedSipser_d$ on at most $({\frac 1 2} + n^{-\Omega(1/d)})\cdot 2^n$ inputs.
\end{theorem}

(We actually prove two incomparable lower bounds, each of which implies Theorem~\ref{main-theorem} as a special case.  Roughly speaking, the first of these says that $\BalancedSipser_d$ cannot be approximated by size-$S$, depth-$d$ circuits  which have significantly smaller bottom fan-in than $\BalancedSipser_d$, and the second of these says that $\BalancedSipser_d$ cannot be approximated by size-$S$, depth-$d$ circuits with a different top-level output gate than $\BalancedSipser_d$.)
\ignore{\lnote{This used to say `` --- which has $\AND$ gates at the bottom layer --- cannot be approximated by size-$S$, depth-$d$ circuits which have $\OR$ gates at the bottoms layer.'')}}

Theorem~\ref{main-theorem} is an average-case extension of the worst-case depth hierarchy theorems of Sipser, Yao, and H{\aa}stad~\cite{sipser1983,yao1985,Hastad86}, and answers an open problem of H{\aa}stad~\cite{Hastad86} (which also appears in~\cite{Hastad:86,Hastad:89}). We discuss the background and context for Theorem~\ref{main-theorem} in Section~\ref{sec:previous-work}, and state our two main lower bounds more precisely in Section~\ref{sec:precise}. 

\paragraph{Applications.} We give two applications of our main result, one in structural complexity and the other in the analysis of Boolean functions. First, via a classical connection between small-depth computation and the polynomial hierarchy~\cite{FSS81,sipser1983}, Theorem~\ref{main-theorem} implies that the polynomial hierarchy is infinite relative to a random oracle:

\begin{theorem}
\label{thm:random-oracle}
With probability $1$, a random oracle $A$ satisfies $\Sigma_{d}^{\mathrm{P}, A} \subsetneq \Sigma_{d+1}^{\mathrm{P}, A}$ for all $d\in \N$.
\end{theorem}

This resolves a well-known conjecture in structural complexity, which first appeared in~\cite{Hastad86,Cai86,babai1987} and has subsequently been discussed in a wide range of surveys~\cite{Johnson86,Hem94,ST95,HRZ95,VollmerWagner97,Zoo}, textbooks~\cite{DK00,HO02}, and research papers \cite{Hastad:86,Hastad:89,Tar89,Fortnow99,Aaronson10b}.  (Indeed, the results of~\cite{Hastad86,Cai86,babai1987}, along with much of the pioneering work on lower bounds against small-depth circuits in the 1980's, were largely motivated by the aforementioned connection to the  polynomial hierarchy.) See Section~\ref{sec:application-oracles} for details.

Our second application is a strong negative answer to questions of Kalai, Hatami, and O'Donnell in the analysis of Boolean functions.  Seeking an \emph{approximate converse} to the fundamental results of Linial, Mansour, Nisan~\cite{linmannis93} and Boppana~\cite{Boppana97} on the total influence of small-depth circuits, Kalai asked
whether every Boolean function with total influence $\polylog(n)$ can be approximated by a constant-depth circuit of quasipolynomial size~\cite{KalaiBlog10,Kalaioverflow12,Hatami14}. O'Donnell posed a variant of the same question with a more specific quantitative bound on how the size of the approximating circuit depends on its influence and depth~\cite{ODonnell07open}. As a consequence of Theorem \ref{main-theorem} we obtain the following:

\begin{theorem}
\label{thm:killBKS}
There are functions $d(n)=\omega_n(1)$ and $S(n) = \exp((\log n)^{\omega_n(1)})$ such that there is a monotone $f: \{0,1\}^n \to \{0,1\}$ with total influence $\Inf(f) = O(\log n)$, but any circuit $C$ that
has depth $d(n)$ and agrees with $f$ on at least $({\frac 1 2} + o_n(1)) \cdot 2^n$ inputs in $\{0,1\}^n$ must have size greater than $S(n)$.
\end{theorem}

Theorem~\ref{thm:killBKS} significantly strengthens O'Donnell and Wimmer's counterexample~\cite{OW07} to a conjecture of Benjamini, Kalai, and Schramm~\cite{BKS:99}, and shows that the total influence bound of~\cite{linmannis93,Boppana97} does not admit even a very weak approximate converse. See Section~\ref{sec:killBKS} for details.

\subsection{Previous work}
\label{sec:previous-work}
In this subsection we discuss previous work related to our average-case depth hierarchy theorem. We discuss the background and context for our applications, Theorems~\ref{thm:random-oracle} and~\ref{thm:killBKS}, in Sections~\ref{sec:application-oracles} and~\ref{sec:killBKS} respectively.

Sipser was the first to prove a worst-case depth hierarchy theorem for small-depth circuits~\cite{sipser1983}. He showed that for every $d\in \N$, there exists a Boolean function $F_d : \zo^n\to\zo$ such that $F_d$ is computed by a linear-size depth-$d$ circuit, but any depth-$(d-1)$ circuit computing $F_d$ has size $\Omega(n^{\log^{(3d)}n})$, where $\log^{(i)}n$ denotes the $i$-th iterated logarithm.  The family of functions $\{F_d\}_{d\in \N}$ witnessing this  separation are depth-$d$ read-once monotone formulas with alternating layers of $\AND$  and $\OR$ gates with fan-in $n^{1/d}$ --- these came to be known as the \emph{Sipser functions}.  Following Sipser's work, Yao claimed an improvement of Sipser's lower bound to $\exp(n^{c_d})$ for some constant $c_d > 0$~\cite{yao1985}. Shortly thereafter H{\aa}stad proved a near-optimal separation for (a slight variant of) the Sipser functions:

\begin{theorem} [Depth hierarchy of small-depth circuits~\cite{Hastad86}; see also~\cite{Hastad:86,Hastad:89}]
\label{thm:worst-case-hierarchy}
For every $d \in \N$, there exists a Boolean function $F_d : \zo^n\to\zo$ such that $F_d$ is computed by a linear-size depth-$d$ circuit, but any depth-$(d-1)$ circuit computing $F_d$ has
size $\exp(n^{\Omega(1/d)})$.
\end{theorem}

The parameters of H{\aa}stad's theorem were subsequently refined by Cai, Chen, and H{\aa}stad~\cite{CCH98}, and Segerlind, Buss, and Impagliazzo~\cite{SBI04}. Prior to the work of Yao and H{\aa}stad, Klawe, Paul, Pippenger, and Yannakakis~\cite{KPNY84} proved a depth hierarchy theorem for small-depth \emph{monotone} circuits, showing that for every $d\in \N$, depth-$(d-1)$ \emph{monotone} circuits require size $\exp(\Omega(n^{1/(d-1)}))$ to compute the depth-$d$ Sipser function.  Klawe et al.~also gave an upper bound, showing that every linear-size monotone formula --- in particular, the depth-$d$ Sipser function for all $d\in \N$ --- can be computed by a depth-$k$ monotone formula of size $\exp(O(k\,n^{1/(k-1)}))$ for all $k\in \N$. 
%\footnote{The lower bound of Klawe et al.~is actually stronger than we have stated: they show that any depth-$d$ circuit \emph{with the opposite alternation pattern} from the depth-$d$ Sipser function $F_d$ requires size $\exp(n^{1/(d-1)})$ to compute $F_d$.  Their depth hierarchy theorem is therefore optimal by their matching lower bound.}

To the best of our knowledge, the first progress towards an \emph{average-case} depth hierarchy theorem for small-depth circuits was made by O'Donnell and Wimmer~\cite{OW07}.  They constructed a linear-size depth-$3$ circuit $F$ and proved that any depth-$2$ circuit that approximates $F$ must have size $2^{\Omega(n/\log n)}$:
\begin{theorem}[Theorem 1.9 of~\cite{OW07}]
\label{thm:OW}
For $w \in \N$ and $n := w2^w$, let $\Tribes : \zo^n\to\zo$ be the function computed by a $2^w$-term read-once monotone DNF formula where every term has width exactly $w$. Let $\Tribes^\dagger$ denote its Boolean dual, the function computed by a $2^w$-clause read-once monotone CNF formula where every clause has width exactly $w$, and define
the $2n$-variable function $F: \{0,1\}^{2n} \to \{0,1\}$ as
\[ F(x) = \Tribes(x_1,\ldots,x_n) \vee \Tribes^\dagger(x_{n+1},\ldots,x_{2n}). \]
Then any depth-$2$ circuit $C$ on $2n$ variables that has size $2^{O(n/\log n)}$ agrees with $F$ on at most a $0.99$-fraction of the $2^{2n}$ inputs.  (Note that $F$ is computed by a
linear-size depth-3 circuit.)
 \end{theorem}
Our Theorem~\ref{main-theorem} gives an analogous separation between depth-$d$ and depth-$(d+1)$ for all $d \geq 2$, with $(1/2 - o_n(1))$-inapproximability rather than $0.01$-inapproximability.  The~\cite{OW07} size lower bound of $2^{\Omega(n/\log n)}$ is much larger, in the case $d=2$, than our $\exp({n^{\Omega(1/d)}})$ size bound. However, we recall that achieving a $\exp({\omega(n^{1/(d-1)})})$ lower bound against depth-$d$ circuits for an explicit function, even for worst-case computation, is a well-known and major open problem in complexity theory~(see e.g.~Chapter \S 11 of~\cite{Jukna:12} and \cite{Val83,GW13,viola13}). In particular, an extension of the $2^{\Omega(n/\polylog(n))}$-type lower bound of~\cite{OW07} to depth $3$, even for worst-case computation, would constitute a significant breakthrough.

\subsection{Our main lower bounds}
\label{sec:precise} 
We close this section with precise statements of our two main lower bound results, a discussion of the (near)-optimality of our correlation bounds, and a very high-level overview of our techniques.

\begin{theorem}[First main lower bound] \label{thm:smallbottomfanin}
For $2 \leq d \leq  {\frac {c \sqrt{\log n}}{\log \log n}}$, the $n$-variable $\BalancedSipser_d$ function has the following property: Any depth-$d$ circuit $C:\zo^n\to\zo$ of size at most $S = 2^{n^{{\frac 1 {6(d-1)}}}}$ and bottom fan-in ${\frac {\log n}{10(d-1)}}$ agrees with $\BalancedSipser_d$ on at most $(\frac1{2} + n^{-\Omega(1/d)})\cdot 2^n$  inputs. 
%Let $C : \zo^n \to \zo$ be any depth-$d$ circuit of size $S = 2^{n^{{\frac 1 {6(d-1)}}}}$ and bottom fan-in ${\frac {\log n}{10(d-1)}}$.  Then for a uniform random input $\bX$, we have
%\[ \Pr[\BalancedSipser_d(\bX) \ne C(\bX)] \ge \frac1{2}  - {\frac 1 {n^{\Omega(1/d)}}}.\]
\end{theorem}

\begin{theorem}[Second main lower bound]
\label{thm:alternationpattern}
For $2 \leq d \leq  {\frac {c\sqrt{\log n}}{\log \log n}}$, the $n$-variable $\BalancedSipser_d$ function has the following property:  Any depth-$d$ circuit $C:\zo^n\to\zo$ of size at most $S = 2^{n^{{\frac 1 {6(d-1)}}}}$ and the opposite alternation pattern to $\BalancedSipser_d$ (i.e.~its top-level output gate is $\OR$ if $\BalancedSipser_d$'s is $\AND$ and vice versa) agrees with $\BalancedSipser_d$ on at most $(\frac1{2} + n^{-\Omega(1/d)})\cdot 2^n$  inputs. 
%Let $C : \zo^n \to \zo$ be any depth-$d$ circuit of size $S = 2^{n^{{\frac 1 {6(d-1)}}}}$ and the opposite alternation pattern to $\BalancedSipser_d,$ (i.e.~its top-level output gate is $\OR$ if $\BalancedSipser_d$'s is $\AND$ and vice versa).
%Then for a uniform random input $\bX$, we have
%\[ \Pr[\BalancedSipser_d(\bX) \ne C(\bX)] \ge \frac1{2}  - {\frac 1 {n^{\Omega(1/d)}}}.\]
\end{theorem}

Clearly both these results imply Theorem~\ref{main-theorem} as a special case, since any size-$S$ depth-$(d-1)$ circuit may be viewed as a size-$S$ depth-$d$ circuit satisfying the assumptions of Theorems~\ref{thm:smallbottomfanin} and~\ref{thm:alternationpattern}.

\paragraph{(Near)-optimality of our correlation bounds.}
For constant $d$, our main result shows that the depth-$d$ $\BalancedSipser_d$ function has correlation at most $(1/2 + n^{-\Omega(1)})$ with any subexponential-size circuit of depth $d-1$.  Since $\BalancedSipser_d$ is a monotone function, well-known results \cite{BshoutyTamon:96} imply that its correlation with some input variable $x_i$ or one of the constant functions 0,1 (trivial approximators of depth at most one) must be at least $(1/2 + \Omega(1/n))$; thus significant improvements on our correlation bound cannot be achieved for this (or for any monotone) function.

What about non-monotone functions?  If $\{f_d\}_{d \geq 2}$ is any family of $n$-variable functions computed by poly$(n)$-size, depth-$d$ circuits, the ``discriminator lemma'' of Hajnal et al.~\cite{HMP+:87} implies that $f_d$ must have correlation at least $(1/2 + n^{-O(1)})$ with one of the depth-$(d-1)$ circuits feeding into its topmost gate. Therefore a ``$d$ versus $d-1$'' depth hierarchy theorem for correlation $(1/2 + n^{-\omega(1)})$ does not hold.

\paragraph{Our techniques.}  Our approach is based on \emph{random projections}, a generalization of random restrictions.  At a high level, we design a carefully chosen (adaptively chosen) sequence of random projections, and argue that with high probability under this sequence of random projections, (i) any circuit $C$ of the type specified in Theorem~\ref{thm:smallbottomfanin} or Theorem~\ref{thm:alternationpattern} ``collapses,'' while (ii) the $\BalancedSipser_d$ function ``retains structure,''  and (iii) moreover this happens in such a way as to imply that the circuit $C$ must have originally been a very poor approximator for $\BalancedSipser_d$   (before the random projections).  Each of (i)--(iii) above requires significant work; see Section \ref{sec:techniques} for a much more detailed explanation of our techniques (and of why previous approaches were unable to successfully establish the result).

\section{Application \#1: Random oracles separate the polynomial hierarchy}
\label{sec:application-oracles}

\subsection{Background: $\mathsf{PSPACE} \neq \mathsf{PH}$ relative to a
random oracle}

The pioneering work on lower bounds against small-depth circuits in the 1980's was largely motivated by a connection between small-depth computation and the polynomial hierarchy shown by Furst, Saxe, and Sipser~\cite{FSS81}. They gave a super-polynomial size lower bound for constant-depth circuits, proving that depth-$d$ circuits computing the $n$-variable parity function must have size $\Omega(n^{\log^{(3d-6)}n})$, where $\log^{(i)}n$ denotes the $i$-th iterated logarithm.  They also showed that an improvement of this lower bound to super-quasipolynomial for constant-depth circuits (i.e.~$\Omega_d\big(2^{(\log n)^k}\big)$ for all constants $k$) would yield an oracle $A$ such that $\mathsf{PSPACE}^A \ne \mathsf{PH}^A$.  Ajtai independently proved a stronger lower bound of $n^{\Omega_d(\log n)}$~\cite{ajtai1983}; his motivation came from finite model theory.  Yao gave the first  super-quasipolynomial lower bounds on the size of constant-depth circuits computing the parity function~\cite{yao1985}, and shortly after H{\aa}stad proved the optimal lower bound of $\exp(\Omega(n^{1/(d-1)}))$ via his influential Switching Lemma~\cite{Hastad86}.

Yao's relativized separation of {\sf PSPACE} from {\sf PH} was improved qualitatively by Cai, who showed that the separation holds even relative to a \emph{random} oracle~\cite{Cai86}. Leveraging the connection made by~\cite{FSS81}, Cai accomplished this by proving \emph{correlation bounds} against constant-depth circuits, showing that constant-depth circuits of sub-exponential size agree with the parity function only on a $(1/2+ o_n(1))$ fraction of inputs. (Independent work of Babai~\cite{babai1987} gave a simpler proof of the same relativized separation.)

\subsection{Background: The polynomial hierarchy is infinite relative to some oracle}

Together, these results paint a fairly complete picture of the status of the $\mathsf{PSPACE}$ versus $\mathsf{PH}$ question in relativized worlds: not only does there exist an oracle $A$ such that $\mathsf{PSPACE}^A \ne \mathsf{PH}^A$, this separation holds relative to almost all oracles. A natural next step is to seek analogous results showing that the relativized polynomial hierarchy is infinite; we recall that the polynomial hierarchy being infinite implies $\mathsf{PSPACE} \ne \mathsf{PH}$, and furthermore, this implication relativizes. We begin with the following question, attributed to Albert Meyer in~\cite{BGS75}:

\newtheorem*{MQ}{Meyer's Question}
\begin{MQ}
Is there a relativized world within which the polynomial hierarchy is infinite?  Equivalently, does there exist an oracle $A$ such that $\Sigma_{d}^{\mathrm{P}, A} \subsetneq \Sigma_{d+1}^{\mathrm{P}, A}$ for all $d\in \N$?
\end{MQ}

Early work on Meyer's question predates~\cite{FSS81}. It was first considered by Baker, Gill, and Solovay in their paper introducing the notion of relativization~\cite{BGS75}, in which they prove the existence of an oracle $A$ such that ${\sf P}^A \ne {\sf NP}^A \ne {\sf coNP}^A$, answering Meyer's question in the affirmative for $d \in \{0,1\}$. Subsequent work of Baker and Selman proved the $d=2$ case~\cite{BS79}.  Following~\cite{FSS81}, Sipser noted the analogous connection between Meyer's question and circuit lower bounds~\cite{sipser1983}: to answer Meyer's question in the affirmative, it suffices to exhibit, for every constant $d\in \N$, a Boolean function $F_d$ computable by a depth-$d$ $\acz$ circuit such that any depth-$(d-1)$ circuit computing $F_d$ requires super-quasipolynomial size. (This is a significantly more delicate task than proving super-quasipolynomial size lower bounds for the parity function; see Section~\ref{sec:techniques} for a detailed discussion.)  Sipser also constructed a family of Boolean functions for which he proved an $n$ versus $\Omega(n^{\log^{(3d)}n})$ separation --- these came to be known as the \emph{Sipser functions}, and they play the same central role in Meyer's question as the parity function does in the relativized {\sf PSPACE} versus {\sf PH} problem.

As discussed in the introduction (see Theorem~\ref{thm:worst-case-hierarchy}), H{\aa}stad gave the first proof of a near-optimal $n$ versus $\exp(n^{\Omega(1/d)})$ separation for the Sipser functions~\cite{Hastad86}, obtaining a strong depth hierarchy theorem for small-depth circuits and answering Meyer's question in the affirmative for all $d\in \N$.

\subsection{This work: The polynomial hierarchy is infinite relative to a random oracle}

Given H{\aa}stad's result, a natural goal is to complete our understanding of Meyer's question by showing that the polynomial hierarchy is not just infinite with respect to \emph{some} oracle, but in fact with respect to \emph{almost all} oracles. Indeed, in \cite{Hastad86,Hastad:86,Hastad:89}, H{\aa}stad poses the problem of extending his result
to show this as an open question:

\begin{question}[Meyer's Question for Random Oracles~\cite{Hastad86,Hastad:86,Hastad:89}]
\label{question:random-oracle}
Is the polynomial hierarchy infinite relative to a random oracle?  Equivalently, does a random oracle $A$ satisfy $\Sigma_{d}^{\mathrm{P}, A} \subsetneq \Sigma_{d+1}^{\mathrm{P}, A}$ for all $d\in \N$?\end{question}

Question~\ref{question:random-oracle} also appears as the main open problem in~\cite{Cai86,babai1987}; as mentioned above, an affirmative answer to Question~\ref{question:random-oracle} would imply Cai and Babai's result showing that $\mathsf{PSPACE}^A\ne \mathsf{PH}^A$ relative to a random oracle $A$.  Further motivation for studying Question~\ref{question:random-oracle} comes from a surprising result of Book, who proved that the \emph{unrelativized} polynomial hierarchy collapses if it collapses relative to a random oracle~\cite{book1994}.  Over the years Question~\ref{question:random-oracle} has been discussed in a wide range of surveys~\cite{Johnson86,Hem94,ST95,HRZ95,VollmerWagner97,Zoo}, textbooks~\cite{DK00,HO02}, and research papers \cite{Hastad:86,Hastad:89,Tar89,Fortnow99,Aaronson10b}.

\paragraph{Our work.} As a corollary of our main result (Theorem~\ref{main-theorem}) --- an \emph{average-case} depth hierarchy theorem for small-depth circuits --- we answer Question~\ref{question:random-oracle} in the affirmative for all $d\in \N$:

\begin{reptheorem}{thm:random-oracle}
The polynomial hierarchy is infinite relative to a random oracle: with probability $1$, a random oracle $A$ satisfies $\Sigma_{d}^{\mathrm{P}, A} \subsetneq \Sigma_{d+1}^{\mathrm{P}, A}$ for all $d\in \N$.
\end{reptheorem}

Prior to our work, the $d \in \{0,1\}$ cases were proved by Bennett and Gill in their paper initiating the study of random oracles~\cite{bengil81}. Motivated by the problem of obtaining relativized separations in quantum structural complexity, Aaronson recently showed that a random oracle $A$ separates $\Pi^{\mathrm{P}}_2$ from $\mathsf{P^{NP}}$~\cite{Aaronson10,Aaronson10b}; he conjectures in~\cite{Aaronson10b} that his techniques can be extended to resolve the $d=2$ case of Theorem~\ref{thm:random-oracle}.
We observe that O'Donnell and Wimmer's techniques (Theorem~\ref{thm:OW} in our introduction) can be used to prove the $d=2$ case~\cite{OW07}, though the authors of~\cite{OW07} do not discuss this connection to the relativized polynomial hierarchy in their paper.

\begin{table}[h]
\renewcommand{\arraystretch}{1.6}
\centering
\begin{tabular}{| >{\centering\arraybackslash}p{5cm} |c|c|}
\hline
 &\ \ $\mathsf{PSPACE}^A \ne \mathsf{PH}^A$\ \ & $\Sigma^{\mathrm{P},A}_{d} \subsetneq \Sigma^{\mathrm{P},A}_{d+1}$ for all $d \in \N$ \\ \hline
 Connection to lower bounds for constant-depth circuits  & \cite{FSS81} & \cite{sipser1983} \\ \hline
 Hard function(s) & Parity & Sipser functions \\ \hline
 Relative to \emph{some} oracle $A$ &  \cite{yao1985,Hastad86} &
  \cite{yao1985,Hastad86} \\ \hline
  Relative to \emph{random} oracle $A$ & \cite{Cai86,babai1987} & {\bf This work}  \\ \hline
\end{tabular}
\caption{Previous work and our result on the relativized polynomial hierarchy}
\label{orale-table}
\end{table}

We refer the reader to Chapter \S7 of H{\aa}stad's thesis~\cite{Hastad:86} for a detailed exposition (and complete proofs) of the aforementioned connections between small-depth circuits and the polynomial hierarchy (in particular, for the proof of how Theorem \ref{thm:random-oracle}
follows from Theorem~\ref{main-theorem}).

\section{Application \#2: No approximate converse to Boppana--Linial--Mansour--Nisan}
\label{sec:killBKS}

The famous result of Linial, Mansour, and Nisan gives strong bounds on Fourier concentration of small-depth circuits~\cite{linmannis93}. As a corollary, they derive an upper bound on the total influence of small-depth circuits, showing that depth-$d$ size-$S$ circuits have total influence $(O(\log S))^d$.  (We remind the reader that the total influence of an $n$-variable Boolean function $f$ is $\Inf(f) := \sum_{i=1}^n \Inf_i(f)$, where $\Inf_i(f)$ is the probability that flipping coordinate $i\in [n]$ of a uniform random input from $\{0,1\}^n$ causes the value of $f$ to change.) This was subsequently sharpened by Boppana via a simpler and more direct proof~\cite{Boppana97}:

\begin{theorem}[Boppana, Linial--Mansour--Nisan]
\label{thm:LMN}
Let $f:\zo^n\to\zo$ be a computed by a size-$S$ depth-$d$ circuit. Then $\Inf(f) = (O(\log S))^{d-1}$.
\end{theorem}

(We note that Boppana's bound is asymptotically tight by considering the parity function.)  Several researchers have asked whether an \emph{approximate converse} of some sort holds for Theorem~\ref{thm:LMN}:
\begin{quote}
 \emph{If $f: \{0,1\}^n \to \{0,1\}$ has low total influence, is it the case that $f$
can be approximated to high accuracy by a small constant-depth circuit?}
\end{quote}
A result of this flavor, taken together with Theorem~\ref{thm:LMN}, would yield an elegant characterization of  Boolean functions with low total influence.
In this section we
formulate a very weak approximate converse to Theorem~\ref{thm:LMN} and show, as a
consequence of our main result (Theorem~\ref{main-theorem}), that even this weak converse does not hold.

\subsection{Background: BKS conjecture and O'Donnell--Wimmer's counterexample}

An approximate converse to Theorem
\ref{thm:LMN} was first conjectured by Benjamini, Kalai, and Schramm, with a very specific quantitative
bound on how the size of the approximating circuit depends on its influence and depth~\cite{BKS:99} (the conjecture also appears in the surveys~\cite{Kal00,KalaiSafra:05}). They posed the following:

\newtheorem*{BKS}{Benjamini--Kalai--Schramm (BKS) Conjecture}
\begin{BKS}
For every $\eps>0$ there is a constant $K=K(\eps)$ such that the following holds:  Every monotone $f: \{0,1\}^n \to \{0,1\}$ can be $\eps$-approximated by a depth-$d$ circuit of size at most
\[ \exp\big((K \cdot \Inf(f))^{1/(d-1)}\big) \]
for some $d \ge 2$.
\end{BKS}

(We associate a circuit with the Boolean function that it computes, and we say that a circuit \emph{$\eps$-approximates} a Boolean function $f$ if it agrees with $f$ on all but an $\eps$-fraction of all inputs.)
If true, the BKS conjecture would give a quantitatively strong converse to Theorem~\ref{thm:LMN} for
monotone functions.\footnote{We remark that although the BKS conjecture was stated for monotone Boolean functions, it seems that (a priori) it could have been true for all Boolean functions: prior to~\cite{OW07}, we are not aware of any counterexample to the BKS conjecture even if $f$ is allowed to be non-monotone.}  In addition, it would have important implications for the study of threshold phenomena in Erd\"os--R\'enyi random graphs, which is the context in which Benjamini, Kalai, and Schramm made their conjecture; we refer the reader to~\cite{BKS:99} and Section 1.4 of~\cite{OW07} for a detailed discussion of this connection.  However, the BKS conjecture was disproved by O'Donnell and Wimmer~\cite{OW07}.  Their result (Theorem~\ref{thm:OW} in our introduction) disproves the case $d=2$ of the BKS conjecture, and the case $d>2$ is disproved by an easy argument which~\cite{OW07} give.

\subsection{This work: Disproving a weak variant of the BKS conjecture}

A significantly weaker variant of the BKS conjecture is the following:

\begin{conjecture}
\label{conj:weak-BKS}
For every $\eps > 0$ there is a $d=d(\eps)$ and $K_1=K_1(\eps),K_2=K_2(\eps)$ such that the following holds:
Every monotone $f: \{0,1\}^n \to \{0,1\}$ can be $\eps$-approximated by a depth-$d$ circuit
of size at most \[
\exp\left( (K_1 \cdot \Inf(f))^{K_2}\right).
\]
\end{conjecture}

The~\cite{OW07} counterexample to the BKS conjecture does not disprove Conjecture~\ref{conj:weak-BKS}; indeed, the function $f$
that~\cite{OW07} construct and analyze is computed by a depth-3 circuit of size $O(n)$.\footnote{As with the BKS conjecture, prior to our work we are not aware of any counterexample to Conjecture~\ref{conj:weak-BKS} even if $f$ is allowed to be non-monotone.}  Observe that
Conjecture \ref{conj:weak-BKS}, if true, would yield the following rather appealing consequence:  every monotone $f: \{0,1\}^n \to \{0,1\}$ with total influence at most $\polylog(n)$ can be approximated to any
constant accuracy by a quasipolynomial-size, constant-depth circuit (where both the
constant in the quasipolynomial size bound and the constant depth of the circuit may depend
on the desired accuracy).

Following O'Donnell and Wimmer's disproof of the BKS conjecture, several researchers have posed questions similar in spirit to Conjecture~\ref{conj:weak-BKS}. O'Donnell asked if the BKS conjecture is true if the bound on the size of the approximating circuit is allowed to be $\exp\big((K\cdot \Inf(f))^{1/d}\big)$ instead of
 $\exp\big((K\cdot \Inf(f))^{1/{(d-1)}}\big)$~\cite{ODonnell07open}. This is a weaker statement than the original BKS conjecture (in particular, it is not ruled out by the counterexample of~\cite{OW07}), but still significantly stronger than Conjecture~\ref{conj:weak-BKS}.  Subsequently Kalai asked if Boolean functions with total influence $\polylog(n)$ (resp.~$O(\log n)$) can be approximated by constant-depth circuits of quasipolynomial size (resp.~$\acz$)~\cite{Kalaioverflow12} (see also~\cite{KalaiBlog10} where he states a qualitative version).  Kalai's question is a variant of Conjecture~\ref{conj:weak-BKS} in which $f$ is allowed to be non-monotone, but $\Inf(f)$ is only allowed to be $\polylog(n)$; furthermore, $K_2(\eps)$ is only allowed to be $1$ if $\Inf(f) = O(\log n)$.  Finally, H.~Hatami recently restated the $\Inf(f) = O(\log n)$ case of Kalai's question:

\newtheorem*{Hatami}{Problem 4.6.3 of~\cite{Hatami14}}
\begin{Hatami}
Is it the case that for every $\eps,C>0$, there are
constants $d,k$ such that for every $f: \{0,1\}^n \to \{0,1\}$ with $\Inf(f) \leq C \log n$, there is
a size-$n^k$, depth-$d$ circuit which $\eps$-approximates $f$?
% computing a function $g$ such that $\Pr[f(\bX) \neq g(\bX)] \leq \eps$ where $\bX$
%is uniform over $\{0,1\}^n$?
\end{Hatami}

\paragraph{Our work.} As a corollary of our main result (Theorem~\ref{main-theorem}), we show that Conjecture~\ref{conj:weak-BKS} is false even for (suitable choices of) $\eps = {\frac 1 2} - o_n(1).$ Our counterexample also provides a strong negative answer to O'Donnell's and Kalai--Hatami's versions of Conjecture~\ref{conj:weak-BKS}.  We prove the following:

\begin{reptheorem}{thm:killBKS}
Conjecture \ref{conj:weak-BKS} is false.  More precisely, there is a
monotone $f: \{0,1\}^n \to \{0,1\}$ and a $\delta(n)=o_n(1)$ such that $\Inf(f) =O(\log n)$ but any circuit of
depth $d(n)=\sqrt{\log \log n}$ that agrees with $f$ on $({\frac 1 2} + \delta(n))$ fraction of all inputs must have size at least
$S(n) = 2^{2^{\tilde{\Omega}\big(2^{\sqrt{\log\log n}}\big)}}$.\end{reptheorem}

\begin{proof}[Proof of Theorem~\ref{thm:killBKS} assuming Theorem~\ref{main-theorem}]
Consider the monotone Boolean function $f:\zo^n\to\zo$ corresponding to $\BalancedSipser_d$ of Theorem~\ref{main-theorem} defined over the first $m = 2^{{2^{\lfloor \sqrt{\log \log n}\rfloor}}}$ variables, and of depth $d = \lfloor \log\log m \rfloor + 1 = \lfloor \sqrt{\log\log n}\rfloor + 1$. By Boppana's theorem (Theorem~\ref{thm:LMN}), we have that
\[ \Inf(f) = O(\log m)^{d-1} = O\left(2^{\lfloor \sqrt{\log\log n}\rfloor}\right)^{\lfloor \sqrt{\log\log n}\rfloor } = O(\log n).\]
On the other hand, our main theorem (Theorem~\ref{main-theorem}) implies that even circuits of depth $d-1 = \lfloor\sqrt{\log\log n} \rfloor$ which agree with $f$ on $({\frac 1 2} + \delta(n))$ fraction of all inputs, where $\delta(n) = 2^{-\Omega(2^{\lfloor \sqrt{\log \log n} \rfloor} / \lfloor \sqrt{\log \log n} \rfloor)}$, must have size at least
\begin{equation*} S(n) = 2^{m^{\Omega(1/d)}} = 2^{\big( 2^{2^{\sqrt{\log\log n}}}\big)^{\Omega(1/\sqrt{\log\log n})}} = 2^{2^{\tilde{\Omega}\big(2^{\sqrt{\log\log n}}\big)}}. \qedhere
\end{equation*}
\end{proof}

\section{Our techniques}
\label{sec:techniques}

The method of random restrictions dates back to Subbotovskaya~\cite{Sub61} and continues to be an indispensable technique in circuit complexity. Focusing only on small-depth circuits, we mention that the random restriction method is the common essential ingredient underlying the landmark lower bounds discussed in the previous sections~\cite{FSS81,ajtai1983,sipser1983,yao1985,Hastad86,Cai86,babai1987,IMP12,Has14}.

We begin in Section~\ref{sec:lbrr} by describing the general framework for proving worst- and average-case lower bounds against small-depth circuits via the random restriction method.  Within this framework, we sketch the now-standard proof of correlation bounds for the parity function based on H{\aa}stad's Switching Lemma.  We also recall why the lemma is not well-suited for proving a depth hierarchy theorem for small-depth circuits, hence necessitating the ``blockwise variant'' of the lemma that H{\aa}stad developed and applied to prove his (worst-case) depth hierarchy theorem.  In Section~\ref{sec:our-techniques} we highlight the  difficulties that arise in extending H{\aa}stad's depth hierarchy theorem to the average-case, and how our techniques --- specifically, the notion of random \emph{projections} --- allow us to overcome these difficulties. 

\subsection{Background: Lower bounds via random restrictions} \label{sec:lbrr}

Suppose we would like to show that a \emph{target function} $f : \zo^n\to\zo$ has small correlation with any size-$S$ depth-$d$ \emph{approximating circuit} $C$ under the uniform distribution $\calU$ over $\zo^n$.  A standard approach is to construct a series of random restrictions $\{\calR_k\}_{k \in \{2,\ldots, d\}}$ satisfying three properties:

\begin{itemize}[leftmargin = 0.5cm]
\item {\bf Property 1: Approximator $C$ simplifies.}  The randomly-restricted circuit $C \uhr \brho^{(d)} \cdots \brho^{(2)}$, where $\brho^{(k)} \leftarrow \calR_k$ for $2\le k \le d$, should ``collapse to a simple function'' with high probability. This is typically shown via iterative applications of an appropriate ``Switching Lemma for the $\calR_k$'s\,'', which shows that each random restriction $\brho^{(k)}$ decreases the depth of the circuit $C \uhr \brho^{(d)} \cdots \brho^{(k-1)}$ by one with high probability.  The upshot is that while $C$ is a depth-$d$ size-$S$ circuit, $C \uhr \brho^{(d)} \cdots \brho^{(2)}$ will be a small-depth decision tree, a ``simple function", with high probability.

\item {\bf Property 2: Target $f$ retains structure.}  In contrast with the approximating circuit, the target function $f$ should (roughly speaking) be resilient against the random restrictions $\brho^{(k)} \leftarrow \calR_k$.  While the precise meaning of ``resilient'' depends on the specific application, the key property we need is that $f \uhr \brho^{(d)} \cdots \brho^{(2)}$ will with high probability be a ``well-structured'' function that is uncorrelated with any small-depth decision tree.
\end{itemize}

 Together, these two properties imply that random restrictions of $f$ and $C$ are uncorrelated with high probability.
Note that this already yields \emph{worst-case} lower bounds, showing that $f :\zo^n \to\zo$ cannot be computed exactly by $C$. To obtain correlation bounds, we need to translate such a statement into the fact that $f$ and $C$ \emph{themselves} are uncorrelated.  For this we need the third key property of the random restrictions:

\begin{itemize}[leftmargin=0.5cm]
\item  {\bf Property 3: Composition of $\calR_k$'s completes to $\calU$.}  Evaluating a Boolean function $h : \zo^n \to\zo$ on a random input $\bX \leftarrow \calU$ is equivalent to first applying random restrictions $\brho^{(d)}, \ldots, \brho^{(2)}$ to $h$, and then evaluating the randomly-restricted function $h \uhr \brho^{(d)} \cdots\brho^{(2)}$ on $\bX' \leftarrow \calU$.
\end{itemize}

\paragraph{Correlation bounds for parity.} For uniform-distribution correlation bounds against constant-depth circuits computing the parity function, the random restrictions are all drawn from $\calR(p)$, the ``standard'' random restriction which independently sets each free variable to $0$ with probability $\frac1{2}(1-p)$, to $1$ with probability $\frac1{2}(1-p)$, and keeps it free with probability $p$.  The main technical challenge arises in proving that Property 1 holds --- this is precisely H{\aa}stad's Switching Lemma --- whereas Properties 2 and 3 are straightforward to show. For the second property, we note that
\[ \Parity_n \uhr \rho \equiv \pm\, \Parity(\rho^{-1}(\ast)) \quad \text{for all restrictions $\rho \in \{0,1,\ast\}^n$}, \]
and so $\Parity_n \uhr \brho^{(d)} \cdots \brho^{(2)}$ computes the parity of a random subset $\bS\sse [n]$ of coordinates (or its negation). With an appropriate choice of the $\ast$-probability $p$ we have that $|\bS|$ is large with high probability; recall that $\pm\,\Parity_k$ (the $k$-variable parity function or its negation) has zero correlation with any decision tree of depth at most $k-1$.  For the third property, we note that for all values of $p\in (0,1)$, a random restriction $\brho \leftarrow \calR(p)$ specifies a uniform random subcube of $\zo^n$ (of dimension $|\brho^{-1}(\ast)|$). Therefore, the third property is a consequence of the simple fact that a uniform random point within a uniform random subcube is itself a uniform random point from $\{0,1\}^n$.

\paragraph{H{\aa}stad's blockwise random restrictions.}  With the above framework in mind, we notice a conceptual challenge in proving $\acz$ depth hierarchy theorems via the random restriction method: even focusing only on the worst-case (i.e.~ignoring Property 3), the random restrictions $\calR_k$ will have to satisfy Properties 1 and 2 with the target function $f$ being \emph{computable in $\acz$}.  This is a significantly more delicate task than (say) proving $\Parity \notin \acz$ since, roughly speaking, in the latter case the target function $f \equiv \Parity$ is ``much more complex'' than the circuit $C \in \acz$ to begin with.  In an $\acz$ depth hierarchy theorem, \emph{both} the target $f$ and the approximating circuit $C$ are constant-depth circuits; the target $f$ is ``more complex'' than $C$ in the sense that it has larger circuit depth, but this is offset by the fact that the circuit size of $C$ is allowed to be exponentially larger than that of $f$ (as is the case in both H{\aa}stad's and our theorem).  We refer the reader to Chapter \S6.2 of Hastad's thesis~\cite{Hastad:86} which contains a discussion of this very issue.

H{\aa}stad overcomes this difficulty by replacing the ``standard'' random restrictions $\calR(p)$ with random restrictions \emph{specifically suited to Sipser functions being the target}: his ``blockwise'' random restrictions are designed so that (1) they reduce the depth of the formula computing the Sipser function by one, but otherwise essentially preserve the rest of its structure, and yet (2) a switching lemma still holds for any circuit with sufficiently small bottom fan-in. These correspond to Properties 2 and 1 respectively.  However, unlike $\calR(p)$, H{\aa}stad's blockwise random restrictions are not independent across coordinates and do not satisfy Property 3: their composition does not complete to the uniform distribution $\calU$ (and indeed it does not complete to any product distribution). This is why  H{\aa}stad's construction establishes a worst-case rather than average-case depth hierarchy theorem.
\subsection{Our main technique: Random projections}
\label{sec:our-techniques}

The crux of the difficulty in proving an average-case $\acz$ depth hierarchy theorem therefore lies in designing random restrictions that satisfy Properties 1, 2, and 3 simultaneously, for a target $f$ in $\acz$ and an arbitrary approximating circuit $C$ of smaller depth but possibly exponentially larger size.  To recall, the ``standard'' random restrictions $\calR(p)$ satisfy Properties 1 and 3 but not 2, and H{\aa}stad's blockwise variant satisfies Properties 1 and 2 but not 3.

In this paper we overcome this difficulty with \emph{projections}, a generalization of restrictions.  Given a set of formal variables $\calX = \{ x_1,\ldots,x_n\}$, a restriction $\rho$ either fixes a variable $x_i$  (i.e.~$\rho(x_i) \in \{0,1\}$) or keeps it alive (i.e.~$\rho(x_i) = x_i$, often denoted by $\ast$).  A \emph{projection}, on the other hand, either fixes $x_i$ or maps it to a variable $y_j$ from a possibly different space of formal variables $\calY = \{ y_1,\ldots,y_{n'}\}$.  Restrictions are therefore a special case of projections where $\calY \equiv \calX$, and each $x_i$ can only be fixed or mapped to itself. (See Definition~\ref{def:projection} for precise definitions.)  Our arguments crucially employ projections in which $\calY$ is smaller than $\calX$, and where moreover each $x_i$
is only mapped to a specific element $y_j$ where $j$ depends on $i$ in a carefully designed way that depends on the structure of the formula computing the Sipser function. Such ``collisions'', where blocks of distinct formal variables in $\calX$ are mapped to the same new formal variable $y_i \in \calY$, play a crucial role in our approach.
(We remark that ours is not the first work to consider such a generalization of restrictions. Random projections are also used in the work of Impagliazzo and Segerlind, which establishes lower bounds against constant-depth Frege systems with counting axioms  in proof complexity~\cite{IS01}.) 

At a high level, our overall approach is structured around a sequence $\mathbf{\Psi}$ of (\emph{adaptively chosen}) random projections satisfying Properties 1, 2, and 3 simultaneously, with the target $f$ being $\BalancedSipser$, a slight variant of the Sipser function which we define in Section~\ref{sec:sipser}.   We briefly outline how we establish each of the three properties (it will be more natural for us to prove them in a slightly different order from the way they are listed in Section~\ref{sec:lbrr}):

\begin{itemize}[leftmargin = 0.5cm]

\item {\bf Property 3: $\mathbf{\Psi}$ completes to the uniform distribution.}  Like H{\aa}stad's blockwise random restrictions (and unlike the ``standard'' random restrictions $\calR(p)$), the distributions of our random projections are not independent across coordinates: they are carefully correlated in a way that depends on the structure of the formula computing $\BalancedSipser$.  As discussed above, there is an inherent tension between the need for such correlations on one hand (to ensure that $\BalancedSipser$ ``retains structure''), and the requirement that their composition completes to the uniform distribution on the other hand (to yield uniform-distribution correlation bounds). We overcome this difficulty with our notion of projections: in Section~\ref{sec:complete-to-uniform} we prove that the composition $\mathbf{\Psi}$ of our sequence of random projections completes to the uniform distribution (despite the fact that every one of the individual random projections comprising $\mathbf{\Psi}$ is highly-correlated among coordinates.)

\item {\bf Property 1: Approximator $C$ simplifies.}  Next we prove that approximating circuits $C$ of the types specified in our main lower bounds (Theorems~\ref{thm:smallbottomfanin} and~\ref{thm:alternationpattern}) ``collapse to a simple function'' with high probability under our sequence $\mathbf{\Psi}$ of random projections.   Following the standard ``bottom-up'' approach to proving lower bounds against small-depth circuits, we establish this by arguing that each of the individual random projections comprising $\mathbf{\Psi}$ ``contributes to the simplification'' of $C$ by reducing its depth by (at least) one.  

More precisely, in Section~\ref{sec:approximator-simplifies} we prove a \emph{projection switching lemma}, showing that a small-width DNF or CNF ``switches'' to a small-depth decision tree with high probability under our random projections. (The depth reduction of $C$ follows by applying this lemma to every one of its bottom-level depth-$2$ subcircuits.)  Recall that the random projection of a depth-$2$ circuit over a set of formal variables $\calX$ yields a function over a new set of formal variables $\calY$, and in our case $\calY$ is significantly smaller than $\calX$.  In addition to the structural simplification that results from setting variables to constants (as in H{\aa}stad's Switching Lemma for random \emph{restrictions}), the proof of our projection switching lemma also crucially exploits the additional structural simplification that results from distinct variables in $\calX$ being mapped to the same variable in $\calY$.

\item {\bf Property 2: Target $\BalancedSipser$ retains structure.}  Like H{\aa}stad's blockwise random restrictions, our random projections are defined with the target function $\BalancedSipser$ in mind; in particular, they are carefully designed so as to ensure that $\BalancedSipser$ ``retains structure'' with high probability under their composition $\mathbf{\Psi}$.  

In Section~\ref{sec:typical-stuff} we define the notion of a ``typical'' outcome of our random projections, and prove that with high probability \emph{all} the individual projections comprising $\mathbf{\Psi}$ are typical. (Since our sequence of random projections is chosen adaptively, this requires a careful definition of typicality to facilitate an inductive argument showing that our definition ``bootstraps'' itself.) Next, in Section~\ref{sec:sipser-survives} we show that typical projections have a ``very limited and well-controlled'' effect on the structure of $\BalancedSipser$; equivalently, $\BalancedSipser$ is resilient against typical projections. Together, the results of Section~\ref{sec:typical-stuff} and~\ref{sec:sipser-survives} show that with high probability, $\BalancedSipser$ reduces under $\mathbf{\Psi}$ to a ``well-structured'' formula, in sharp contrast with our results from Section~\ref{sec:approximator-simplifies} showing that the approximator ``collapses to a simple function'' with high probability under $\mathbf{\Psi}$.

\end{itemize}

We remark that the notion of random projections plays a key role in ensuring all three properties above. (We give a more detailed overview of our proof in Section~\ref{sec:second-outline} after setting up the necessary terminology and definitions in the next two sections.)

\section{Preliminaries}

\subsection{Basic mathematical tools}

\begin{fact}[Chernoff bounds]
\label{fact:chernoff}
Let $\bZ_1,\ldots,\bZ_n$ be independent random variables satisfying $0 \le \bZ_i \le 1$ for all $i\in [n]$.  Let $\bS = \bZ_1 + \cdots + \bZ_n$, and $\mu = \E[\bS]$.  Then for all $\gamma\ge 0$,
\begin{align*}
\Pr[\bS \ge (1+\gamma)\mu] &\le \exp\left(-\frac{\gamma^2}{2+\gamma}\cdot \mu\right)  \\
\Pr[\bS \le (1-\gamma)\mu] & \le \exp\left(-\frac{\gamma^2}{2}\cdot \mu\right).
\end{align*}
\end{fact}

We will use the following fact implicitly in many of our calculations:

\begin{fact} \label{fact:totally-standard}
Let $\delta = \delta(n) > 0$ and $n\in \N$, and suppose $\delta n = o_n(1)$. The following inequalities hold for sufficiently large $n$:
\[ 1-\delta n \le (1-\delta)^n \le 1-\lfrac1{2}\delta n.\]
\end{fact}

Finally, the following standard approximations will be useful:

\begin{fact} \label{fact:approx}
For $x \geq 2$, we have
\[
e^{-1} \left(1 - {\frac 1 x} \right) \leq \left(1 - {\frac 1 x} \right)^x \leq e^{-1},
\quad \quad
\text{or equivalently,}
\quad \quad
\left(1 - {\frac 1 x} \right)^x \leq e^{-1} \leq  \left(1 - {\frac 1 x} \right)^{x-1},
\]
and for $0 \leq x \leq 1$, we have $1+x \leq e^x \leq 1+2x.$
\end{fact}

We write $\log$ to denote logarithm base 2 and $\ln$ to denote natural log.

\subsection{Notation}

A DNF is an $\OR$ of $\AND$s (terms) and a CNF is an $\AND$ of $\OR$s (clauses).  The \emph{width} of a DNF (respectively, CNF) is the maximum number of variables that occur in any one of its terms (respectively, clauses). We will assume throughout that our circuits are \emph{alternating}, meaning that every root-to-leaf path alternates between $\AND$ gates and $\OR$ gates, and \emph{layered}, meaning that for every gate $\mathsf{G}$, every root-to-{\sf G} path has the same length.  By a standard conversion, every depth-$d$ circuit is equivalent to a depth-$d$ alternating layered circuit with only a modest increase in size (which is negligible given the slack on our analysis). The size of a circuit is its number of gates, and the depth of a circuit is the length of its longest root-to-leaf path.

For $p \in [0,1]$ and symbols $\bullet,\circ$, we write ``$\{\bullet_p,\circ_{1-p}\}$'' to denote the distribution over $\{\bullet,\circ\}$ which outputs $\bullet$ with probability $p$ and $\circ$ with probability $1-p.$  We write ``$\,\{\bullet_p,\circ_{1-p}\}^k\,$'' to denote the product distribution over $\{\bullet,\circ\}^k$ in which each coordinate is distributed independently according to $\{\bullet_p,\circ_{1-p}\}$.  We write
``\,$\{\bullet_p,\circ_{1-p}\}^k \setminus \{\bullet\}^k$\,'' to denote the product distribution conditioned on not outputting $\{\bullet\}^k$.  %Finally, we use similar notation $\{\alpha_p,\beta_q,\gamma_{1-p-q}\}$ to denote a distribution over $\{\alpha,\beta,\gamma\}$.

Given $\tau \in \{0,1,\ast\}^{A \times [\ell]}$ and $a \in A$, we write
$\tau_a$ to denote the $\ell$-character string $(\tau_{a,i})_{i \in [\ell]} \in \{0,1,\ast\}^{[\ell]}$, and we sometimes refer to this as the ``$a$-th block of $\tau$.''

Throughout the paper we use boldfaced characters such as $\brho$, $\bX$, etc.\ to denote random variables.  We write ``$a = b \pm c$'' as shorthand to denote that $a \in [b-c,b+c]$, and similarly $a \ne b \pm c$ to denote that $a \notin [b-c,b+c]$.  For a positive integer $k$ we write ``$[k]$'' to denote the set $\{1,\dots,k\}.$

\newcommand{\bias}{{\mathrm{bias}}}

The \emph{bias} of a Boolean function $f$ under an input distribution $\bZ$ is defined as
\[
\bias(f,\bZ) := \min \left\{\Prx_{\bZ}[f(\bZ)=0], \Prx_{\bZ}[f(\bZ)=1] \right\}.
\]

\subsection{Restrictions and random restrictions}

\begin{definition}[Restriction]
\label{def:restriction}
A \emph{restriction} $\rho$ of a finite base set $\{x_\alpha\}_{\alpha\in\Omega}$ of Boolean variables is a string $\rho \in \{0,1,\ast\}^{\Omega}$.  (We sometimes equivalently view a restriction $\rho$ as a function $\rho: \Omega \to \{0,1,\ast\}.$)  Given a function $f : \zo^{\Omega} \to \zo$ and restriction $\rho \in \{0,1,\ast\}^\Omega$, the \emph{$\rho$-restriction} of $f$ is the function $(f \uhr \rho) : \zo^{\Omega}\to\zo$ where
\[ (f\uhr \rho)(x) =  f(x\uhr \rho), \quad  \text{and } (x\uhr \rho)_\alpha := \left\{
\begin{array}{cl}
x_\alpha & \text{if $\rho_\alpha = \ast$} \\
\rho_{\alpha} & \text{otherwise}
\end{array}
\right. \ \  \text{for all $\alpha\in \Omega$}.
\]
Given a distribution $\calR$ over restrictions $\{0,1,\ast\}^{\Omega}$ the \emph{$\calR$-random restriction} of $f$ is the random function $f \uhr \brho$ where $\brho\leftarrow \calR$.
\end{definition}

\begin{definition}[Refinement]
Let $\rho,\tau \in \{0,1,\ast\}^{\Omega}$ be two restrictions.  We say that $\tau$ is a
\emph{refinement} of $\rho$ if $\rho^{-1}(1) \subseteq \tau^{-1}(1)$ and $\rho^{-1}(0) \subseteq \tau^{-1}(0)$, i.e. every variable $x_\alpha$ that is set to 0 or 1 by $\rho$ is set in the same way by $\tau$ (and $\tau$ may set additional variables to 0 or 1 that $\rho$ does not set).
\end{definition}

\begin{definition}[Composition]
Let $\rho,\rho'\in \{0,1,\ast\}^{\Omega}$ be two restrictions. Their \emph{composition}, denoted $\rho\rho'\in \{0,1,\ast\}^{\Omega}$, is the restriction defined by
\[ (\rho\rho')_\alpha = \left\{
\begin{array}{cl}
\rho_\alpha & \text{if $\rho_\alpha \in \{0,1\}$} \\
\rho'_\alpha & \text{otherwise.}
\end{array}
\right.\]
Note that $\rho \rho'$ is a refinement of $\rho$.
\end{definition}

\subsection{Projections and random projections}

A key ingredient in this work is the notion of \emph{random projections} which generalize random restrictions. Throughout the paper we will be working with functions over spaces of formal variables that are partitioned into disjoint blocks of some length $\ell$ (see Section~\ref{sec:sipser} for a precise description of these spaces). In other words, our functions will be over spaces of formal variables that can be described as $\calX = \{ x_{a,i} \colon a \in A, i \in [\ell]\}$, where we refer to $x_{a,i}$ as the $i$-th variable in the $a$-th block.  We associate with each such space $\calX$ a smaller space $\calY = \{y_a \colon a \in A \}$ containing a new formal variable for each block of $\calX$.  Given a function $f$ over $\calX$, the \emph{projection} of $f$ yields a function over $\calY$, and the \emph{random projection} of $f$ is the projection of a random restriction of $f$ (which again is a function over $\calY$). Formally, we have the following definition:

\begin{definition}[Projection]
\label{def:projection}
The \emph{projection operator} $\proj$ acts on functions $f : \zo^{A\times [\ell]} \to \zo$ as follows. The \emph{projection of $f$} is the function $(\proj\,f) : \zo^A \to \zo$ defined by
\[
(\proj\,f)(y) = f(x) \quad \text{where $x_{a,i} = y_a$ for all $a\in A$ and $i\in [\ell]$.}
\]
Given a restriction $\rho \in \{0,1,\ast\}^{A\times [\ell]}$, the \emph{$\rho$-projection of $f$} is the function $(\proj_\rho\,f) : \zo^A \to\zo$ defined by
\[
(\proj_\rho\,f)(y) = f(x) \quad \text{where $x_{a,i} =$}
\left\{
\begin{array}{cl}
y_a & \text{if $\rho_{a,i} = \ast$} \\
\rho_{a,i} & \text{otherwise}
\end{array}
\right. \quad \text{for all $a\in A$ and $i \in [\ell]$.}
\]
Equivalently, $(\proj_\rho\,f) \equiv (\proj\,(f\uhr\rho))$. Given a distribution $\calR$ over restrictions in $\{0,1,\ast\}^{A\times [\ell]}$, the associated random projection operator is $\proj_\brho$ where $\brho \leftarrow \calR$, and for $f : \zo^{A\times [\ell]}\to\zo$ we call $\proj_\brho\,f$ its \emph{$\calR$-random projection}. \end{definition}

Note that when $\ell = 1$, the spaces $\calX$ and $\calY$ are identical and our definitions of a $\rho$-projection and $\calR$-random projection coincide exactly with that of a $\rho$-restriction and $\calR$-random restriction in Definition~\ref{def:restriction} (in this case the projection operator $\proj$ is simply the identity operator).

\begin{remark}
\label{rem:equiv-projection}
The following interpretation of the projection operator will be useful for us. Let $f$ be a function over $\calX$, and consider its representation as a circuit $C$ (or decision tree) accessing the formal variables $x_{a,i}$ in $\calX$.  The projection of $f$ is the function computed by the circuit $C'$, where $C'$ is obtained from $C$ by replacing every occurrence of $x_{a,i}$ in $C$ by $y_a$ for all $a\in A$ and $i\in [\ell]$. Note that this may result in a significant simplification of the circuit: for example, an $\AND$ gate ($\OR$ gate, respectively) in $C$ that access both $x_{a,i}$ and $\overline{x}_{a,j}$ for some $a\in A$ and $i,j \in [\ell]$ will access both $y_a$ and $\overline{y}_a$ in $C'$, and therefore can be simplified and replaced by the constant $0$ ($1$, respectively). This is a fact we will exploit in the proof of our projection switching lemma in Section~\ref{sec:psl}.
\end{remark}

\section{The $\BalancedSipser$ function and its basic properties}
\label{sec:sipser}

For $2 \leq d \in \N$, in this subsection we define the depth-$d$ monotone $n$-variable read-once Boolean formula $\BalancedSipser_d$ and establish some of its basic properties.  The $\BalancedSipser_d$ function is very similar to the depth-$d$ formula considered by H\aa stad \cite{Hastad:86}; the only difference is that the fan-ins of the gates in the top and bottom layers have been slightly adjusted, essentially so as to ensure that the formula is very close to balanced between the two output values 0 and 1 (note that such balancedness is a prerequisite for any $(1/2 - o_n(1))$-inapproximability result.)  The $\BalancedSipser_d$ formula is defined in terms of an integer parameter $m$; in all our results this is an asymptotic parameter that approaches $+\infty$, so $m$ should be thought of as ``sufficiently large'' throughout the paper.

 Every leaf of $\BalancedSipser_d$ occurs at the same depth (distance from the root) $d$; there are exactly $n$ leaves ($n$ will be defined below) and each variable occurs at precisely one leaf.  The formula is \emph{alternating}, meaning that every root-to-leaf path alternates between $\AND$ gates and $\OR$ gates; all of the gates that are adjacent to input variables (i.e. the depth-$(d-1)$ gates) are $\AND$ gates, so the root is an $\OR$ gate if $d$ is even and is an $\AND$ gate if $d$ is odd.  The formula is also \emph{depth-regular}, meaning that for each depth (distance from the root) $0 \leq k \leq d-1$, all of the depth-$k$ gates have the same fan-in.  Hence to completely specify the $\BalancedSipser_d$ formula it remains only to specify the fan-in sequence $w_0,\dots,w_{d-1}$, where $w_k$ is the fan-in of every gate at depth $k$.  These fan-ins are as follows:
\begin{itemize}

\item The bottommost fan-in is
\begin{equation}
\label{eq:def-of-wd-1}
w_{d-1}:=m.
\end{equation}
We define
\begin{equation}
\label{eq:def-of-p}
p := 2^{-w_{d-1}}=2^{-m},
\end{equation} and we observe that $p$ is the probability that a depth-$(d-1)$ $\AND$ gate is satisfied by a uniform random choice of $\bX \leftarrow \{0_{1/2},1_{1/2}\}^n$.

\item For each value $1 \leq k \leq d-2$, the value of $w_k$ is $w_k=w$ where
\begin{equation}
\label{eq:def-of-w}
w:= \lfloor m2^m/\log(e) \rfloor.
\end{equation}

\item The value $w_0$ is defined to be
\begin{equation}
\label{eq:def-of-w0}
w_0 := \text{the smallest integer such that~}(1-t_1)^{qw_0} \text{~is at most~}{\frac 1 2},
\end{equation}
where $t_1$ and $q$ will be defined in Section \ref{sec:key-params}, see specifically Equations (\ref{eq:def-of-tk}) and (\ref{eq:def-of-lambda-and-q}). Roughly speaking, $w_0$ is chosen so that the overall formula is essentially balanced under the uniform distribution (i.e.~$\BalancedSipser_d$ satisfies (\ref{eq:sipser-is-balanced}) below); see (\ref{eq:upper lower}) and the discussion thereafter.

\end{itemize}

The number of input variables $n$ for $\BalancedSipser_d$ is $n=\prod_{k=0}^{d-1} w_k = w^{d-2}w_{d-1}w_0$.
The estimates for $t_1$ and $q$ given in (\ref{eq:estimates}) imply that $w_0=2^m \ln(2) \cdot (1 \pm o_m(1))$, so we have that

\begin{equation}
\label{eq:nversusd}
n = {\frac {1 \pm o_m(1)} {\log e}} \cdot \left({\frac {m2^m}{\log e}}\right)^{d-1}  \ignore{\left(\Theta(m2^m)\right)^{d-1}}.\ignore{ \quad \quad \text{~and hence~}\quad \quad (d-1) \cdot \Theta(m) = \log n.}
\end{equation}

We note that for the range of values $2 \leq d \leq {\frac {c \sqrt{\log n}}{\log \log n}}$ that we consider in this paper, a direct (but somewhat tedious) analysis implies that the $\BalancedSipser_d$ function is indeed essentially balanced, or more precisely, that it satisfies
\begin{equation}
 \Prx_{\bX \leftarrow \{0_{1,2},1_{1,2}\}^n}[\BalancedSipser_d(\bX)=1] =
{\frac 1 2} \pm o_n(1).\label{eq:sipser-is-balanced}
\end{equation}
However, since this fact is a direct byproduct of our main theorem (which shows that $\BalancedSipser_d$ cannot be $(1/2 - o_n(1))$-approximated  by any depth-$(d-1)$ formula, let alone by a constant function), we omit the tedious direct analysis here.

We specify an addressing scheme for the gates and input variables of
our $\BalancedSipser_d$ formula which will be heavily used throughout the paper.   Let $A_0 = \{\mathsf{output}\}$, and for $1 \le k \le d$, let $A_{k} = A_{k-1} \times [w_{k-1}]$.  An element of $A_k$ specifies the address of a gate at depth (distance from the output node) $k$ in $\BalancedSipser_d$ in the obvious way; so $A_{d} = \{{\mathsf{output}}\} \times [w_0] \times
 \cdots \times [w_{d-1}]$ is the set of addresses of the input variables and $|A_d| = n$.

 We close this section by introducing notation for the following family of formulas related to $\BalancedSipser_d$:

\begin{definition}
\label{def:truncate-sipser}
For $1 \le k \le d$, we write $\BalancedSipser_d^{(k)} : \zo^{A_k} \to \zo$ to denote the depth-$k$ formula obtained from $\BalancedSipser_d$ by discarding all gates at depths $k+1$ through $d-1$, and replacing every depth-$k$ gate at address $a\in A_k$ with a fresh formal variable $y_a$.
\end{definition}

Note that $\BalancedSipser^{(1)}_d$ is the top gate of $\BalancedSipser_d$; in particular, $\BalancedSipser^{(1)}_d$ is  an $w_0$-way $\OR$ if $d$ is even, and an $w_0$-way $\AND$ if $d$ is odd.  Note also that $\BalancedSipser_d^{(d)}$ is simply $\BalancedSipser_d$ itself, although we stress that $\BalancedSipser^{(k)}_d$  is not the same as $\BalancedSipser_k$ for $1\le k \le d-1$.

\section{Setup for and overview of our proof}
\label{sec:setup} 

\subsection{Key parameter settings} \label{sec:key-params}

The starting point for our parameter settings is the pair of fixed values
\begin{equation} \lambda := \frac{(\log w)^{3/2}}{w^{5/4}}  \quad \text{and}\quad q := \sqrt{p} = 2^{-m/2} .\label{eq:def-of-lambda-and-q}
\end{equation}
Given these fixed values of $\lambda$ and $q$, we define a sequence of parameters $t_{d-1},\dots,t_1$ as
\begin{equation} \label{eq:def-of-tk}
t_{d-1} := {\frac {p-\lambda} q}, \quad \quad \quad
t_{k-1} := {\frac {(1-t_k)^{q w} - \lambda} q} \quad \text{for~}k=d-1,\dots,2.
\end{equation}

Each of our $d-1$ random projections will be defined with respect to an underlying product distribution. Our first random projection $\proj_{\brho^{(d)}}$ will be associated with the uniform distribution over $\zo^n$; this is because our ultimate goal is to establish uniform-distribution correlation bounds. For $k \in \{2,\ldots,d-1\}$ the subsequent random projections $\proj_{\brho^{(k)}}$ will be associated with either the $t_k$-biased or $(1-t_k)$-biased product distribution (depending on whether $d-k$ is even or odd).  Recalling our discussion in Section~\ref{sec:techniques} of the framework for proving correlation bounds --- in particular, the three key properties our random projections have to satisfy --- the values for $t_1,\ldots,t_{d-1}$ are chosen carefully so that the compositions of our $d-1$ random projections complete to the uniform distribution, satisfying Property 3 (we prove this in Section~\ref{sec:complete-to-uniform}).
%; in particular, see Lemma~\ref{lem:OW-trick-initial} for a justification of our choice of $t_{d-1}$, and Lemma~\ref{lem:OW-trick-subsequent} for a justification for our inductive definition of $t_{k-1}$ in terms $t_k$. }

The next lemma gives bounds on $t_{d-1},\ldots,t_1$ which show that these values ``stay under control''. By our definitions of $\lambda,p$ and $q$ in (\ref{eq:def-of-lambda-and-q}), we have that $t_{d-1} = q - o(q)$, and we will need the fact that the values of $t_k$ for $k = d-1,\ldots,2$ remain in the range $q \pm o(q)$. Roughly speaking, since each $t_{k-1}$ is defined inductively in terms of $t_k$ from $k=d-1$ down to $1$, we have to argue that these values do not ``drift'' significantly from the initial value of $t_{d-1} = q - o(q)$.   We need to keep these values under control for two reasons:  first, the magnitude of these values directly affects the strength of our Projection Switching Lemma --- as we will see in Section~\ref{sec:psl}, our error bounds depend on the magnitude of these $t_k$'s. Second, since the top fan-in $w_0$ of our $\BalancedSipser_d$ function is directly determined by $t_1$ (recall (\ref{eq:def-of-w0})), we need a bound on $t_1$ to control the structure of this function.

\begin{lemma} \label{lemma:tk-bound}
There is a universal constant $c>0$ such that for $2 \leq d \leq {\frac {cm}{\log m}}$, we have that $t_k = q \pm q^{1.1}$ for all $k \in [d-1]$.
\end{lemma}

We defer the proof of Lemma~\ref{lemma:tk-bound} to Appendix~\ref{ap:tk-bound}.  The $k=1$ case of Lemma~\ref{lemma:tk-bound} along with our definition of $w_0$ (recall (\ref{eq:def-of-w0})) give us the bounds
\begin{equation} \frac1{2} \ge    (1-t_1)^{qw_0} \ge \frac1{2} \left(1- tq\right) = \frac1{2} \left(1-\frac{\Theta(\log w)}{w}\right) = \frac1{2} \left(1-\Theta(2^{-m})\right). \label{eq:upper lower}
\end{equation}
 These bounds (showing that $(1-t_1)^{qw_0}$ is very close to $1/2$) will be useful for our proof in Section~\ref{sec:sipser-survives} that $\BalancedSipser_d$ remains essentially unbiased (i.e.~it remains ``structured'') under our random projections, which in turn implies our claim (\ref{eq:sipser-is-balanced}) that $\BalancedSipser_d$ is essentially balanced (see Remark~\ref{rem:sipser-is-balanced}).

We close this subsection with the following estimates of our key parameters in terms of $w$ for later reference:

\begin{equation} p = \Theta\left(\frac{\log w}{w}\right), \quad q = \Theta\left(\sqrt{\frac{\log w}{w}}\right), \quad t_k = \Theta\left(\sqrt{\frac{\log w}{w}}\right) \quad \text{for all $k\in [d-1]$.} \label{eq:estimates}
\end{equation}

\subsection{The initial and subsequent random projections}
\label{sec:initial-projection}

As described in Section~\ref{sec:techniques}, our overall approach is structured around a sequence of random projections which we will apply to both the target function $\BalancedSipser_d$ and the approximating circuit~$C$.  Both are functions over $\zo^n \equiv \zo^{A_d}$, and  our $d-1$ random projections will sequentially transform them from being over $\zo^{A_k}$ to being over $\zo^{A_{k-1}}$ for $k = d$ down to $k=1$. Thus, at the end of the overall process both the randomly projected target and the randomly projected approximator are functions over $\zo^{A_1} \equiv \zo^{w_0}$.

We now formally define this sequence of random projections; recalling Definition~\ref{def:projection}, to define a random projection operator it suffices to specify a distribution over random restrictions, and this is what we will do.  We begin with the initial random projection:

\begin{definition}[Initial random projection]
\label{def:calR-init}
The distribution $\calR_{\init}$ over restrictions $\rho$ in $\{0,1,\ast\}^{A_{d-1}\times [m]} \equiv \{0,1,\ast\}^{n}$ (recall that $w_{d-1} = m$) is defined as follows: independently for each $a \in A_{d-1}$,
\begin{equation}
\brho_b \leftarrow \left\{
\begin{array}{ll}
\{1\}^{m} & \text{with probability $\lambda$} \\
\{\ast_{1/2}, 1_{1/2} \}^{m} \setminus \{ 1\}^{m}  & \text{with probability $q$} \\
\{ 0_{1/2}, 1_{1/2} \}^{m} \setminus \{ 1\}^{m} & \text{with probability $1-\lambda-q$}.
\end{array}
\right.
\label{eq:initial}
\end{equation}
\end{definition}

\begin{remark}
The description of $\calR_\init$ given in Definition~\ref{def:calR-init} will be most convenient for our arguments, but we note here the following equivalent view of an $\calR_\init$-random projection.  Let $\calR'_\init$ be the distribution over restrictions $\rho'$ in $\{0,1,\ast\}^{A_{d-1} \times [m]} \equiv \{0,1,\ast\}^n$ where
\[ \brho'_a \leftarrow \{ \ast_{1/2}, 1_{1/2}\}^m \setminus \{1\}^m  \quad \text{independently for each $a \in A_{d-1}$,}  \]
and $\calR''_\init$ be the distribution of restrictions $\rho''$ in $\{0,1,\ast\}^{A_{d-1}}$ where
\[ \brho''_a \leftarrow \left\{
\begin{array}{cl}
1 & \text{with probability $\lambda$} \\
\ast & \text{with probability $q$} \\
0 & \text{with probability $1-\lambda-q$}
\end{array}
\right. \quad
\text{independently for each $a \in A_{d-1}$}.
\]
%\[ \brho''_a \leftarrow \{ 1_\lambda, \ast_q, 0_{1-\lambda-q} \} \quad \text{independently for each $a\in A_{d-1}$.} \]
Then for all $f : \zo^n \to\zo$ we have that $\proj_\brho\, f$, where $\brho\leftarrow\calR_\init$, is distributed identically to
\[ (\proj_{\brho'}\,f) \uhr \brho'' \quad \text{where $\brho' \leftarrow\calR'_\init$ and $\brho''\leftarrow\calR''_\init$}. \]
\end{remark}

\subsubsection{Subsequent random projections}
\label{sec:subsequent-projections}
Our subsequent random projections will alternate between two types, depending on whether $d-k$ is even or odd. These types are dual to each other in the sense that their distributions are completely identical, except with the roles of $1$ and $0$ swapped; in other words, the bitwise complement of a draw from the first type yields a draw from the second type. To avoid redundancy in our definitions we introduce the notation in Table~\ref{table:circ-bullet}: we represent $\zo^{A_k}$ as $\{\bullet,\circ\}^{A_k}$, where a $\circ$-value corresponds to either $1$ or $0$ depending on whether $d-k$ is even or odd, and the $\bullet$-value is simply the complement of the $\circ$-value. For example, the string $(\circ,\circ,\bullet,\circ)$ translates to $(1,1,0,1)$ if $d-k$ is even, and $(0,0,1,0)$ if $d-k$ is odd.

\begin{table}[h]
\renewcommand{\arraystretch}{1.6}
\centering
\begin{tabular}{|c|c| >{\centering\arraybackslash}m{0.5in} | >{\centering\arraybackslash}m{.5in} |}
\hline
& \text{Gates of $\BalancedSipser_d$ at depth $k-1$}& $\circ$ & $\bullet$ \\
\hline
$d-k \equiv 0 \mod 2$ &  $\AND$  &
$1$  & $0$ \\\hline
 $d-k \equiv 1 \mod 2$ & $\OR$ & $0$ & $1$ \\ \hline
\end{tabular}
\caption{Conversion table for $\tau \in \{\bullet,\circ,\ast\}^{A_k}$ where $1\le k \le d$.}
\label{table:circ-bullet}
\end{table}

In an interesting contrast with H{\aa}stad's proofs of the worst-case depth hierarchy theorem (Theorem~\ref{thm:worst-case-hierarchy}) and of $\Parity\notin \acz$, our stage-wise random projection process is \emph{adaptive}: apart from the initial $\calR_\init$-random projection, the distribution of each random projection depends on the outcome of the previous.  We will need the following notion of the ``lift'' of a restriction to describe this dependence:

\begin{definition}[Lift]
\label{def:lift}
Let $2\le k \le d$ and $\tau \in \{\bullet,\circ,\ast\}^{A_{k-1} \times [w_{k-1}]} \equiv \{\bullet,\circ,\ast\}^{A_k}$.  The \emph{lift of $\tau$} is the string $\hat{\tau}\in \{\bullet,\circ,\ast\}^{A_{k-1}}$ defined as follows:  for each $a \in A_{k-1}$,
the coordinate $\hat{\tau}_a$ of $\hat{\tau}$ is
\[
\hat{\tau}_a =
\begin{cases}
\circ & \text{if~} \tau_{a,i}=\bullet \text{~for any~}i \in [w_{k-1}]\\
\bullet & \text{if~} \tau_a = \{\circ\}^{w_{k-1}}\\
\ast & \text{if~} \tau_a \in \{\ast,\circ\}^{w_{k-1}} \setminus \{\circ\}^{w_{k-1}}.
\end{cases}
\]
We remind the reader that $\tau \in \{\bullet, \circ,\ast\}^{A_k}$ and $\hat\tau \in \{\bullet,\circ,\ast\}^{A_{k-1}}$ belong to adjacent levels (i.e.~they fall under different rows in Table~\ref{table:circ-bullet}).  Consequently, for example, if $1$ corresponds to $\bullet$ as a symbol in $\tau$ then it corresponds to $\circ$ as a symbol in $\hat\tau$, and vice versa.
\end{definition}

Later this notion of the ``lift'' of a restriction will also be handy when we describe the effect of our random projections on the target function $\BalancedSipser_d$. The high-level rationale behind it is that $\hat\tau \in \{ \bullet,\circ,\ast\}^{A_{k-1}}$ denotes the values that the bottom-layer gates of $\BalancedSipser^{(k)}_d$ take on when its input variables are set according to $\tau \in \{ \bullet,\circ,\ast\}^{A_k}$.  As a concrete example, suppose $d - k \equiv 0 \mod 2$ and let $\tau \in \{0,1,\ast\}^{A_k}$ be a restriction. Since $d-k \equiv 0 \mod 2$, recalling Table~\ref{table:circ-bullet} we have that the bottom-layer gates of $\BalancedSipser^{(k)}_d$ (or equivalently, the gates of $\BalancedSipser_d$ at depth $k-1$) are $\AND$ gates. For every block $a\in A_{k-1}$,
\begin{itemize}
\item If $\tau_{a,i} = 0$ for some $i\in [w_{k-1}]$, the $\AND$ gate at address $a$ is falsified and has value $0$.
\item If $\tau_{a,i} = \{1\}^{w_{k-1}}$, the $\AND$ gate at address $a$ is satisfied and has value $1$.
\item If $\tau_a \in \{\ast, 1\} \setminus \{1\}^{w_{k-1}}$, the value of the $\AND$ gate at address $a$ remains undetermined (which we denote as having value $\ast$).
\end{itemize}
These three cases correspond exactly to the three branches in Definition~\ref{def:lift}, and so indeed $\hat\tau_a \in \{0,1,\ast\}$ represents the value that the $\AND$ gate at address $a$ takes  when its input variables are set according to $\tau_a \in \{0,1,\ast\}^{w_{k-1}}$.

We shall require the following technical definition:

\begin{definition}[$k$-acceptable]
\label{def:acceptable}
For $2\le k \le d-1$ and a set $S\sse [w_{k-1}]$, we say that $S$ is \emph{$k$-acceptable} if
\[ |S| = qw \pm w^{\gradual(k,d)}, \quad \text{where $\gradual(k,d) := \frac1{3} + \frac{d-k-1}{12d}$}. \]
Note that $\frac1{3} \le \gradual(k,d) \le \frac{5}{12} < \frac1{2}$ for all $d\in \N$ and $2\le k \le d-1$.
\end{definition}

For intuition, in the above definition $S$ should be thought of as specifying those children of a particular depth-$(k-1)$ gate of $\BalancedSipser_d$ that take the value $\ast$ under certain restrictions (defined below).  We want the size of this set to be essentially $qw$, and as $k$ gets smaller (closer to the root), for technical reasons we allow more and more --- but never too much ---  deviation from this desired value. See Section~\ref{sec:typical-stuff} for a detailed discussion.

We are now ready to give the key definition for our subsequent random projections:

\begin{definition}[Subsequent random projections]
\label{def:main-projection}
Let $\tau \in \{\bullet,\circ,\ast\}^{A_k}$ where $2\le k \le d-1$. We define a distribution $\calR(\tau)$ over refinements $\brho \in \{\bullet,\circ,\ast\}^{A_k}$ of $\tau$ as follows. Independently for each $a \in A_{k-1}$, writing $S_a = S_a(\tau)$ to denote $\tau_a^{-1}(\ast) = \{ i \in [w_{k-1}] \colon \tau_{a,i} = \ast\}$ and $\brho(S_a)$ to denote the substring of $\brho_a$ with coordinates in $S_a$,
\begin{itemize}
\item If $\hat{\tau}_a = \circ$ (i.e.~if $\tau_{a,i} = \bullet$ for some $i\in [w_{k-1}]$) or if $S_a$ is not $k$-acceptable\ignore{\lnote{Check that this change does not affect typicality implies typicality (it shouldn't)}}, then
\[  \brho(S_a) \leftarrow \{ \bullet_{t_k}, \circ_{1-t_k}\}^{S_a}.\]
\item If $\hat{\tau}_a = \ast$ (i.e.~if $\tau_{a,i} \in \{\ast,\circ\}^{w_{k-1}} \setminus \{\circ\}^{w_{k-1}}$) and $S_a$ is $k$-acceptable, then
\begin{equation} \brho(S_a) \leftarrow \left\{
\begin{array}{ll}
\circ^{S_a} & \text{with probability $\lambda$} \\
\{ \ast_{t_k}, \circ_{1-t_k}\}^{S_a} \setminus \{ \circ\}^{S_a} & \text{with probability $q_a$}   \\
\{ \bullet_{t_k}, \circ_{1-t_k}\}^{S_a} \setminus \{ \circ\}^{S_a} & \text{with probability $1-\lambda - q_a$},
\end{array}
\right.\label{eq:OW-restriction}
\end{equation}
where
\begin{equation} q_a := \frac{(1-t_k)^{|S_a|} -\lambda}{t_{k-1}} \ \text{\ is chosen to satisfy } (1-t_k)^{|S_a|} = \lambda + q_a t_{k-1}. \label{eq:def-of-qa}
\end{equation}
\end{itemize}
(Note that if $\hat{\tau}_a = \bullet$ then $\tau_{a,i} = \circ$ for all $i\in [w_{k-1}]$, and so $\tau_a$ cannot be refined further.)

For all $a\in A_{k-1}$ and $i\in [w_{k-1}]$ such that $\tau_{a,i} \in \{\bullet,\circ\}$, we set $\brho_{a,i} = \tau_{a,i}$ and so $\brho$ is indeed a refinement of~$\tau$.
\end{definition}

\begin{remark}
We remark that $q_a$ as defined in (\ref{eq:def-of-qa}) is indeed a well-defined quantity in $[0,1]$ if $S_a$ is $k$-acceptable. We omit the straightforward verification here since our analysis in Section~\ref{sec:typical-stuff} will in fact establish a stronger statement showing that $q_a = q \pm o(q)$; see Lemma~\ref{lem:bound-on-qa}. \end{remark}

\begin{remark}
\label{rem:complete-dead-blocks}
By inspecting Definition~\ref{def:calR-init}, we see that for all $\rho\in \supp(\calR_\init)$ and blocks $a\in A_{d-1}$
\begin{align*} \rho_{a,i} &= \ast \text{ for some $i\in [m]$} \quad \text{iff} \quad \rho_a \in \{\ast,1 \}^{m} \setminus \{1\}^m, \quad \text{or equivalently,} \\
\rho_{a,i} &= \ast \text{ for some $i\in [m]$} \quad \text{iff} \quad \hat\rho_a = \ast,
\end{align*}
and hence for all $h : \zo^n\to\zo$ the projection $\proj_\rho\,h : \zo^{A_{d-1}}\to \zo$ depends only on the coordinates in $(\hat\rho)^{-1}(\ast)\sse A_{d-1}$.  Likewise, by inspecting Definition~\ref{def:main-projection} we have that for all $\tau\in \{\bullet,\circ,\ast\}^{A_k}, \rho\in \supp(\calR(\tau))$, and blocks $a\in A_{k-1}$,
\begin{align*}
\rho_{a,i} &= \ast \text{ for some $i\in [w_{k-1}]$} \quad \text{iff} \quad \rho_a \in \{\ast,\circ \}^{w_{k-1}} \setminus \{\circ\}^{w_{k-1}}, \quad \text{or equivalently,} \\
\rho_{a,i} &= \ast \text{ for some $i\in [w_{k-1}]$} \quad \text{iff} \quad \hat\rho_a = \ast,
\end{align*}
and hence for all $h : \zo^{A_k} \to \zo$ the projection $\proj_\rho h : \zo^{A_{k-1}} \to \zo$ depends only on the coordinates in $(\hat\rho)^{-1}(\ast)\sse A_{k-1}$. Our proof that our sequence of random projections (based on Definitions~\ref{def:calR-init} and~\ref{def:main-projection} as described in Definition \ref{def:projection}) completes to the uniform distribution will rely on these properties; see Section~\ref{sec:complete-to-uniform}.
\end{remark}

\subsection{Overview of our proof}
\label{sec:second-outline}

With the definitions from Section \ref{sec:initial-projection} in hand, we are (finally) in a position to give a detailed overview of our proof.  Let $C$ be a depth-$d$ approximating circuit for $\BalancedSipser_d$, where $C$ either has significantly smaller bottom fan-in than $\BalancedSipser_d$ (in the case of Theorem~\ref{thm:smallbottomfanin}) or the opposite alternation pattern to $\BalancedSipser_d$ (in the case of Theorem~\ref{thm:alternationpattern}), and $C$ satisfies the size bounds given in the respective theorem statements.  In both cases our goal is to show that $C$ has small correlation with $\BalancedSipser_d$, i.e. to  prove that
\begin{equation} \Pr[\BalancedSipser_d(\bX) \ne C(\bX)] \geq \frac1{2} - o_n(1) \label{eq:goal-correlation}
\end{equation}
for a uniform random input $\bX \leftarrow \{0_{1/2},1_{1/2}\}^n$.  At a high level, we do this by analyzing the effect of  $d-1$ random projections on the target and the approximator: we begin with an $\calR_\init$-random projection $\proj_{\brho^{(d)}}$ where $\brho^{(d)}\leftarrow\calR_\init$, followed by $\proj_{\brho^{(d-1)}}$ where $\brho^{(d-1)}\leftarrow\calR(\hat{\brho^{(d)}})$, and then $\proj_{\brho^{(d-2)}}$ where $\brho^{(d-2)}\leftarrow \calR(\hat{\brho^{(d-1)}})$, and so on. It is interesting to note that unlike H{\aa}stad's proofs of the worst-case depth hierarchy theorem (Theorem~\ref{thm:worst-case-hierarchy}) and of $\Parity\notin \acz$, the distribution of our $k$-th random projection is defined \emph{adaptively} depending on the outcome of the $(k-1)$-st. For notational concision we introduce the following definition for this overall $(d-1)$-stage projection:

\begin{definition} \label{def:boldpsi}
Given a function $f : \zo^{n}\to\zo$, we write $\mathbf{\Psi}(f) : \zo^{w_0} \to \zo$ to denote the following random projection of $f$:
\[ \mathbf{\Psi}(f) \equiv \proj_{\brho^{(2)}}\,\proj_{\brho^{(3)}}\cdots \proj_{\brho^{(d-1)}}\,\proj_{\brho^{(d)}}\, f, \]
where   $\brho^{(d)}\leftarrow\calR_{\init}$ and $\brho^{(k)}\leftarrow\calR(\hat{\brho^{(k+1)}})$ for all $2 \le k\le d-1$. We will sometimes refer to the overall process as a $\mathbf{\Psi}$-random projection, and $\mathbf{\Psi}(f)$ as the $\mathbf{\Psi}$-random projection of $f$. (We remind the reader that the projection of a function over $\zo^{A_k}$ yields a function over $\zo^{A_{k-1}}$ for all $2\le k\le d$, and in particular $\mathbf{\Psi}(f)$ is indeed a function over $\zo^{A_1} \equiv \zo^{w_0}$.)
\end{definition}

Recalling the framework for proving correlation bounds discussed in Section~\ref{sec:techniques}, the rest of the paper is structured around showing that a $\mathbf{\Psi}$-random projection satisfies the three key properties outlined in Section~\ref{sec:techniques}:
\begin{description}
\item[Property 1.] The approximating circuit $C$ simplifies under a $\mathbf{\Psi}$-random projection.
\item[Property 2.] The target $\BalancedSipser_d$ remains structured under a $\mathbf{\Psi}$-random projection.
\item[Property 3.] $\mathbf{\Psi}$ completes to the uniform distribution.
\end{description}

\paragraph{Section~\ref{sec:complete-to-uniform}.}
We begin in Section~\ref{sec:complete-to-uniform} with Property 3. We show that
\begin{equation} \Pr[\BalancedSipser_d(\bX) \ne C(\bX)] = \Pr[(\mathbf{\Psi}(\BalancedSipser_d))(\bY) \ne (\mathbf{\Psi}(C))(\bY)] \label{eq:reduction}
\end{equation}
where $\bY$ is drawn from an appropriate product distribution $\calD$ over $\zo^{w_0}$ ($\calD$ is the $t_{1}$-biased product distribution if $d$ is even, and $(1-t_{1})$-biased product distribution if $d$ is odd).  This reduces our goal of bounding the correlation between $\BalancedSipser_d$ and $C$ (i.e.~(\ref{eq:goal-correlation})) under the uniform distribution, to the task of bounding the correlation between their $\mathbf{\Psi}$-random projections $\mathbf{\Psi}(\BalancedSipser_d)$ and $\mathbf{\Psi}(C)$ with respect to $\calD$.

\paragraph{Section~\ref{sec:approximator-simplifies}.} With the reduction (\ref{eq:reduction}) in hand, we turn our attention to Property 1, showing that the approximating circuit $C$ of the type specified in either Theorems~\ref{thm:smallbottomfanin} or~\ref{thm:alternationpattern} ``collapses to a simple function'' under a $\mathbf{\Psi}$-random projection. More precisely, for the case that the depth-$d$ circuit $C$ has significantly smaller bottom fan-in than $\BalancedSipser_d$ we show that $C$ collapses to a shallow decision tree, and for the case that $C$ has the opposite alternation pattern to $\BalancedSipser_d$ we show that $C$ collapses to a small-width depth-two circuit with top gate opposite to that of $\mathbf{\Psi}(\BalancedSipser_d)$. (In both cases these statements are with high probability under a $\mathbf{\Psi}$-random projection.)

In close parallel with H{\aa}stad's ``bottom-up'' proof of $\Parity \notin \acz$, the main technical ingredient in this section is a \emph{projection switching lemma} showing that the random projection $\proj_{\brho^{(k)}}$ of a small-width DNF or CNF ``switches'' to a small-depth decision tree with high probability. Applying this lemma to every bottom-level depth-$2$ subcircuit of $C$, we are able to argue that each of the $d-1$ random projections comprising $\mathbf{\Psi}$ reduces the depth of $C$ by one with high probability, and thus $\mathbf{\Psi}(C)$ collapses to a small-depth decision tree or small-width depth-two circuit as claimed.

\paragraph{Section~\ref{sec:combo}.}

It remains to argue that the target $\BalancedSipser_d$ --- in contrast with the approximating circuit $C$ ---  ``retains structure'' with high probability under a $\mathbf{\Psi}$-random restriction.  This is a high-probability statement because there is a nonzero failure probability introduced by each of the $d-1$ individual random projections $\proj_{\brho^{(k)}}$ that comprise $\mathbf{\Psi} \equiv \{ \brho^{(k)}\}_{k \in \{2,\ldots,d\}}$ (see Footnote \ref{footnote:sec9} for an example of a possible ``failure event'' for one of these restrictions).  To reason about and bound these failure probabilities, in Section~\ref{sec:typical-stuff} we introduce the notion of a ``typical'' restriction.  The parameters of our definition of typicality are chosen carefully to ensure that
\begin{enumerate}
\itemsep -.5pt
\item[(i)] $\brho^{(d)}\leftarrow\calR_\init$ is typical with  high probability, and
\item[(ii)] if $\rho^{(k+1)}$ is typical, then $\brho^{(k)} \leftarrow\calR(\hat{\rho^{(k+1)}})$ is also typical with  high probability.
\end{enumerate}
We establish (i) and (ii) in Section \ref{sec:typical-stuff}.  Together, (i) and (ii) imply that with high probability $\mathbf{\Psi} \equiv \{ \brho^{(k)}\}_{k\in \{2,\ldots,d\}}$ is such that $\brho^{(d)},\ldots,\brho^{(2)}$ are \emph{all} typical; we use this in Section \ref{sec:sipser-survives}.

With the notion of typical restrictions in hand, in Section \ref{sec:sipser-survives} we establish Property 2 showing that $\BalancedSipser_d$ ``survives'' a $\mathbf{\Psi}$-random projection (i.e.~it ``retains structure'') with high probability.  More formally, for outcomes $\Psi \equiv \{ \rho^{(k)}\}_{k\in \{2,\ldots,d\}}$ of $\mathbf{\Psi}$ such that $\rho^{(d)},\ldots,\rho^{(2)}$ are all typical, we prove that the $\Psi$-projected target $\Psi(\BalancedSipser_d)$ is ``well-structured'' in the following sense: 
\begin{enumerate}
\item[(i)]
$\Psi(\BalancedSipser_d)$ is a depth-one formula: an $\OR$ if $d$ is even, an $\AND$ if $d$ is odd. 
 \item[(ii)] The bias of $\Psi(\BalancedSipser_d)$ under $\calD$ is close to $1/2$; that is,
\[ \bias(\Psi(\BalancedSipser_d),\bY) = \frac1{2} - o_n(1). \] 
\end{enumerate} 
Recall that we have shown in Subsection~\ref{sec:typical-stuff} that with high probability $\mathbf{\Psi} \equiv \{\brho^{(k)}\}_{k\in \{2,\ldots,d\}}$ is such that $\brho^{(d)},\ldots,\brho^{(2)}$ are all typical. Therefore, the results of these two subsections together imply that the randomly projected target $\mathbf{\Psi}(\BalancedSipser_d)$ satisfies both (i) and (ii) with high probability. 
%(This is a stronger statement than the fact that $\mathbf{\Psi}(\BalancedSipser_d)$ is close to balanced under $\calD$ with high probability, i.e.~that $\Pr[\mathbf{\Psi}(\BalancedSipser_d)(\bY) = 1] = 1/2 \pm o_n(1)$ with high probability over $\mathbf{\Psi}$; see Remark~\ref{rem:sipser-is-balanced}).

\paragraph{Section~\ref{sec:puttogether}.}
Having established Properties 1, 2, and 3, it remains to bound the correlation between a depth-one formula  with bias essentially $1/2$ and a small-width CNF formula of opposite alternation with respect to the product distribution $\calD$ over $\zo^{w_0}$. (Recall that our results from Section~\ref{sec:sipser-survives} show that $\mathbf{\Psi}(\BalancedSipser_d)$ collapses to the former with high probability, and our results from Section~\ref{sec:approximator-simplifies} shows that $\mathbf{\Psi}(C)$ collapses to the latter with high probability --- this holds in both cases since a shallow decision tree is a small-width CNF.)  We prove this correlation bound using a slight extension of an argument in~\cite{OW07}, and with this final piece in hand our main theorems follow from straightforward arguments putting the pieces together.

%\newpage

\section{Composition of projections complete to uniform} \label{sec:complete-to-uniform}

Our goal in this section is to establish the following lemma:

\begin{proposition}
\label{prop:complete-to-uniform}
Consider $f,g : \zo^n \to \zo$.  Let $\bX \leftarrow \{ 0_{1/2}, 1_{1/2}\}^{n}$. Let $\bY \leftarrow \{0_{1-t_1},1_{t_1}\}^{w_0}$ if $d$ is even, and $\bY\leftarrow \{0_{t_1}, 1_{1-t_1}\}^{w_0}$ if $d$ is odd.  Then
\[ \Pr[f(\bX) \ne g(\bX)] = \Pr[(\mathbf{\Psi}(f))(\bY) \ne (\mathbf{\Psi}(g))(\bY)]. \]
\end{proposition}

As discussed in Section~\ref{sec:second-outline} we will ultimately apply Proposition~\ref{prop:complete-to-uniform} with $f$ being our target function $\BalancedSipser_d$ and $g$ being the approximating circuit $C$. This allows us to translate the inapproximability of $\mathbf{\Psi}(\BalancedSipser_d)$ by $\mathbf{\Psi}(C)$ (either with respect to the $t_1$-biased or $(1-t_1)$-based product distribution, depending on whether $d$ is even or odd) into the uniform-distribution inapproximability of $\BalancedSipser_d$ by $C$.

\paragraph{Overview of proof.} We will actually derive Proposition~\ref{prop:complete-to-uniform} as a consequence of a stronger claim, which, roughly speaking, states that we can generate a uniformly random input $\bX \leftarrow\{0_{1/2},1_{1/2}\}^n$ via $\mathbf{\Psi}$ and $\bY$ in a stage-wise manner.   In more detail, given $\mathbf{\Psi} \equiv \{\brho^{(k)}\}_{k\in \{2,\ldots,d\}}$ and $\bY$ we consider the following random $\{0,1,\ast\}$-valued labeling $\bell$ of the leaves and non-root nodes of the depth-$d$ depth-regular tree corresponding to the depth-$d$ formula computing $\BalancedSipser_d$:
\begin{itemize}
\item The $|A_d| = n$ leaves of the tree are each labeled $\{0,1,\ast\}$ according to $\brho^{(d)} \leftarrow \calR_\init$.
\item For $2\le k \le d-1$, the $|A_k|$  nodes at depth $k$ are each labeled $\{0,1,\ast\}$ according to $\brho^{(k)}\leftarrow \calR(\hat{\brho^{(k+1)}})$.
\item Finally, for each $i\in [w_0] = [|A_1|]$, if $\hat{\brho^{(2)}}_i = \ast$ then the $i$-th node at depth $1$ is labeled $\bY_i \in \zo$, and otherwise it is labeled $\hat{\brho^{(2)}}_i \in \zo$. (The root of the tree is left unlabeled.)
\ignore{
\gray{\item Finally, the $|A_1| = w_0$  nodes at depth $1$\lnote{Actually, just the $\ast$'s of $\hat\brho^{(2)}$} are each labeled $\{0,1\}$ according to $\bY$; that is, each of them is independently labeled $1$ with probability $t_1$, and $0$ otherwise. (The root of the tree is left unlabeled.)}
}
\end{itemize}

Next, we let the $\zo$-valued labels of $\bell$ ``percolate down the tree'' as follows: every node or leaf that is labeled $\ast$ by $\bell$ inherits the ($\{0,1\}$-valued) label from its closest ancestor that is not labeled $\ast$.  Note that this ``percolation step'' ensures that every leaf and non-root node of the tree is labeled either $0$ or $1$, since every depth-$1$ node is assigned a $\zo$-valued label by $\bell$.

Let $\bell^{\downarrow}$  denote this $\zo$-valued random labeling of the leaves and non-root nodes.  Our main result in this section, Proposition~\ref{prop:preserve-dist}, can be viewed as stating that the random string $\bX \in \zo^n$ defined by $\bell^\downarrow$'s labeling of the $n$ leaves is distributed uniformly at random; Proposition~\ref{prop:complete-to-uniform} follows as a straightforward consequence of this claim along with our definition of projections.\bigskip

We begin with the following lemma, which explains our choice of $t_{d-1}$ in (\ref{eq:def-of-tk}) in the definition of $\calR_\init$ (Definition~\ref{def:calR-init}).  (Note that in the lemma each coordinate of $\bY$ is distributed as $\{0_{1-t_{d-1}},1_{t_{d-1}}\}$ regardless of whether $d$ is even or odd; this is because of our convention that the bottom-layer gates of $\BalancedSipser_d$ are always $\AND$ gates.)

\begin{lemma}
\label{lem:OW-trick-initial}
Let $\brho\leftarrow\calR_\init$ and $\bY \leftarrow \{ 0_{1-{t_{d-1}}}, 1_{t_{d-1}} \}^{(\hat{\brho})^{-1}(\ast)}$ \ignore{\lnote{This used to say $\bY \leftarrow \{0_{1-t},1_t\}^{A_{d-1}}$ but that is misleading (though the lemma is still formally correct). We will only be hitting the $\ast$'s of $\hat\rho$ with the $t$-biased distribution.}}, and consider the string $\bX \in \{0,1\}^{n} \equiv \zo^{A_{d-1} \times [m]}$ defined as follows:
\[ \bX_{a,i} = \left\{
\begin{array}{cl}
\bY_a & \text{if $\brho_{a,i} = \ast$} \\
\brho_{a,i} & \text{otherwise}
\end{array}
\quad \text{for all $a\in A_{d-1}$ and $i\in [m]$.}
\right.\]
The string $\bX$ is distributed according to the uniform distribution $\{ 0_{1/2},1_{1/2}\}^{n}$. (Recalling Remark~\ref{rem:complete-dead-blocks} we have that $\brho_{a,i} = \ast$ if and only if $\hat\brho_a =\ast$, and so $\bY_a$ in the equation above is indeed well-defined.)
\end{lemma}

\begin{proof}
Since the blocks of $\brho$ are independent across $a\in A_{d-1}$ and the coordinates of $\bY$ are independent across $a \in (\hat\brho)^{-1}(\ast) \sse A_{d-1}$, it suffices to prove that $\bX_a$ is distributed according to $\{ 0_{1/2},1_{1/2}\}^{m}$ for a fixed $a\in A_{d-1}$.  We first observe that
\[
\Pr[\bX_a = 1^{m}] = \lambda + q\Pr[\bY_a = 1]
= \lambda + qt_{d-1} = p = 2^{-m},
\]
\ignore{\begin{align*}
\Pr[\bX_a = 1^{m}] &= \lambda + q\Pr[\bY_a = 1]  \\
&= \lambda + qt_{d-1}  \\
&= p = 2^{-m},
\end{align*}}
where the $\lambda$ is from the first line of (\ref{eq:initial}), the $q\Pr[\bY_a = 1]$ is from the second line of (\ref{eq:initial}), and the penultimate equality is by our choice of $t_{d-1}$ in (\ref{eq:def-of-tk}).  Next, for any string $\mathrm{Z} \in \zo^{m} \setminus \{1\}^{m}$, we have that
\begin{align}
\Pr[\bX_a = \mathrm{Z}] &= \left(1-\lambda-q\right) \cdot {\frac {2^{-m}} {1 - 2^{-m}}}
+ q \Pr[\bY_a= 0] \cdot  {\frac {2^{-m}} {1 - 2^{-m}}} \label{eq:OW-trick-initial-1}\\
&=
\frac{p}{1-p}  \cdot \left(
 \left(1-\lambda-q\right)  + q(1-t_{d-1}) \right) \nonumber \\
 &=
 \frac{p} {1-p}  \cdot
 \left(1-\lambda-qt_{d-1} \right) = \frac{p}{1-p} \cdot (1-p) = p = 2^{-m},\label{eq:OW-trick-initial-2}
\end{align}
where the first summand on the RHS of (\ref{eq:OW-trick-initial-1}) is by the third line of (\ref{eq:initial}), the second summand is by the second line of (\ref{eq:initial}), and (\ref{eq:OW-trick-initial-2}) again uses our choice of $t_{d-1}$ in (\ref{eq:def-of-tk}). Since this is exactly the probability mass function of the uniform distribution $\{0_{1/2},1_{1/2}\}^{m}$, the proof is complete. \end{proof}

The following lemma, the analogue of Lemma~\ref{lem:OW-trick-initial} for $\calR(\tau)$, explains our choice of $q_a$ in terms of $t_k$ and $t_{k-1}$ in (\ref{eq:def-of-qa}):

\begin{lemma}
\label{lem:OW-trick-subsequent}
For $2\le k \le d-1$ let $\tau \in \{0,1,\ast\}^{A_k}$, $\brho\leftarrow\calR(\tau)$, and
\[
\begin{cases}
\bY \leftarrow \{0_{1-t_{k-1}},1_{t_{k-1}}\}^{(\hat\brho)^{-1}(\ast)} & \text{~if~}d - k \equiv 0 \mod 2\\
\bY \leftarrow \{0_{t_{k-1}},1_{1-t_{k-1}}\}^{(\hat\brho)^{-1}(\ast)} & \text{~if~}d - k \equiv 1 \mod 2.
\end{cases}
\]
For each $a \in A_{k-1}$, writing $S_a = S_a(\tau)$ to denote $\tau_a^{-1}(\ast) = \{ i \in [w_{k-1}] \colon \tau_{a,i} = \ast\}$ and $\brho(S_a)$ to denote the substring of $\brho_a$ with coordinates in $S_a$, we consider the string $\bZ_a \in \{0,1\}^{S_a}$ defined as follows:
\[ \bZ_{a,i} = \left\{
\begin{array}{cl}
\bY_a & \text{if $\brho_{a,i} = \ast$} \\
\brho_{a,i} & \text{otherwise}
\end{array}
\quad \text{for all $i\in S_a$.}
\right.\]
The string $\bZ_a$ is distributed according to
\[
\begin{cases}
\{ 0_{t_k},1_{1-{t_k}}\}^{S_a} & \text{~if~}d - k \equiv 0 \mod 2\\
\{ 0_{1-t_k},1_{t_k}\}^{S_a} & \text{~if~}d - k \equiv 1 \mod 2,
\end{cases}
\]
and furthermore, $\bZ_a$ and $\bZ_{a'}$ are independent for any two distinct $a,a' \in A_{k-1}$.  (Again, recalling Remark~\ref{rem:complete-dead-blocks} we have  that $\brho_{a,i} = \ast$ if and only if $\hat\brho_a =\ast$, and so $\bY_a$ in the equation above is indeed well-defined.)
\end{lemma}

\ignore{ OLD VERSION OF LEMMA FOLLOWS:

\begin{lemma}
\label{lem:OW-trick-subsequent}
For $2\le k \le d-1$ where $d-k \equiv 0 \mod 2$ let $\tau \in \{0,1,\ast\}^{A_k}$, $\brho\leftarrow\calR(\tau)$, and $\bY \leftarrow \{0_{1-t_{k-1}},1_{t_{k-1}}\}^{(\hat\brho)^{-1}(\ast)}$\ignore{\lnote{Again the exponent here used to be $A_{k-1}$}}.  For each $a \in A_{k-1}$, writing $S_a = S_a(\tau)$ to denote $\tau_a^{-1}(\ast) = \{ i \in [w_{k-1}] \colon \tau_{a,i} = \ast\}$ and $\brho(S_a)$ to denote the substring of $\brho_a$ with coordinates in $S_a$, we consider the string $\bX_a \in \{0,1\}^{S_a}$ defined as follows:
\[ \bX_{a,i} = \left\{
\begin{array}{cl}
\bY_a & \text{if $\brho_{a,i} = \ast$} \\
\brho_{a,i} & \text{otherwise}
\end{array}
\quad \text{for all $i\in S_a$.}
\right.\]
The string $\bX_a$ is distributed according to $\{ 0_{t_k},1_{1-{t_k}}\}^{S_a}$, and furthermore, $\bX_a$ and $\bX_{a'}$ are independent for any two distinct $a,a' \in A_{k-1}$.  (Again, recalling Remark~\ref{rem:complete-dead-blocks} we have  that $\brho_{a,i} = \ast$ if and only if $\hat\brho_a =\ast$, and so $\bY_a$ in the equation above is indeed well-defined.)
\end{lemma}
}

\begin{proof}
We prove the $d-k \equiv 0 \mod 2$ case (the other case follows by a symmetric argument). If $\hat\tau_a$ falls in the first case of Definition~\ref{def:main-projection} (i.e.~if $\hat\tau_a = 0$ or if $S_a$ is not $k$-acceptable)  then the claim is   true since $\bZ_a \equiv \brho(S_a) \leftarrow \{0_{t_k},1_{1-{t_k}}\}^{S_a}$. Otherwise, if $\hat\tau_a$ falls in the second case of Definition~\ref{def:main-projection} (i.e.~if $\hat\tau_a = \ast$ and $S_a$ is $k$-acceptable) we first observe that
\begin{align*}
\Pr[\bZ_a = 1^{S_a}] &= \lambda + q_a\Pr[\bY_a = 1] \\
&=  \lambda + q_a t_{k-1} \\
&= (1-t_k)^{|S_a|},
\end{align*}
where as before the $\lambda$ is from the first line of (\ref{eq:OW-restriction}), the $q_a\Pr[\bY_a = 1]$ is from the second line of (\ref{eq:OW-restriction}), and the final equality is by our definition of $q_a$ in (\ref{eq:def-of-qa}). Next, for any string $\mathrm{Z} \in \{ 0, 1\}^{S_a} \setminus \{1\}^{S_a}$ and $u := | \mathrm{Z}^{-1}(0)| \in \{ 1, \ldots, |S_a|\}$, we have that
\begin{align}
\Pr[\bZ_a = \mathrm{Z}]
&= \left(1-\lambda-q_a\right) \cdot {\frac {t_k^u (1-t_k)^{|S_a|-u}} {1 - (1-t_k)^{|S_a|}}}
+ q_a \Pr[\bY_a = 0]\cdot  {\frac {t_k^u (1-t_k)^{|S_a|-u}} {1 - (1-t_k)^{|S_a|}}} \label{eq:OW-trick-subsequent-1} \\
&=
{\frac {t_k^u (1-t_k)^{|S_a|-u}} {1 - (1-t_k)^{|S_a|}}} \cdot
\left(1-\lambda-q_a +q_a(1-t_{k-1})\right) \nonumber \\
 &=
{\frac {t_k^u (1-t_k)^{|S_a|-u}} {1 - (1-t_k)^{|S_a|}}} \cdot
\left(1-\lambda-q_a t_{k-1}\right) \nonumber \\
&=
{\frac {t_k^u (1-t_k)^{|S_a|-u}} {1 - (1-t_k)^{|S_a|}}} \cdot \big(1-(1-t_k)^{|S_a|}\big) =
t_k^u (1-t_k)^{|S_a|-u}, \label{eq:OW-trick-subsequent-2}
\end{align}
where as before the first summand on the RHS of (\ref{eq:OW-trick-subsequent-1}) is by the third line of (\ref{eq:OW-restriction}), the second summand is by the second line of (\ref{eq:OW-restriction}), and (\ref{eq:OW-trick-subsequent-2}) again uses our definition of $q_a$. Therefore indeed, the resulting string is distributed according to $\{0_{t_k},1_{1-t_k}\}^{S_a}$. Finally, since the blocks of $\brho$ are independent across $a\in A_{k-1}$ and the coordinates of $\bY$ are independent across $a \in (\hat\brho)^{-1}(\ast) \sse A_{k-1}$, we have that $\bZ_a$ and $\bZ_{a'}$ are independent for any two distinct $a,a'\in A_{k-1}$.
\end{proof}

\ignore{ OLD REMARK FOLLOWS:

\begin{remark}
\label{rem:OW-trick-subsequent}
A symmetric argument establishes the following dual to Lemma~\ref{lem:OW-trick-subsequent}: for $2\le k \le d-1$ where $d-k \equiv 1 \mod 2$ let $\tau \in \{0,1,\ast\}^{A_k}$, $\brho\leftarrow\calR(\tau)$, and $\bY \leftarrow \{0_{t_{k-1}},1_{1-t_{k-1}}\}^{(\hat\brho)^{-1}(\ast)}$.  For each $a \in A_{k-1}$, writing $S_a = S_a(\tau)$ to denote $\tau_a^{-1}(\ast) = \{ i \in [w_{k-1}] \colon \tau_{a,i} = \ast\}$ and $\brho(S_a)$ to denote the substring of $\brho_a$ with coordinates in $S_a$, we consider the string $\bX_a \in \{0,1\}^{S_a}$ defined as follows:
\[ \bX_{a,i} = \left\{
\begin{array}{cl}
\bY_a & \text{if $\brho_{a,i} = \ast$} \\
\brho_{a,i} & \text{otherwise}
\end{array}
\quad \text{for all $i\in S_a$.}
\right.\]
The string $\bX_a$ is distributed according to $\{ 0_{1-t_k},1_{t_k}\}^{S_a}$, and furthermore, $\bX_a$ and $\bX_{a'}$ are independent for any two distinct $a,a' \in A_{k-1}$.
\end{remark}
}

Together Lemmas~\ref{lem:OW-trick-initial} and~\ref{lem:OW-trick-subsequent} give us the following proposition, which in turn yields Proposition~\ref{prop:complete-to-uniform}, our main result in this section.

\begin{proposition}
\label{prop:preserve-dist}
Let $\brho^{(d)} \leftarrow\calR_\init$ and $\brho^{(k)}\leftarrow\calR(\hat{\brho^{(k+1)}})$ for $2\le k \le d-1$.  Let
\[  \bY^{(1)} \leftarrow
\begin{cases}
 \{ 0_{1-t_1}, 1_{t_1}\}^{(\hat{\brho^{(2)}})^{-1}(\ast)} & \text{~if $d$ is even}\\
 \{ 0_{t_1}, 1_{1-t_1}\}^{(\hat{\brho^{(2)}})^{-1}(\ast)} & \text{~if $d$ is odd},\\
\end{cases}
\]
and for $2 \le k \le d-1$ consider random strings $\bY^{(k)} \in \zo^{(\hat{\brho^{(k+1)}})^{-1}(\ast)}$ defined inductively from $k=2$ up to $d-1$ as follows:
\begin{equation}
\bY^{(k)}_{a,i} = \left\{
\begin{array}{cl}
\bY^{(k-1)}_a & \text{if $\brho^{(k)}_{a,i} = \ast$}  \\
\brho^{(k)}_{a,i} & \text{otherwise}
\end{array}
\right.\quad \text{for all $a\in A_{k-1}$ and $i \in [w_{k-1}]$ s.t.~$\hat{\brho^{(k+1)}}_{a,i} = \ast$}.
\label{eq:def-of-Y}
\end{equation}
Then the string $\bX \in \{0,1\}^n \equiv \zo^{A_{d-1}\times [m]}$ defined by
\[ \bX_{a,i} = \left\{
\begin{array}{cl}
\bY^{(d-1)}_a & \text{if $\brho_{a,i} = \ast$} \\
\brho^{(d)}_{a,i} & \text{otherwise}
\end{array}
\quad \text{for all $a\in A_{d-1}$ and $i\in [m]$}
\right.\]
is distributed according to the uniform distribution $\{0_{1/2},1_{1/2}\}^n$.
\end{proposition}

\ignore{ OLD VERSION OF PROPOSITION:

\begin{proposition}
\label{prop:preserve-dist}
Let $d \ge 2$ be even. (A symmetric argument establishes the dual claim for odd values of $d$.) Let $\brho^{(d)} \leftarrow\calR_\init$ and $\brho^{(k)}\leftarrow\calR(\hat{\brho^{(k+1)}})$ for $2\le k \le d-1$.  Let
\[ \bY^{(1)} \leftarrow \{ 0_{1-t_1}, 1_{t_1}\}^{(\hat{\brho^{(2)}})^{-1}(\ast)},\]
and for $2 \le k \le d-1$ consider random strings $\bY^{(k)} \in \zo^{(\hat{\brho^{(k+1)}})^{-1}(\ast)}$ defined inductively from $k=2$ up to $d-1$ as follows:
\begin{equation}
\bY^{(k)}_{a,i} = \left\{
\begin{array}{cl}
\bY^{(k-1)}_a & \text{if $\brho^{(k)}_{a,i} = \ast$}  \\
\brho^{(k)}_{a,i} & \text{otherwise}
\end{array}
\right.\quad \text{for all $a\in A_{k-1}$ and $i \in [w_{k-1}]$ s.t.~$\hat{\brho^{(k+1)}}_{a,i} = \ast$}.
\label{eq:def-of-Y}
\end{equation}
Then the string $\bX \in \{0,1\}^n \equiv \zo^{A_{d-1}\times [m]}$ defined by
\[ \bX_{a,i} = \left\{
\begin{array}{cl}
\bY^{(d-1)}_a & \text{if $\brho_{a,i} = \ast$} \\
\brho^{(d)}_{a,i} & \text{otherwise}
\end{array}
\quad \text{for all $a\in A_{d-1}$ and $i\in [m]$}
\right.\]
is distributed according to the uniform distribution $\{0_{1/2},1_{1/2}\}^n$.
\end{proposition}

}

\begin{proof}
By the $k=2$ case of Lemma~\ref{lem:OW-trick-subsequent}, for all possible outcomes $\rho^{(k)}$ of $\brho^{(k)}$ for $3\le k \le d$, conditioned on such an outcome the random string $\bY^{(2)}$ is distributed according to $\{0_{t_2},1_{1-t_2}\}^{\hat{\rho^{(3)}}}$ if $d$ is even and according to
$\{0_{1-t_2},1_{t_2}\}^{\hat{\rho^{(3)}}}$ if $d$ is odd.  Applying this argument repeatedly and arguing inductively from $k=2$ up to $k=d-1$, we have that conditioned on any outcome $\rho^{(d)}$ of $\brho^{(d)}\leftarrow \calR_\init$, the random string $\bY^{(d-1)}$ is distributed according to $\{ 0_{1-t_{d-1}}, 1_{t_{d-1}}\}^{\hat{\rho^{(d)}}}$.  The claim then follows by Lemma~\ref{lem:OW-trick-initial}.
\end{proof}

\ignore{
\gray{
\begin{proposition}
%\label{prop:preserve-dist}
Let $d$ be even and consider $\BalancedSipser_d$. (A symmetric argument establishes the dual claim for odd values of $d$.) Let $\brho^{(d)} \leftarrow\calR_\init$ and $\brho^{(k)}\leftarrow\calR(\hat{\brho^{(k+1)}})$ for $2\le k \le d-1$.  Let $\bY^{(1)} \leftarrow \{ 0_{1-t_1}, 1_{t_1}\}^{w_0}$, and consider random strings $\bY^{(k)} \in \zo^{A_k}$ defined inductively from $k=2$ up to $d$ as follows:
\begin{equation}
\bY^{(k)}_{a,i} = \left\{
\begin{array}{cl}
\bY^{(k-1)}_a & \text{if $\brho^{(k)}_{a,i} = \ast$}  \\
\brho^{(k)}_{a,i} & \text{otherwise}
\end{array}
\right.\quad \text{for all $a\in A_{k-1}$ and $i \in [w_{k-1}]$}.
\label{eq:def-of-Y}
\end{equation}
Then the string $\bY^{(d)} \in \zo^{n}$ is distributed according to the uniform distribution $\{ 0_{1/2},1_{1/2}\}^{n}$.
\end{proposition}

\begin{proof}
Applying Lemma~\ref{lem:OW-trick-subsequent} (and Remark~\ref{rem:OW-trick-subsequent}) repeatedly and arguing inductively from $k=2$ up to $d-1$, we have that $\bY^{(k)}$ is distributed according to $\{0_t, 1_{1-t}\}^{A_k}$ if $k$ is even, and $\{0_{1-t},1_t\}^{A_k}$ if $k$ is odd.  The claim then follows by Lemma~\ref{lem:OW-trick-initial}.
\end{proof}
}
}

\begin{proof}[Proof of Proposition~\ref{prop:complete-to-uniform}]
Recall that $\bX \leftarrow \{0_{1/2},1_{1/2}\}^n$ and $\bY \leftarrow \{0_{1-t_1},1_{t_1}\}^{w_0}$ if $d$ is even,  $\bY\leftarrow \{0_{t_1}, 1_{1-t_1}\}^{w_0}$ if $d$ is odd. Let $\brho^{(d)}\leftarrow \calR_\init$ and $\brho^{(k)}\leftarrow \calR(\hat{\brho^{(k+1)}})$ for $2 \le k \le d-1$.  For $1 \le k \le d-1$ let $\bY^{(k)}\in \zo^{(\hat{\brho^{(k+1)}})^{-1}(\ast)}$ be defined as in Proposition~\ref{prop:preserve-dist}.  Recalling Remark~\ref{rem:complete-dead-blocks}, for all functions $h :\zo^n\to\zo$ and $1 \le k \le d-1$, the random projection
\[ (\proj_{\brho^{(k+1)}} \cdots \proj_{\brho^{(d)}}\, h) : \zo^{A_{k}} \to \zo \]
depends only on the coordinates in $(\hat{\brho^{(k+1)}})^{-1}(\ast) \sse A_{k}$, and so we may equivalently view it as a function $\zo^{(\hat{\brho^{(k+1)}})^{-1}(\ast)} \to \zo$.  By Proposition~\ref{prop:preserve-dist}, the definition of the $\bY^{(k)}$'s, and the definition of projections, we see that
\begin{align*}
\Pr[f(\bX) \ne g(\bX)]
&= \Pr[(\proj_{\brho^{(d)}}\,f)(\bY^{(d-1)}) \ne (\proj_{\brho^{(d)}}\,g)(\bY^{(d-1)})] \\
&= \Pr[(\proj_{\brho^{(d-1)}}\,\proj_{\brho^{(d)}}\,f)(\bY^{(d-2)}) \ne (\proj_{\brho^{(d-1)}}\,\proj_{\brho^{(d)}}\,g)(\bY^{(d-2)})] \\
&= \ \cdots \\
&= \Pr[(\proj_{\brho^{(2)}} \cdots \proj_{\brho^{(d)}}\, f)(\bY^{(1)}) \ne (\proj_{\brho^{(2)}} \cdots \proj_{\brho^{(d)}}\, g)(\bY^{(1)})] \\
&= \Pr[(\proj_{\brho^{(2)}} \cdots \proj_{\brho^{(d)}}\, f)(\bY) \ne (\proj_{\brho^{(2)}} \cdots \proj_{\brho^{(d)}}\, g)(\bY)] \\
&=\Pr[(\mathbf{\Psi}(f))(\bY) \ne (\mathbf{\Psi}(g))(\bY)]
\end{align*}
where the final inequality is by the definition of $\mathbf{\Psi}$ (Definition~\ref{def:boldpsi}).
\end{proof}

\section{Approximator simplifies under random projections}
\label{sec:approximator-simplifies}

With Proposition~\ref{prop:complete-to-uniform} in hand we next prove that the approximating circuit $C$ of the type specified in either Theorems~\ref{thm:smallbottomfanin} or~\ref{thm:alternationpattern} ``collapses to a simple function'' with high probability under a $\mathbf{\Psi}$-random restriction.  For the case that the depth-$d$ circuit $C$ has significantly smaller bottom fan-in than $\BalancedSipser_d$ we show that $C$ collapses to a shallow decision tree with high probability, and for the case that $C$ has the opposite alternation pattern to $\BalancedSipser_d$ we show that $C$ collapses to a small-width depth-two circuit with top gate opposite to that of $\mathbf{\Psi}(\BalancedSipser_d)$ with high probability.

We do so via a \emph{projection switching lemma}, showing that each of the $d-1$ individual random projections $\proj_{\brho^{(k)}}$ comprising $\mathbf{\Psi}$ ``contribute to the simplification'' of $C$ with high probability.  We state and prove our projection switching lemma in Sections~\ref{sec:psl} through~\ref{sec:weight-diff}, and in Section~\ref{sec:two} we show how the lemma can be applied iteratively to prove our structural claims about $\mathbf{\Psi}(C)$.

\subsection{The projection switching lemma and its proof}
\label{sec:psl}

\begin{proposition}[Projection switching lemma for $\calR_{\init}$]
\label{prop:psl-initial}
Let $F : \zo^{n} \to \zo$ be a depth-$2$ circuit with bottom fan-in $r$.   Then for all $s\ge 1$,
\[ \Prx_{\brho \leftarrow \calR_{\init}}[\proj_{\brho}\, F \text{~is not a depth-$s$ decision tree}]
=\left(O\Big(r2^r \cdot w^{-1/4}\Big)\right)^s.
\]
\end{proposition}

\begin{proposition}[Projection switching lemma for $\calR(\tau)$] \label{prop:psl-typical-tau}
Let $2 \leq k \leq d-1$ and $F: \{0,1\}^{A_k} \to \{0,1\}$ be a depth-$2$ circuit with bottom fan-in $r$.  Then for all $\tau \in \{0,1,\ast\}^{A_{k}}$ and $s\ge 1$,
\[
\Prx_{\brho \leftarrow \calR(\tau)}[\proj_\brho\, F \text{~is not a depth-$s$ decision tree}]
= \left(O\Big(re^{rt_k/(1-t_k)} \cdot w^{-1/4}\Big)\right)^s.\]
\end{proposition}

The proofs of Propositions~\ref{prop:psl-initial} and~\ref{prop:psl-typical-tau} have the same overall structure, and they share many of the same ingredients. We will only prove (the slightly more involved) Proposition~\ref{prop:psl-typical-tau}, and at the end of this section we point out the essential differences in the proof of Proposition~\ref{prop:psl-initial}.

Furthermore, we will prove Proposition~\ref{prop:psl-typical-tau} assuming that $F$ is a DNF and $d-k \equiv 0 \mod 2$. Both assumptions are without loss of generality. (For the first, we recall that $F$ is a width-$r$ DNF if and only if its Boolean dual $F^\dagger$ is a width-$r$ CNF, and that a Boolean function is computed by a depth-$s$ decision tree if and only if its Boolean dual is as well, and we observe that $(\proj_\rho\, F)^\dagger = \proj_\rho\, (F^\dagger)$ for all $\rho$ and all $F$.  For the second we note that the definition of $\calR(\tau)$ when $d-k\equiv 0 \mod 2$ is dual to that of $\calR(\tau)$ when $d-k\equiv 1\mod 2$, and so applying the former to $F(x)$ is equivalent to applying the latter to $F(\overline{x})$.)

\paragraph{Overview of proof.} At a high level, we adopt Razborov's strategy in his alternative proof~\cite{Raz95} of H{\aa}stad's Switching Lemma. We briefly recall the overall structure of Razborov's argument.  Given a DNF $F : \zo^n\to\zo$ and a distribution $\calR$ over restrictions in $\{0,1,\ast\}^n$, we let $\calB \sse \{0,1,\ast\}^n$ denote the set of all \emph{bad} restrictions, namely the ones such that $F \uhr \rho$ is not computed by a small-depth decision tree.  Our goal in a switching lemma is to bound $\Pr_{\brho\leftarrow\calR}[\brho\in \calB]$, the weight of $\calB$ under $\calR$.  To do so, we define an \emph{encoding} of each bad restriction $\rho \in \calB$ as a different restriction $\rho' \in \{0,1,\ast\}^n$ and a small amount (say at most $\ell$ bits) of ``auxiliary information'':
\begin{gather*}
\encode : \calB \rightarrow \{0,1,\ast\}^n \times \zo^{\ell} \\
 \encode(\rho) = (\rho', \text{auxiliary information}).
 \end{gather*}
This encoding should satisfy two key properties. First, it should be uniquely decodable, meaning that one is always able to recover $\rho$ given $\rho'$ and the auxiliary information; equivalently, the function $\encode(\cdot)$ is an injection.  Second, the weight $\Pr_{\brho\leftarrow\calR}[\brho = \rho']$ of $\rho'$ under $\calR$ should be larger than that of $\rho$ by a significant multiplicative factor (say by a factor of $\Gamma$).  It is not hard to see that together, these two properties imply that total weight of all bad restrictions with the \emph{same} auxiliary information is at most $1/\Gamma$.  To complete the proof of the switching lemma, we then bound the overall weight of $\calB$ via a union bound over all $2^\ell$ possible strings of auxiliary information. (For a detailed exposition of Razborov's proof technique see~\cite{Beame:94,Tha09} and Chapter \S14 of~\cite{AB09}.)

The proof of our projection switching lemma follows this high-level strategy quite closely; specifically, we build off of a reformulation (due to Thapen~\cite{Tha09}) of H{\aa}stad's proof of the blockwise variant of his Switching Lemma in Razborov's framework.  In Section~\ref{sec:encode} we define our encoding, specifying the restriction $\rho'$ and auxiliary information that is associated with every bad restriction $\rho$; in Section~\ref{sec:decode} we prove that our encoding is an injection by describing a procedure for unique decoding; in Section~\ref{sec:weight-diff}  we verify that every bad $\rho$ is indeed paired with a $\rho'$ whose weight under $\calR(\tau)$ is much larger, and show how this completes the proof of our projection switching lemma.

 One important aspect in which we differ from H{\aa}stad's and Razborov--Thapen's proof --- and indeed, this is the key distinction between our projection switching lemma and previous switching lemmas --- is that we will be concerned with the complexity of the randomly \emph{projected} DNF $\proj_{\brho}\,F \equiv \proj\,(F \uhr \brho)$, rather than the randomly \emph{restricted} DNF $F\uhr \brho$.  Recalling our definition of projections (Definition~\ref{def:projection}) and Remark~\ref{rem:equiv-projection} in particular, we see that the decision tree depth of $\proj\,(F\uhr \brho)$ can in general be significantly smaller than that of $F\uhr \brho$, since groups of distinct formal variables $\{ x_{a,i} \colon i\in [w]\}$  of $F \uhr \brho$ get mapped to the same formal variable $y_a$ under the projection operator.  As we will see, the proof of our projection switching lemma crucially exploits this fact.

\subsection{Canonical projection decision tree}

To emphasize the fact that the DNF $F$ and its random projection $\proj_\brho\,F$ are over two different spaces of formal variables, we will let $\calX = \{ x_{a,i} \colon a \in A_{k-1}, i \in [w_{k-1}]\} $ denote the formal variables of $F$, and $\calY = \{ y_a \colon a \in A_{k-1}\}$ denote the formal variables of $\proj_\brho\,F$.  For notational clarity, from this section through Section~\ref{sec:decode} we omit the subscripts on $A_{k-1}$ and $w_{k-1}$ and simply write $A$ and~$w$.

\begin{definition}Let $G : \zo^{A\times [w]}\to\zo$ be a DNF over $\calX$ and $T$ be a term in $G$. We say that a variable $x_{a,i}$ \emph{occurs positively in $T$} if $T$ contains the unnegated literal $x_{a,i}$, and that it \emph{occurs negatively in $T$} if $T$ contains the negated literal $\overline{x}_{a,i}$.  We say that $x_{a,i}$ \emph{occurs in $T$} if it either occurs positively or negatively in $T$.
\end{definition}

\begin{definition}
For any $\eta\sse \calY$ and assignment $\pi \in \zo^{\eta}$, the restriction $(\eta \mapsto\pi) \in \{0,1,\ast\}^{A\times [w]}$ to the variables in $\calX$ is defined as follows: for all $a \in A$ and $i\in [w]$,
\[ (\eta\mapsto\pi)_{a,i} = \left\{
\begin{array}{cl}
\pi(y_a) & \text{if $y_a \in \eta$} \\
\ast & \text{otherwise.}
\end{array}
\right.\]
\end{definition}

We stress that for a given $a$, the value of $(\eta \mapsto \pi)_{a,i}$ is independent of the value of
$i \in [w].$

Next, we define a procedure which, given any DNF $G$ over $\calX$, returns a ``canonical'' decision tree $\Tree(G)$ over $\calY$ computing its projection $\proj\, G$. The proof of our switching lemma will establish that the depth of $\Tree(F\uhr\brho)$ is small with high probability; this clearly implies that the decision tree depth of $\proj_\brho\,F \equiv \proj\, (F\uhr \brho)$ is small with high probability. (We remark that both H{\aa}stad's and Razborov's proofs of H{\aa}stad's Switching Lemma consider an analogous notion of a canonical decision tree whose depth they bound; in their context, however, the canonical decision tree computes the DNF itself, whereas the canonical decision tree we now define computes the \emph{projection} of the DNF.)
\begin{definition}[Canonical projection decision tree]
\label{def:proj-tree}
Let $G : \zo^{A\times [w]}\to\zo$ be a DNF over $\calX$,  where we assume a fixed but arbitrary ordering on its terms, and likewise on the literals within each term. The \emph{canonical projection decision tree} $\Tree(G) : \zo^A\to\zo$ associated with $G$ is defined recursively as follows:
\begin{enumerate}
\item If $G \equiv 1$ (i.e.~if $G(X) = 1$ for all $X\in\zo^{A \times [w]}$) output the trivial decision tree $\Tree(G) \equiv 1$, and likewise, if $G \equiv 0$ output $\Tree(G) \equiv 0$.
\item Otherwise, let $T$ be the first term in $G$ such that $T \not\equiv 0$, and let
\[ \eta = \big\{ y_a \colon x_{a,i} \text{ occurs in $T$ for some $i\in [w]$}\big\} \sse \calY \]
\item $\Tree(G)$ queries all the variables in $\eta$ in its first $|\eta|$ levels.
\item For each path $\pi \in \zo^{\eta}$, recurse on $G \uhr (\eta\mapsto\pi)$.
\end{enumerate}
\end{definition}

We stress that while $G$ is a DNF over the variables in $\calX$, the canonical projection decision tree $\Tree(G)$ queries variables in $\calY$. The following fact is a straightforward consequence of Definition~\ref{def:proj-tree}:

\begin{fact}
\label{fact:tree-computes-projection}
$\Tree(G)$ computes $\proj\, G$.
\end{fact}

\subsection{Encoding bad restrictions}
\label{sec:encode}

Fix $\tau \in \{0,1,\ast\}^{A\times [w]}$, and consider
\[ \calB = \big\{ \rho \in \{0,1,\ast\}^{A\times [w]}\colon \text{$\rho$ refines $\tau$ and }\proj_\rho\,F \text{ is not a depth-$s$ decision tree}\big\}, \]
We call these restrictions $\rho\in \calB$ \emph{bad}, and recall that our goal is to bound $\Pr[\brho\in\calB]$ for $\brho\leftarrow\calR(\tau)$.  Fix a bad restriction $\rho \in \calB$. It will be convenient for us to adopt the equivalent view of $\proj_\rho\,F$ as $\proj\,(F\uhr \rho)$ in this section. Since $\proj\,(F\uhr\rho)$ is not computed by a depth-$s$ DT over $\calY$, this in particular implies that the canonical projection decision tree $\Tree(F\uhr \rho)$ has depth at least $s$ (recall by Fact~\ref{fact:tree-computes-projection} that $\Tree({F\uhr \rho})$ computes $\proj\,(F\uhr \rho))$, and so we may let $\pi \in \zo^{\ge s}$ be the leftmost root-to-leaf path of length at least $s$ in $\Tree({F\uhr \rho})$. \medskip

We now define a few objects associated with $\rho$ and $\pi$: for some $1\le j\le s$, we define
\begin{itemize}
\item A collection of terms $T_1,\ldots,T_j$ in $F$.
\item Disjoint sets of variables $\eta_1,\ldots,\eta_j \sse \calY$, and for each such $\eta_\ell$, a bit string $\encode(\eta_\ell) \in \zo^{|\eta_\ell|(\log r+1)}$.
%where each $\eta_\ell$ is the set of $y_a$ such that $x_{a,j}$ occurs in $T_\ell \uhr \rho$ for some $j\in [w]$.
\item A restriction $\sigma = \sigma^1 \sigma^2 \cdots \sigma^j \in \{0,1,\ast\}^{A\times [w]}$ such that $\sigma^{-1}(\{0,1\}) \sse \rho^{-1}(\ast)$ (i.e.~$\sigma$ only sets to constants variables left free by $\rho$).
\item Disjoint sets of variables $\gamma_1,\ldots,\gamma_j \sse \calX$, and for each such $\gamma_\ell$, a bit string $\encode(\gamma_\ell)$ of Hamming weight $|\gamma_\ell|$ and length $r$.
%where each $\gamma_\ell$ is a the set of variables in $\calX$ that occur positively in $T_\ell \uhr \rho$.
\item A decomposition of the length-$s$ prefix $\pi' = \pi^1\pi^2 \cdots \pi^j \in \zo^s$ of $\pi$.
\end{itemize}
These objects are defined inductively starting from $\ell = 1$ up to $\ell = j$, where $j \in [s]$ is the smallest integer such that the $\eta_\ell$'s as defined below satisfy $|\eta_1 \cup\cdots \cup \eta_j| \ge s$. For $\ell \in [j]$,

\begin{itemize}[leftmargin = 0.5cm]
\item $T_\ell$ is the first term in $F$ such that $T_\ell \uhr \rho\,(\eta_1\mapsto\pi^1) \cdots (\eta_{\ell-1}\mapsto \pi^{\ell-1}) \not\equiv 0$
 and
\[ \eta_\ell = \big\{ y_a \colon x_{a,i} \text{ occurs in $T_\ell \uhr \rho\,(\eta_1\mapsto \pi^1) \cdots (\eta_{\ell-1}\mapsto \pi^{\ell-1})$ for some $i\in [w]$}\big\} \sse \calY. \]
We define $\encode(\eta_\ell) \in \zo^{|\eta_\ell|(\log r + 1)}$ as follows: for each $y_a \in \eta_\ell$, we use $\log|T_\ell|\le \log r$ bits to encode the location of $x_{a,i_1}$ in $T_\ell$, where
\[ i_1 := \min\big\{ i \in [w] \colon x_{a,i} \text{ occurs in $T_\ell \uhr \rho\,(\eta_1\mapsto \pi^1) \cdots (\eta_{\ell-1}\mapsto \pi^{\ell-1}) $}\big\}, \]
along with a single bit to  indicate whether $y_a$ is the last variable in $\eta_\ell$.

\item Let $\sigma^\ell \in \{0,1,\ast\}^{A\times [w]}$ be defined as follows: for each $y_a \in \eta_\ell$ and $i \in [w]$,
\[
\sigma^\ell_{a,i} = \left\{
\begin{array}{cl}
1 & \text{if $x_{a,i} \text{ occurs positively in } T_\ell \uhr \rho\,(\eta_1\mapsto\pi^1) \cdots (\eta_{\ell-1}\mapsto \pi^{\ell-1})$,} \\
0 & \text{if $x_{a,i} \text{ occurs negatively in } T_\ell \uhr\rho\,(\eta_1\mapsto \pi^1) \cdots (\eta_{\ell-1}\mapsto \pi^{\ell-1})$,} \\
0 & \text{if $\rho_{a,i} = \ast$ and $x_{a,i} \text{ does not occur in } T_\ell \uhr\rho\,(\eta_1\mapsto \pi^1) \cdots (\eta_{\ell-1}\mapsto \pi^{\ell-1})$.}
\end{array}
\right.
\]

%\rnote{Maybe we should add some prose here commenting on / explaining why we ``break symmetry'' here as we do by
%going with 0 for the third option.}
%\lnote{Agreed, I'll add this. We want to completely kill the $\{\ast,1\}$-valued block, and not pay the price for $\encode(\gamma)$.}

(Note that if $x_{a,i}$ occurs in $T_\ell \uhr \rho\,(\eta_1\mapsto\pi^1) \cdots (\eta_{\ell-1}\mapsto \pi^{\ell-1})$, then certainly $\rho_{a,i} = \ast$.)  All remaining entries of $\sigma^\ell$ not specified above have value $\ast$.

We make a few observations that will be useful for us later. First observe that for every $y_a \in \eta_\ell$, 
\begin{enumerate}
\item[(i)] $(\rho\,(\eta_1\mapsto\pi^1) \cdots (\eta_{\ell-1}\mapsto \pi^{\ell-1}))_a \equiv \rho_a$
\end{enumerate} 
since $y_a \notin \eta_1 \cup \cdots \cup \eta_{\ell-1}$. Furthermore, writing $S_a = S_a(\tau)$ to denote $\tau_a^{-1}(\ast) = \{i\in [w]\colon \tau_{a,i} = \ast\}$ and $\rho(S_a)$  to denote the substring of $\rho_a$ with coordinates in $S_a$, we claim that for every $y_a \in \eta_\ell$, 
\begin{enumerate}
\item[(ii)] $\tau_a \in \{ \ast,1\}^{w}\setminus \{1\}^w$, 
\item[(iii)] the set $S_a$ is $k$-acceptable,
\item[(iv)]  $\rho(S_a) \in \{\ast,1\}^{S_a} \setminus \{1\}^{S_a}$ (and hence $\rho_a \in \{\ast,1\}^{w}\setminus \{1\}^w$ by (ii)),
\item[(v)] $(\rho\sigma^\ell)(S_a)  \in \{ 0,1\}^{S_a}$,
\item[(vi)] $(\rho(S_a))^{-1}(1) \sse ((\rho\sigma^\ell)(S_a))^{-1}(1)$.
\end{enumerate} 
To see this, first note that since $y_a \in \eta_\ell$ it must be the case that $\rho_{a,i} =\ast$ for at least one $i \in S_a$, and by inspecting (\ref{eq:OW-restriction}) of Definition~\ref{def:main-projection} we have that indeed (ii), (iii), and (iv) hold. Claims (v) and (vi) follow from the fact that $\sigma^{\ell}$ is defined so that $\sigma^{\ell}_{a,i} \in \zo$ iff $\rho_{a,i} = \ast$. These claims will be useful for us later in the proof of Lemma~\ref{lem:fix-other-outcomes}.

Second, we claim that
\begin{equation} T_\ell \uhr\rho\,(\eta_1\mapsto \pi^1) \cdots (\eta_{\ell-1}\mapsto \pi^{\ell-1}) \sigma^{\ell} \equiv 1. \label{eq:satisfies-term}
\end{equation}
To see this, we note that every variable that occurs in term $T_\ell \uhr \rho\,(\eta_1\mapsto\pi^1) \cdots (\eta_{\ell-1}\mapsto \pi^{\ell-1})$ is fixed by $\sigma^\ell$, and furthermore, each is fixed in the unique way so as to satisfy the term. This will be useful for us later in the proof of Proposition~\ref{prop:injection}.

\[
\gamma_\ell = \{x_{a,i}: \sigma^\ell_{a,i}=1\} \sse \calX,
\]
and let $\encode(\gamma_\ell)$  be the string $\encode(\gamma_\ell) \in \{0,1\}^r$ of Hamming weight $|\gamma_\ell|$ and length $r$  indicating the location of the elements of $\gamma_\ell$ within $T_\ell$.
\item Let $\pi^\ell$ be the length-$|\eta_\ell|$ substring of $\pi$ from index $|\eta_1 \cup\cdots\cup \eta_{\ell-1}|+ 1$ through $|\eta_1 \cup\cdots \cup \eta_\ell|$ inclusive.
\end{itemize}

After the final iteration $\ell = j$, if necessary, we trim $\eta_j$ and $\pi^j$ so that $|\eta_1 \cup\cdots \cup \eta_j| = |\pi^1\cdots \pi^j|$ is \emph{exactly} $s$,  and redefine $\sigma^j$ and $\gamma_j$ appropriately. We refer the reader to Figure~\ref{figure:PSL-figure} and its caption for a concrete example and explanation of our encoding procedure.
\begin{figure}[p!]
\begin{center}
\includegraphics[width=10cm]{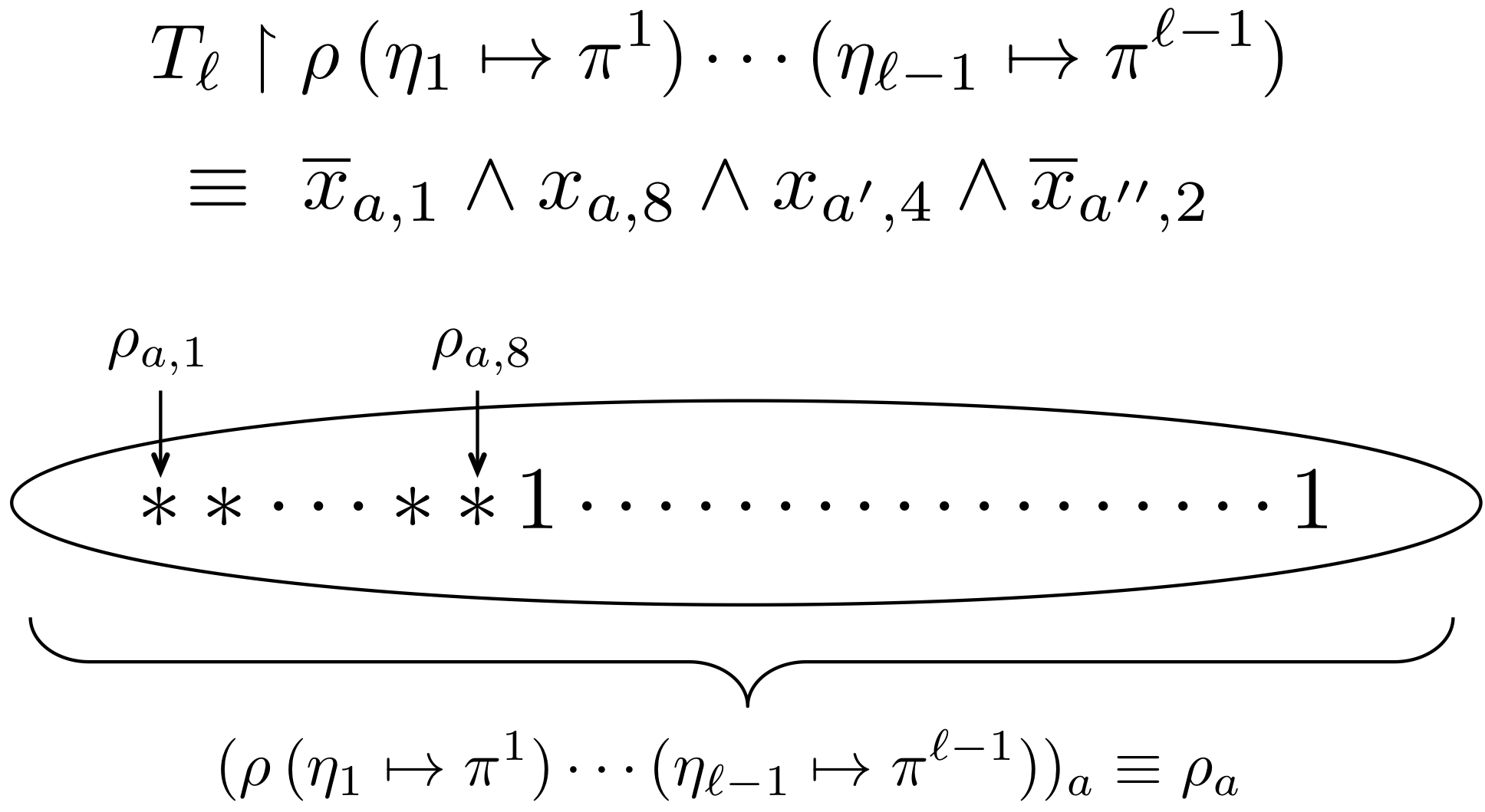}
\end{center}
\caption{Let $T_\ell$ be the first term not falsified by $\rho\,(\eta_1\mapsto\pi^1)\cdots (\eta_{\ell-1}\mapsto\pi^{\ell-1})$, and suppose it evaluates to  $\overline{x}_{a,1} \wedge x_{a,8} \wedge x_{a',4} \wedge \overline{x}_{a'',2}$.  In this example $\eta_\ell$ will be the set $\{ y_a,y_{a'}, y_{a''}\}\sse \calY$.  Focusing on variables from the $a$-th block, we first recall our observation earlier that $(\rho\,(\eta_1\mapsto\pi^1)\cdots (\eta_{\ell-1}\mapsto\pi^{\ell-1}))_a \equiv \rho_a$ since $y_a \notin \eta_1 \cup \cdots\cup \eta_{\ell-1}$ (Claim (i) in Section~\ref{sec:decode}). Furthermore, as illustrated above, we have that $\rho_a \in \{\ast,1\}^w \setminus \{1\}^w$ and $\rho_a$ refines $\tau_a \in \{\ast,1\}^{w} \setminus \{1\}^w$ (Claims (ii) and (iv) of Section~\ref{sec:decode}). \\ \hspace*{10pt} Since $x_{a,1}$ and $x_{a,8}$ occur in $T_\ell \uhr \rho\,(\eta_1\mapsto\pi^1)\cdots (\eta_{\ell-1}\mapsto\pi^{\ell-1})$ it certainly must be the case that $\rho_{a,1} = \rho_{a,8} = \ast$; there may also be other coordinates $i\in [w]$ such that $\rho_{a,i} = \ast$ and $x_{a,i}$ does not occur in $T_\ell \uhr \rho\,(\eta_1\mapsto\pi^1)\cdots (\eta_{\ell-1}\mapsto\pi^{\ell-1})$ (coordinates $2$ through $7$ in our example above).  For $i\in [w]$ such that $\rho_{a,i} = \ast$ and $x_{a,i}$ occurs in $T_\ell \uhr \rho\,(\eta_1\mapsto\pi^1)\cdots (\eta_{\ell-1}\mapsto\pi^{\ell-1})$, the restriction $\sigma^\ell$ fixes $x_{a,i}$ so as to partially satisfy $T_\ell \uhr \rho\,(\eta_1\mapsto\pi^1)\cdots (\eta_{\ell-1}\mapsto\pi^{\ell-1})$: in our example above, $\sigma^{\ell}_{a,1} = 0$ (since $x_{a,1}$ occurs negatively) whereas $\sigma^{\ell}_{a,8} = 1$ (since $x_{a,8}$ occurs positively).  The remaining variables $x_{a,2},\ldots,x_{a,7}$ are set to $0$ by $\sigma^{\ell}$, yielding a completely fixed block $(\rho\sigma^\ell)_a \in \zo^{w}$ (Claim (v) in Section~\ref{sec:decode}). Intuitively, we ``break symmetry'' and set these variables to $0$ (rather than $1$) so that the decoder will be able to ``undo'' them in $(\rho\sigma^\ell)_a$ without any auxiliary information: since $\rho_a \in \{\ast,1\}^{w}\setminus \{1\}^w$, the decoder readily infers $\rho_{a,i} = \ast$ for all $i\in [w]$ such that $(\rho\sigma^\ell)_{a,i} = 0$.  And indeed, for the set $\gamma_\ell \sse \calX$ of variables $x_{a,i}$ that are set to $1$ by $\sigma^{\ell}$, we provide the decoder with the auxiliary information $\encode(\gamma_\ell)$ so that she is able to ``undo'' them in $(\rho\sigma^\ell)_a$.}
\label{figure:PSL-figure}
\end{figure}

\subsection{Decodability}
\label{sec:decode}

Let $\eta = \eta_1 \cup \cdots \cup\eta_j$, $\encode(\eta) = \encode(\eta_1) \cdots \encode(\eta_j) \in \zo^{s(1+\log r)}$, $\sigma = \sigma^1\cdots \sigma^j$, $\gamma = \gamma_1 \cup \cdots \cup\gamma_j$, $\encode(\gamma) = \encode(\gamma_1) \cdots \encode(\gamma_j) \in \zo^{rs}$, and $\pi' = \pi^1 \cdots \pi^j \in \zo^s$. Our main result in this subsection is the following proposition:

\begin{proposition}
\label{prop:injection}
The map $\theta : \calB \to \{0,1,\ast\}^{A\times [w]} \times \{0,1\}^s \times \{0,1\}^{s(1+\log r)} \times \{0,1\}^{rs},$
\[ \theta(\rho) = (\rho\sigma,\pi',\encode(\eta),\encode(\gamma)), \]
is an injection.
\end{proposition}

Before proving Proposition~\ref{prop:injection}, we state a slight extension of an observation made above in the definition of $\sigma^\ell$:

\begin{lemma}
\label{lem:does-not-falsify}
For all $1\le \ell \le j-1$ we have
\[ T_\ell \uhr  \rho\,(\eta_1\mapsto \pi^1) \cdots (\eta_{\ell-1}\mapsto \pi^{\ell-1})\,\sigma^\ell \cdots \sigma^j\equiv 1, \]
and when $\ell = j$ we have $T_j \uhr  \rho\,(\eta_1\mapsto \pi^1) \cdots (\eta_{j-1}\mapsto \pi^{j-1})\,\sigma^j \not\equiv 0$.
\end{lemma}
\begin{proof}
As we observed in the definition of $\sigma^\ell$ above (c.f.~(\ref{eq:satisfies-term})), we have that $\sigma^\ell$ is designed so that
\[ T_\ell \uhr  \rho\,(\eta_1\mapsto \pi^1) \cdots (\eta_{\ell-1}\mapsto \pi^{\ell-1})\,\sigma^\ell \equiv 1,\]
and certainly this remains true when the restriction is further extended by $\sigma^{\ell+1}\cdots \sigma^{j}$.  We do not necessarily have this property for $\ell = j$ due to our possible trimming of $\eta_j$ so that $\eta_1 \cup\cdots\cup\eta_j$ has cardinality exactly $s$; this results in a redefinition of $\sigma^j$ where some of its coordinates are set from $\{0,1\}$ back to $\ast$. However it is still the case that $\sigma^j$ partially satisfies $T_j \uhr  \rho\,(\eta_1\mapsto \pi^1) \cdots (\eta_{{j-1}}\mapsto \pi^{j-1})$, and hence $T_j \uhr  \rho\,(\eta_1\mapsto \pi^1) \cdots (\eta_{{j-1}}\mapsto \pi^{j-1})\,\sigma^j \not\equiv 0$.
\end{proof}

\begin{proof}[Proof of Proposition~\ref{prop:injection}]

We prove the proposition by describing a procedure that allows a ``decoder'' to uniquely obtain $\rho$ given $(\rho\sigma,\pi',\encode(\eta),\encode(\gamma)).$
Recall that $T_1$ is defined to be the first term in $F$ not falsified by $\rho$. By Lemma~\ref{lem:does-not-falsify}, this remains true when $\rho$ is extended by $\sigma$: that is, the first term $T_1'$ in $F$ such that $T_1' \uhr \rho\sigma \not\equiv 0$ is precisely $T_1$ itself.  Therefore, given $\rho\sigma$  the decoder is able to identify $T_1$ in $F$, and with $T_1$ in hand she is able to then use $\encode(\eta_1)$ and $\encode(\gamma_1)$ to recover $\eta_1$ and $\gamma_1$ respectively.  Next, she ``undoes'' $\sigma^1$ in $\rho\sigma = \rho\sigma^1\sigma^2\cdots \sigma^j$ and obtains $\rho \sigma^2 \cdots \sigma^j$ as follows: for every $y_a\in\eta_1$, she sets $(\rho\sigma)_{a,i}$ back to $\ast$ for all $i\in U_a$, where
\[ U_a = \{ i \in [w] \colon \text{$(\rho\sigma)_{a,i} = 0$ or $x_{a,i} \in \gamma_1$} \}. \]
To see that this indeed ``undoes'' $\sigma^1$, first recall that for every $y_a\in\eta_1$, the restriction $\sigma^1$ is defined so that $\sigma^1_{a,i} \in \zo$ iff $\rho_{a,i} = \ast$, and furthermore, $\sigma^1_{a,i} = 1$ iff $x_{a,i} \in \gamma_1$. (Recall the example in Figure~\ref{figure:PSL-figure}.)  Therefore, to obtain $\rho\sigma^2\cdots\sigma^j$ from $\rho\sigma^1\sigma^2\cdots \sigma^j$, for every $y_a \in\eta_1$ and $i\in [w]$ the decoder sets $(\rho\sigma)_{a,i}$ back to $\ast$ if either $(\rho\sigma)_{a,i} = 0$ or $x_{a,i} \in \gamma_1$.
Finally, using $\pi^1 \in \zo^{\eta_1}$ she constructs the hybrid restriction $\rho\,(\eta_1 \mapsto\pi^1)\, \sigma^2\cdots \sigma^j$.

By the same reasoning, for every $2 \le \ell \le j$ the decoder is able to iteratively recover $T_\ell,\eta_\ell,\gamma_\ell$, and $\pi^\ell$ from the hybrid restriction \[ \rho\,(\eta_1\mapsto \pi^1) \cdots (\eta_{\ell-1}\mapsto \pi^{\ell-1})\,\sigma^\ell \cdots \sigma^j. \]
 With this information she ``undoes'' $\sigma^\ell$ within $\rho\,(\eta_1\mapsto \pi^1) \cdots (\eta_{\ell-1}\mapsto \pi^{\ell-1})\,\sigma^\ell \cdots \sigma^j$, and constructs the next hybrid restriction \[ \rho\,(\eta_1\mapsto \pi^1) \cdots (\eta_\ell\mapsto \pi^{\ell})\,\sigma^{\ell+1} \cdots \sigma^j.\]
   Finally, having recovered $\rho\,(\eta_1\mapsto \pi^1) \cdots (\eta_j\mapsto \pi^{j})$ and $\eta = \eta_1 \cup \cdots \cup \eta_j$, the decoder will have all the information she needs to recover the actual restriction $\rho$: she sets $(\rho\,(\eta_1\mapsto \pi^1) \cdots (\eta_j\mapsto \pi^{j}))_{a,i}$ back to $\ast$ for every $y_a \in \eta$ and $i \in U_a$.
\end{proof}

\subsection{Proof of Proposition~\ref{prop:psl-typical-tau}}
\label{sec:weight-diff}

For all possible outcomes $\vartheta_2, \vartheta_3, \vartheta_4$ of the second, third, and fourth coordinates of the map $\theta$ defined in Proposition~\ref{prop:injection}, we define
\begin{align*}
 \calB_{\vartheta_2,\vartheta_3} &= \{\rho\in\calB \colon \theta_2(\rho) = \vartheta_2, \theta_3(\rho) = \vartheta_3 \} \sse \calB. \\
 \calB_{\vartheta_2,\vartheta_3,\vartheta_4} &= \{\rho\in\calB \colon \theta_2(\rho) = \vartheta_2, \theta_3(\rho) = \vartheta_3, \theta_4(\rho) = \vartheta_4 \} \sse \calB_{\vartheta_2,\vartheta_3}.
\end{align*}
We begin by bounding the probability that $\brho \leftarrow \calR(\tau)$ belongs to
$\calB_{\vartheta_2,\vartheta_3, \vartheta_4}$ for a fixed tuple $(\vartheta_2,\vartheta_3,\vartheta_4)$.  The following fact, giving the probability mass function of $\calR(\tau)$, will be useful for us (its proof is by inspection of Definition \ref{def:main-projection}):

\begin{fact}
\label{fact:weights}
Fix $\tau \in \{0,1,\ast\}^{A_k}$, and  write $S_a = S_a(\tau)$ to denote $\tau_a^{-1}(\ast) = \{ i \in [w_{k-1}] \colon \tau_{a,i} = \ast\}$.  Then $\Prx_{\brho\leftarrow\calR(\tau)}[\brho = \rho] = \xi(\rho)$ for all $\rho\in \{0,1,\ast\}^{A_k}$, where $\xi : \{0,1,\ast\}^{A_k} \to [0,1]$ is the probability mass function:
\[  \xi(\rho) = \mathop{\prod_{a\in A_{k-1}}}_{S_a \ne \emptyset} \zeta_a(\rho(S_a)), \]
and $\rho(S_a)$ denotes the substring of $\rho_a$ with coordinates in $S_a$, and $\zeta_a : \{0,1,\ast\}^{S_a} \to [0,1]$ is the probability mass function:
\begin{align*}
  \zeta_a(\varrho)
  &=
  \begin{cases}
       \ds
     \lambda
     &\text{if } \varrho = \{1\}^{S_a},\\
     \ds
     q_a \cdot \frac{{t_k}^{|\varrho^{-1}(\ast)|}(1-t_k)^{|\varrho^{-1}(1)|}}{1-(1-t_k)^{|S_a|}}
     &\text{if } \varrho \in \{\ast,1\}^{S_a} \setminus \{1\}^{S_a},\\
     \ds
(1-\lambda-q_a) \cdot  \frac{{t_k}^{|\varrho^{-1}(0)|}(1-t_k)^{|\varrho^{-1}(1)|}}{1-(1-t_k)^{|S_a|}}
     &\text{if } \varrho \in \{0,1\}^{S_a} \setminus \{1\}^{S_a}.
  \end{cases}
\end{align*}
\end{fact}

\begin{lemma}
\label{lem:fix-other-outcomes}
For all $\vartheta_2,\vartheta_3,\vartheta_4$,
\[ \Prx_{\brho\leftarrow\calR(\tau)}\big[\brho \in \calB_{\vartheta_2,\vartheta_3,\vartheta_4}\big] = \left(O\big(w^{-1/4}\big)\right)^s  \left(\frac{t_k}{1-t_k}\right)^{\| \vartheta_4 \|}, \]
where $\| \vartheta_4 \|$ denotes $|\vartheta_4^{-1}(1)|$, the Hamming weight of  $\vartheta_4$.
\end{lemma}

\begin{proof}
Fix $\rho \in \calB_{\vartheta_2,\vartheta_3,\vartheta_4}$. The restrictions $\rho$ and $\theta_1(\rho) = \rho\sigma$ differ in exactly $s$ blocks: these are the blocks $a\in A_{k-1}$ such that $y_a \in \eta$.  Consider any such $a \in A_{k-1}$, and recall (as observed in the definition of $\sigma$) that $S_a$ is $k$-acceptable and $\rho(S_a) \in \{\ast,1\}^{S_a} \setminus \{1\}^{S_a}$ whereas $(\rho\sigma)(S_a) \in \{0,1\}^{S_a}$. Let $\Delta_a$ denote  $|(\rho\sigma)_a^{-1}(1)| - |\rho_a^{-1}(1)|$, the number of ``new 1's'' that $\sigma$ introduces into block $a$ (note that as observed earlier we have that $\Delta_a \geq 0$). By Fact~\ref{fact:weights}, we have that
\begin{equation} \frac{\zeta_a((\rho\sigma)(S_a))}{\zeta_a(\rho(S_a))} =
\left\{
\begin{array}{cl}
\ds \frac{\lambda}{q_a} \cdot \frac{1-(1-t_k)^{|S_a|}}{t_k^{\Delta_a}(1-t_k)^{|S_a|-\Delta_a}}
 & \text{if $(\rho\sigma)(S_a) = \{1\}^{S_a}$}  \vspace{5pt} \\
  \ds  \frac{1-\lambda-q_a}{q_a}  \bigg(\frac{1-t_k}{t_k}\bigg)^{\Delta_{a}}  & \text{if $(\rho\sigma)(S_a) \in \zo^{S_b} \setminus \{1^{S_a}\}$}.
\end{array}
\right.
\label{eq:weight-ratio}
\end{equation}
Since $S_a$ is $k$-acceptable, we have that $|S_a| = qw \pm w^{\gradual(k,d)}$ and therefore
\begin{align*} (1-t_k)^{|S_a|} &\le \frac{(1-t_k)^{qw}}{(1-t_k)^{w^{\gradual(k,d)}}} \\
&= \frac{qt_{k-1} + \lambda}{(1-t_k)^{w^{\gradual(k,d)}}}  \\
&\le \frac{qt_{k-1} + \lambda}{1-t_kw^{\gradual(k,d)}}\ \le\ 2q^2,
\end{align*}
where the equality is by (\ref{eq:def-of-tk}) and the final inequality uses Lemma~\ref{lemma:tk-bound}, (\ref{eq:def-of-lambda-and-q}) and (\ref{eq:estimates}).  Since $q_a \le 2q$ by Lemma~\ref{lem:bound-on-qa}, we may lower bound the quantity in the first line of (\ref{eq:weight-ratio}) by
\[ \frac{\lambda}{8q^3} \left(\frac{1-t_k}{t_k}\right)^{\Delta_a} = \Omega(w^{1/4}) \left(\frac{1-t_k}{t_k}\right)^{\Delta_a},  \]
where we have used our choice of $\lambda$ in (\ref{eq:def-of-lambda-and-q}) and the estimates (\ref{eq:estimates}).  Similarly, for the second quantity in the second line of (\ref{eq:weight-ratio}) we have the lower bound
\[
\frac{1-\lambda-q_a}{q_a}  \bigg(\frac{1-t_k}{t_k}\bigg)^{\Delta_{a}} = \Omega\left(\sqrt{\frac{w}{\log w}}\right)\left(\frac{1-t_k}{t_k}\right)^{\Delta_a}
\]
and so in both cases we may lower bound the ratio in (\ref{eq:weight-ratio}) by
\[\frac{\zeta_a((\rho\sigma)(S_a))}{\zeta_a(\rho(S_a))} = \Omega\big(w^{1/4}\big)  \left(\frac{1-t_k}{t_k}\right)^{\Delta_a}.
 \]
\ignore{
\gray{
Since $\tau$ is typical, we have that $|S_a| = qw_\star \pm w^{\gradual(k,d)}$, and so by (\ref{eq:equivalent-def-of-q}) we have the bound
\[(1-t)^{|S_a|} \le \frac{p-\lambda}{(1-t)^{w^{\gradual(k,d)}}} \le \frac{p}{1-t w^{\gradual(k,d)}}  \le 2p.
 \]
By (\ref{eq:def-of-qa}), this in turn implies that
\[ q_a \le \frac{2p - \lambda}{t} < \frac{2p}{t}, \]
and so using the estimates (\ref{eq:estimates})  we may lower bound the quantity in the first line of (\ref{eq:weight-ratio}) by

\begin{align*}
\frac{\lambda t}{8p^2} \left(\frac{1-t}{t}\right)^{\Delta_a} &= \Omega\big(w^{1/4}\big)  \left(\frac{1-t}{t}\right)^{\Delta_a}.
\end{align*}
Similarly, for the second quantity in the second line of (\ref{eq:weight-ratio}) we have the lower bound
\[
\frac{1-\lambda-q_a}{q_a}  \bigg(\frac{1-t}{t}\bigg)^{\Delta_{a}} = \Omega\left({\sqrt{w}}\right)\left(\frac{1-t}{t}\right)^{\Delta_a} %\gg \Omega\left(\frac1{\sqrt{t} \log w_\star}\right) \left(\frac1{\sqrt{t}}\right)^{\Delta_b}.
\]
and so in both cases we may lower bound the ratio in (\ref{eq:weight-ratio}) by
\[\frac{\zeta_a((\rho\sigma)_a)}{\zeta_a(\rho_a)} = \Omega\big(w^{1/4}\big)  \left(\frac{1-t}{t}\right)^{\Delta_a}.
 \]
 }
 }

 Since $\sum_{a\colon \rho_a \ne (\rho\sigma)_a} \Delta_a = \| \vartheta_4 \|$, it follows from Fact~\ref{fact:weights} that
\begin{equation} \frac{\xi(\theta_1(\rho))}{\xi(\rho)} =  \frac{\xi(\rho\sigma)}{\xi(\rho)} = \mathop{\prod_{a\in A_{k-1}}}_{S_a \ne \emptyset} \frac{\zeta_a((\rho\sigma)(S_a))}{\zeta_a(\rho(S_a))} =  \left( \Omega\big(w^{1/4}\big)\right)^s  \left(\frac{1-t_k}{t_k}\right)^{\| \vartheta_4 \| }.\label{eq:weight-increase}
\end{equation}
Finally, summing over all $\rho \in \calB_{\vartheta_2,\vartheta_3,\vartheta_4}$  we conclude that
\begin{align*}
 \Prx_{\brho\leftarrow\calR(\tau)} \big[\brho\in\calB_{\vartheta_2,\vartheta_3,\vartheta_4}\big] = \sum_{\rho\in \calB_{\vartheta_2,\vartheta_3,\vartheta_4}} \xi(\rho) & = \left(O\big(w^{-1/4}\big)\right)^s  \left(\frac{t_k}{1-t_k}\right)^{\| \vartheta_4 \| }\sum_{\rho\in \calB_{\vartheta_2,\vartheta_3,\vartheta_4}}  \xi(\theta_1(\rho)) \\
 &=  \left(O\big(w^{-1/4}\big)\right)^s \left(\frac{t_k}{1-t_k}\right)^{\| \vartheta_4 \|}.
  \end{align*}
Here the first inequality is by (\ref{eq:weight-increase}), and the second uses the fact that $\theta$ is an injection (Proposition~\ref{prop:injection}), and hence any two distinct $\rho,\rho' \in \calB_{\vartheta_2,\vartheta_3,\vartheta_4}$ map to distinct $\theta_1(\rho), \theta_1(\rho') \in \{0,1,\ast\}^{A\times [w]}$, so $\sum_{\rho\in \calB_{\vartheta_2,\vartheta_3,\vartheta_4}}  \xi(\theta_1(\rho))$ is at most 1 since $\xi$ is a probability mass function.
\end{proof}

Proposition~\ref{prop:psl-typical-tau} follows as a straightforward consequence of Lemma~\ref{lem:fix-other-outcomes}:

\begin{proof}[Proof of Proposition~\ref{prop:psl-typical-tau}]
Summing over all $\vartheta_4 \in \zo^{rs}$ and stratifying according to Hamming weight, we have that \begin{align*}
\Prx_{\brho\leftarrow\calR(\tau)} [\brho\in\calB_{\vartheta_2,\vartheta_3}] &= \sum_{i=0}^{rs} \mathop{\sum_{\vartheta_4\in \zo^{rs}}}_{\|\vartheta_4\|=i} \Prx_{\brho\leftarrow\calR(\tau)}[\brho\in \calB_{\vartheta_2,\vartheta_3,\vartheta_4}] \\
&\le \sum_{i=0}^{rs} {rs\choose i}   \left(\frac{t_k}{1-t_k}\right)^{i}\left(O\big(w^{-1/4}\big)\right)^s\\
&= \left(1+\frac{t_k}{1-t_k}\right)^{rs} \left(O\big(w^{-1/4}\big)\right)^s = 
\left(O\Big(e^{rt_k/(1-t_k)} \cdot w^{-1/4}\Big)\right)^s.
\end{align*}
Taking a union bound over all $2^s$ possible $\vartheta_2\in \zo^s$ and $(2r)^s$ possible $\vartheta_3 \in \zo^{s(1+\log r)}$ completes the proof.
\end{proof}

\paragraph{Proof of Proposition~\ref{prop:psl-initial}.}

For Proposition~\ref{prop:psl-initial}, we first observe that Proposition~\ref{prop:injection} also holds for $\brho\leftarrow\calR_\init$ (the proof is completely identical, with $\tau$ being the trivial restriction $\{\ast\}^{n}$).
Proposition~\ref{prop:psl-initial} then follows as a consequence of Proposition~\ref{prop:injection} in a very similar manner (the calculations are in fact significantly simpler); we point out the essential differences in this section. We begin with the following analogue of Fact~\ref{fact:weights}, specifying the probability mass function of $\calR_\init$ (like Fact~\ref{fact:weights}, its proof is by inspection of Definition~\ref{def:calR-init}):
\begin{fact}
\label{fact:init-weights}
$\Prx_{\brho\leftarrow\calR_\init} [\brho = \rho] = \xi(\rho)$ for all $\rho\in \{0,1,\ast\}^{A_{d-1}\times [m]}$ (recall that $w_{d-1} = m$), where $\xi : \{0,1,\ast\}^{A_{d-1} \times [m]} \to [0,1]$ is the probability mass function:
\[  \xi(\rho) = \prod_{a\in A_{d-1}}\zeta(\rho_a), \]
and $\zeta : \{0,1,\ast\}^{m} \to [0,1]$ is the probability mass function:
\begin{align*}
  \zeta(\varrho)
  &=
  \begin{cases}
       \ds
     \lambda
     &\text{if } \varrho = \{1\}^{m},\\
     \ds
     q \cdot \frac{p}{1-p}
     &\text{if } \varrho \in \{\ast,1\}^{m} \setminus \{1\}^{m},\\
     \ds
(1-\lambda-q) \cdot  \frac{p}{1-p}
     &\text{if } \varrho \in \{0,1\}^{m} \setminus \{1\}^{m}.
  \end{cases}
\end{align*}
\end{fact}
Fact~\ref{fact:init-weights} gives us the following analogue of (\ref{eq:weight-ratio}):
\[  \frac{\zeta((\rho\sigma)_a)}{\zeta(\rho_a)} =
\left\{
\begin{array}{cl}
\ds \frac{\lambda (1-p)}{qp}  & \text{if $(\rho\sigma)_a = \{1\}^{m}$}  \vspace{5pt} \\
  \ds  \frac{1-\lambda-q}{q}  & \text{if $(\rho\sigma)_a \in \zo^{m} \setminus \{1\}^{m}$},
\end{array}
\right. \]
and so by our choice of $\lambda$ in (\ref{eq:def-of-lambda-and-q}) and our estimates (\ref{eq:estimates}) this ratio is always at least $\Omega\big(w^{1/4}\big)$. (Unlike the proof of Lemma~\ref{lem:fix-other-outcomes}, our lower bound here does not depend on $\Delta_a = |(\rho\sigma)_a^{-1}(1)|-|\rho_a^{-1}(1)|$.) By the same calculations as in the proof of Lemma~\ref{lem:fix-other-outcomes}, we have the following analogue of Lemma~\ref{lem:fix-other-outcomes}:

\begin{lemma}
\label{lem:fix-other-outcomes-init}
For all $\vartheta_2,\vartheta_3,\vartheta_4$, we have that $\ds \Prx_{\brho\leftarrow\calR_\init}\big[\brho \in \calB_{\vartheta_2,\vartheta_3,\vartheta_4}\big] = \left( O\big(w^{-1/4}\big)\right)^s.$
\end{lemma}

Proposition~\ref{prop:psl-initial} follows by a union bound over all $2^s$ possible $\vartheta_2\in \zo^s$, $(2r)^s$ possible $\vartheta_3 \in \zo^{s(1+\log r)}$, and $2^{rs}$ possible $\vartheta_4 \in \zo^{rs}$ (unlike in the proof of Proposition~\ref{prop:psl-typical-tau} we do not have to stratify the union bound over $\vartheta_4 \in\{0,1\}^{rs}$ according to Hamming weight).

%\newpage

\subsection{Approximator simplifies under random projections} \label{sec:two}

The main results of this section are Theorems~\ref{thm:bounded-bottom-fan-in} and~\ref{thm:different-alternation}.  The first of these theorems says that any depth-$d$ circuit whose size is not too large and whose bottom fan-in is significantly smaller than that of $\BalancedSipser_d$ will collapse to a shallow decision tree with high probability under the random projection $\mathbf{\Psi}$ from Definition~\ref{def:boldpsi}:

\begin{theorem}
\label{thm:bounded-bottom-fan-in}
For $2 \leq d \leq {\frac {c\sqrt{\log n}}{\log \log n}}$, let $C : \zo^{n} \to \zo$ be a depth-$d$ circuit with bottom fan-in at most ${\frac {\log n}{10(d-1)}}$ and size $S \le 2^{n^{{\frac 1 {6(d-1)}}}}$.  Then $\mathbf{\Psi}(C)$ is computed by a decision tree of depth $n^{{\frac 1 {4(d-1)}}}$ with probability $1-\exp\big(-\Omega\big(n^{\frac 1 {{6}(d-1)}}\big)\big)$.
\end{theorem}

The second theorem is quite similar; it says that under the random projection $\mathbf{\Psi}$, any depth-$d$ circuit $C$ that is not too large, regardless of its bottom fan-in, will collapse to a depth-2 circuit with bounded bottom fan-in and with top gate matching that of $C$:

\begin{theorem}
\label{thm:different-alternation} For $2\leq d \leq  {\frac {c\sqrt{\log n}}{\log \log n}}$, let $C : \zo^{n} \to \zo$ be a depth-$d$ circuit of size $S \le 2^{{\frac 1 2}n^{{\frac 1 {6(d-1)}}}}$ and unbounded bottom fan-in.
\begin{enumerate}
\item If the top gate of $C$ is an $\AND$, then $\mathbf{\Psi}(C)$ is $(1/S)$-close (with respect to the uniform distribution on $\{0,1\}^n$) to a width-$n^{{\frac 1 {4(d-1)}}}$ CNF with probability $1-\exp\big(-\Omega\big(n^{\frac 1 {{6}(d-1)}}\big)\big)$.
\item If the top gate of $C$ is an $\OR$, then $\mathbf{\Psi}(C)$ is $(1/S)$-close to a width-$n^{{\frac 1 {4(d-1)}}}$ DNF with probability  $1-\exp\big(-\Omega\big(n^{\frac 1 {{6}(d-1)}}\big)\big)$.
\end{enumerate}
\end{theorem}

We first prove Theorem \ref{thm:bounded-bottom-fan-in}, which deals with depth-$d$ circuits with bounded bottom fan-in.  We state the following simple lemma explicitly for convenience of later reference:
\begin{lemma}
\label{lem:combine}
Suppose that $3 \leq d \leq {\frac {c\log w}{\log \log w}}$.
For $2 \le k \le d-1$ and $\ell\in \N$, let $C : \zo^{A_{k+1}} \to\zo$ be a size-$S$ depth-$\ell$ circuit with bottom fan-in $w^{1/5}$.\ignore{ Let $\tau \in \{\bullet,\circ,\ast\}^{A_{k+1}}$ be typical.} For any $\tau \in \{\bullet,\circ,\ast\}^{A_{k+1}}$, with probability at least $1-S\cdot 4^{-w^{1/5}}\ignore{ - e^{-\Omega(w^{1/6})}}$ over $\brho\leftarrow\calR(\tau)$, we have that $\proj_{\brho}\,C$ is a depth-$(\ell-1)$ circuit with  bottom fan-in $w^{1/5}$, and has the same number of gates at distance at least two from the input variables as~$C$. 
\end{lemma}

\begin{proof}
The lemma follows from applying Proposition~\ref{prop:psl-typical-tau} with $r = s = w^{1/5}$ and a union bound over all gates of $C$ (at most $S$ many) that are at distance 2 from the input variables.
\end{proof}

The following proposition directly implies Theorem \ref{thm:bounded-bottom-fan-in} by straightforward translation of parameters, recalling (\ref{eq:nversusd}):

\begin{proposition}
\label{prop:bounded-bottom-fan-in}
For $2 \leq d \leq {\frac {c\sqrt{\log n}}{\log \log n}}$,
let $C : \zo^{A_{d}} \to \zo$ be a depth-$d$ circuit with bottom fan-in $\frac1{5} m$ and size $S \le 2^{w^{1/5}}$.  Then $\mathbf{\Psi}(C)$ is computed by a depth-$(w^{1/5})$ decision tree with probability $1-e^{-\Omega(w^{1/{5}})}$.  \end{proposition}

\begin{proof}
Applying Proposition~\ref{prop:psl-initial} with $r =  \frac1{5} m$ and $s = w^{1/5}$ to each of the bottom-layer gates of $C$,  we have that $\proj_{\brho^{(d)}}\,C$ is a depth-$(d-1)$ circuit with bottom fan-in $w^{1/5}$ with probability at least $1- S\cdot 4^{-w^{1/5}} \ge 1-2^{-w^{1/5}}$ over $\brho^{(d)}\leftarrow\calR_{\init}$.  If $d=2$, we observe that in fact Proposition~\ref{prop:psl-initial} gives us that $\proj_{\brho^{(d)}}\,C$ is a decision tree of the desired depth, and we are done.  If $d \geq 3$,\ignore{we observe that furthermore, by Proposition~\ref{prop:initial-typical} we have that $\hat{\brho^{(d)}}$ is typical with probability $1 - e^{-\Omega(w^{1/6})}$.  T} the claim follows by a union bound over $d-2$ applications of Lemma~\ref{lem:combine} (where we observe from the proof of Lemma~\ref{lem:combine} that in the last application of Lemma~\ref{lem:combine} we may conclude that $\mathbf{\Psi}(C)$ is in fact a decision tree of depth $w^{1/5}).$
\end{proof}

Next we turn to Theorem \ref{thm:different-alternation}.  We require the following standard lemma showing that any circuit can be ``trimmed'' to reduce its bottom fan-in while changing its value on only a few inputs:

\begin{lemma}
\label{lem:truncate}
Let $C : \zo^n \to \zo$ be a circuit and let $\eps > 0$. There exists a circuit $C' : \zo^n\to\zo$ such that
\begin{enumerate}
\item The size and depth of $C'$ are both at most that of $C$;
\item The bottom fan-in of $C'$ is at most $\log(S/\eps)$;
\item $C$ and $C'$ are $\eps$-close with respect to the uniform distribution.
\end{enumerate}
\end{lemma}

\begin{proof}  $C'$ is obtained from $C$ by replacing each bottom-level $\AND$ ($\OR$, respectively) gate whose fan-in is too large with 0 (1, respectively).   Each such gate originally takes its minority value on at most an $\eps/S$ fraction of all inputs so the lemma follows from a union bound.
\end{proof}

The following proposition directly implies Theorem \ref{thm:different-alternation} (by straightforward translation of parameters):

\begin{proposition}
\label{prop:different-alternation}
For $2 \leq d \leq  {\frac {c\sqrt{\log n}}{\log \log n}}$,
let $C : \zo^{A_d} \to \zo$ be a depth-$d$ circuit of size $S \le 2^{\frac1{2}w^{1/5}}$ and unbounded bottom fan-in.
\begin{enumerate}
\item If the top gate of $C$ is an $\AND$, then $\mathbf{\Psi}(C)$ is $(1/S)$-close to a width-$(w^{1/5})$ CNF with probability $1-e^{-\Omega(w^{1/5})}$.
\item If the top gate of $C$ is an $\OR$, then $\mathbf{\Psi}(C)$ is $(1/S)$-close to a width-$(w^{1/5})$ DNF with probability $1-e^{-\Omega(w^{1/5})}$.
\end{enumerate}
\end{proposition}

\begin{proof}
By symmetry it suffices to prove the first claim. Applying Lemma~\ref{lem:truncate} with $\eps = 1/S$, we have that $C$ is $(1/S)$-close to a circuit $C' : \zo^{A_d} \to \zo$ of size and depth at most that of $C$, and with bottom fan-in $\log(S/\eps) = 2 \log(S) \le w^{1/5}$. Certainly the size, depth, and bottom fan-in of $\proj_{\brho^{(d)}}\,C'$ is at most that of $C'$ with probability $1$ over the randomness of $\brho^{(d)} \leftarrow\calR_{\init}$ (note that unlike in the proof of Proposition~\ref{prop:bounded-bottom-fan-in}, we do not argue that the depth of $C'$ decreases by one under an $\calR_{init}$-random projection; the bottom fan-in of $C'$ is too large for us to apply Proposition~\ref{prop:psl-initial}).  If $d=2$ then this already gives the result (in fact with no failure probability).  If $d \geq 3,$ the proposition then follows by a union bound over $d-2$ applications of Proposition~\ref{lem:combine}.
\end{proof}

%\newpage

\section{$\BalancedSipser$ retains structure under random projections} \label{sec:combo}

Now we turn our attention to the randomly projected target $\mathbf{\Psi}(\BalancedSipser_d)$. As discussed in Section~\ref{sec:second-outline}, we would like to establish Property 2 showing that $\BalancedSipser_d$ ``retains structure'' under a $\mathbf{\Psi}$-random projection: with high probability over $\mathbf{\Psi}$, the randomly projected target $\mathbf{\Psi}(\BalancedSipser_d)$ is a depth-one formula whose bias remains very close to $1/2$ (with respect to an appropriate product distribution over $\zo^{w_0}$). This is necessarily a high-probability statement; to establish it, we must account for the failure probabilities introduced by each of the $d-1$ individual random projections $\proj_{\brho^{(k)}}$ that comprise $\mathbf{\Psi} \equiv \{ \brho^{(k)}\}_{k\in \{2,\ldots, d\}}$.\footnote{\label{footnote:sec9}As a concrete example of a failure event, consider an outcome $\rho^{(d)} \in \supp(\calR_\init) \equiv \{0,1,\ast\}^{A_{d-1}\times [m]}$ which is such that $(\rho^{(d)}_b)^{-1}(0)$ is nonempty  for all $b\in A_{d-1}$.   In this case
 \[ \proj_{\rho^{(d)}}\,\BalancedSipser_d  \equiv \proj\, (\BalancedSipser_d \uhr \brho^{(d)}) \equiv 0 \]
(recall that the bottom-level gates of $\BalancedSipser_d$ are $\AND$ gates), and our target function is set to the constant 0 already after the first $\calR_\init$-random projection.} To reason about these failure probabilities and carefully account for them, in Section \ref{sec:typical-stuff} we introduce the notion of a ``typical'' restriction and prove some useful properties about how typicality interacts with our random projections.  In
Section \ref{sec:sipser-survives} we use these properties to establish the main results of this section, that $\BalancedSipser_d$ ``retains structure'' when it is hit with the random projection $\mathbf{\Psi}.$

\subsection{Typical restrictions}
\label{sec:typical-stuff}

Recalling the $\bullet,\circ$ notation from Table~\ref{table:circ-bullet}, we begin with the following definition:

\begin{definition} \label{def:typical}
Let $\tau \in \{\bullet,\circ,\ast\}^{A_k}$ where $2\le k \le d-1$.  We say that $\tau$ is \emph{typical} if it satisfies:
\begin{enumerate}
\item For every $a\in A_{k-1}$ the set $\tau^{-1}_a(\ast) \sse [w_{k-1}]$ is $k$-acceptable, where we recall from Definition~\ref{def:acceptable} that this means
\[ |\tau^{-1}_a(\ast)| = qw \pm w^{\gradual(k,d)} \quad \text{where $\gradual(k,d) := \frac1{3} + \frac{d-k-1}{12d}$}. \] (Note  that $\frac1{3} \le \gradual(k,d) \le \frac{5}{12} < \frac1{2}$ for all $d\in \N$ and $2\le k \le d-1$.) We observe that by Definition \ref{def:lift}, this condition implies that for every $\alpha \in A_{k-2}$, we have
\begin{equation}
\hat{\tau}_\alpha \in \{ \ast, \circ \}^{w_{k-2}}.  \label{eq:old-prop-2-of-typicality}
\end{equation}
\item For every $\alpha \in A_{k-2}$,
\[ |(\hat{\tau}_\alpha)^{-1}(\ast)| \ge w_{k-2}- w^{4/5}. \]
\end{enumerate}
We note that (\ref{eq:old-prop-2-of-typicality}) and Condition (2)  together imply that
\[ \hat{\hat{\tau}}_\alpha  = \ast  \quad \text{for all $\alpha\in A_{k-2}$.} \]
\end{definition}

See Figure~\ref{figure:the-figure} on the next page for an illustration of a typical $\tau$.  The rationale behind Definition~\ref{def:typical} is that projections $\proj_{\rho}$ such that $\hat\rho$ is typical have a very limited (and well-controlled) effect on the target $\BalancedSipser_d$:  roughly speaking, these projections ``wipe out'' the bottom-level gates of the formula (reducing its depth by one), ``trim'' the fan-ins of the next-to-bottom-level gates from $w$ to approximately $qw = \tilde{\Theta}(\sqrt{w})$, but otherwise essentially preserves the rest of the structure of the formula. We give a precise description in Section~\ref{sec:sipser-survives}; see Remark~\ref{rem:typical-restrictions}.
%\makeatletter
%\afterpage{\global\setlength\@fpsep{8\p@ \@plus 2fil}}
%\makeatother

\begin{figure}[p!]
\begin{center}
\includegraphics[width=16.5cm]{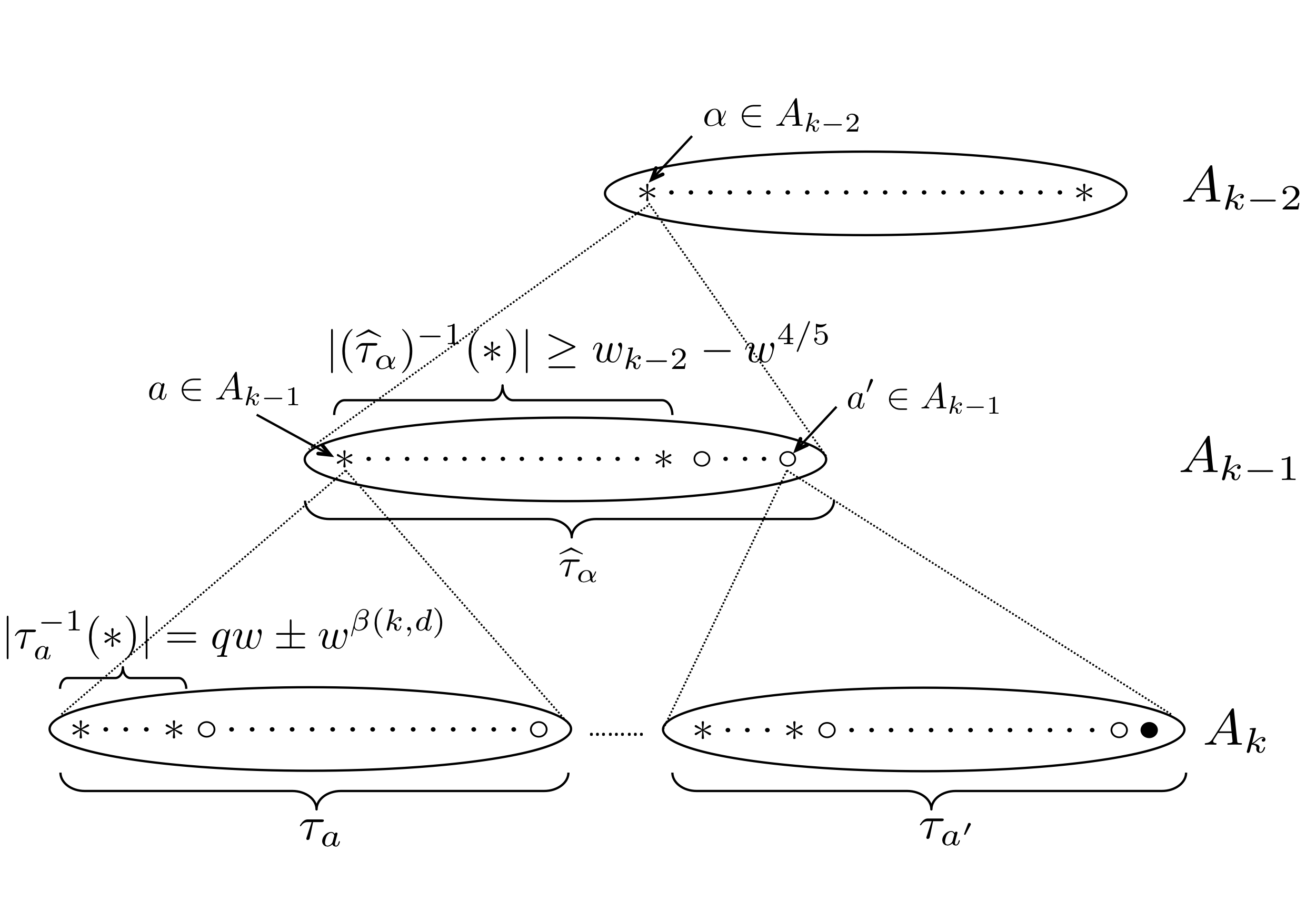}
\end{center}
\caption{The figure illustrates a typical $\tau \in \{\bullet,\circ,\ast\}^{A_{k}}$.   For $a \in A_{k-1}$,  $\tau_a$ is a block of length $w_{k-1}$,  i.e.\ a string in $\{\bullet,\circ,\ast\}^{w_{k-1}}$. We may think of the block $\tau_a$ as being located at level $k$.  By Condition (1) of Definition \ref{def:typical}, for every $a \in A_{k-1}$ we have that $|\tau_a^{-1}(\ast)|$, the number of $\ast$'€™s in $\tau_a$, is roughly $qw = \tilde{\Theta}(\sqrt{w})$.
The lift $\hat{\tau}$ of $\tau$ is a string in $\{\bullet,\circ,\ast\}^{A_{k-1}}$, and for $\alpha \in A_{k-2}$, $\hat{\tau}_\alpha$ is a block of length $w_{k-2}$. We may think of the block $\hat{\tau}_\alpha$ as being located at level $k-1$.  As stipulated by (\ref{eq:old-prop-2-of-typicality}), for every $\alpha \in A_{k-2}$, the string $\hat{\tau}_\alpha$ belongs to $\{\ast,\circ\}^{w_{k-2}}$.  By Condition (2) of Definition \ref{def:typical}, for every $\alpha \in A_{k-2}$, we have that $|(\hat{\tau}_\alpha)^{-1}(\ast)|$, the number of $\ast$'€™s in $\hat{\tau}_\alpha$, is at least $w_{k-2} - w^{4/5} = w_{k-2}(1-o(1))$.  Finally, we observe that (\ref{eq:old-prop-2-of-typicality}) and Condition~(2) of Definition \ref{def:typical} imply that $\hat{\hat\tau}_\alpha = \ast$ for every $\alpha \in A_{k-2}.$
}
\label{figure:the-figure}
\end{figure}

To prove that $\mathbf{\Psi}(\BalancedSipser_d)$ is a well-structured formula with high probability over the random choice of $\mathbf{\Psi} \equiv \{ \brho^{(k)}\}_{k\in \{2,\ldots, d\}}$, we will in fact establish the stronger statement showing that with high probability, every single one of the individual random projections $\proj_{\brho^{(k)}}$ only has a limited and well-controlled effect (in the sense described above) on the structure of $\BalancedSipser_d$. By Definition~\ref{def:typical}, this amounts to showing that the lifts $\hat{\brho^{(d)}},\ldots,\hat{\brho^{(2)}}$ associated with the $d-1$ individual projections comprising $\mathbf{\Psi}$ are \emph{all} typical with high probability.  We prove this inductively:  we first show that for $\brho^{(d)}\leftarrow \calR_\init$ its lift $\hat{\brho^{(d)}}$ is typical with high probability (Proposition~\ref{prop:initial-typical}), and then argue that if $\rho^{(k+1)}$ is typical then the lift $\hat{\brho^{(k)}}$ of $\brho^{(k)}\leftarrow\calR(\hat{\rho^{(k+1)}})$ is also typical with high probability (Proposition~\ref{prop:typical-yields-typical}).  The parameters of Definition~\ref{def:typical} are chosen carefully so that it ``bootstraps'' in the sense of Proposition~\ref{prop:typical-yields-typical}; in particular, this is the reason why we allow more and more deviation from $qw$ in Condition 1 as $k$ gets smaller (closer to the root).

Our two main results in this subsection are the following:

\begin{proposition}[Establishing initial typicality]\label{prop:initial-typical}
Suppose that $3 \leq d \leq {\frac {c\log w}{\log \log w}}$ for a sufficiently small absolute constant $c>0.$   Then
\[
\Prx_{\brho\leftarrow\calR_\init}[\hat\brho \text{~is typical}] \geq 1-e^{\tilde{\Omega}(w^{1/6})}.
\]
\end{proposition}

\begin{proposition}[Preserving typicality]\label{prop:typical-yields-typical}
Suppose that $3 \leq d \leq {\frac {c\log w}{\log \log w}}$ for a sufficiently small absolute constant $c>0.$  Let $2 \leq k \leq d-1$ and let $\tau \in \{\bullet,\circ,\ast\}^{A_{k+1}}$ be typical.  Then
\[
\Prx_{\brho\leftarrow \calR(\tau)}[\hat\brho \text{~is typical}] \geq  1 - e^{-\Omega(w^{1/6})}.
\]
\end{proposition}
\ignore{\lnote{With the new setting of $q$ I think we suffer a $\tilde\Omega$ loss in the first bound, but not the second}}

\subsubsection{Establishing initial typicality:  Proof of Proposition \ref{prop:initial-typical}} \label{sec:pf-initial-typical}

For notational brevity, throughout this subsubsection we write $\btau$ to denote $\hat{\brho} \in \{0,1,\ast\}^{A_{d-1}}$ where $\brho \leftarrow \calR_{\init}$.  We proceed to establish the two conditions of Definition \ref{def:typical}.

\begin{lemma} [Condition (1) of typicality] \label{lemma:condition1-initial}
Fix $a\in A_{d-2}$. Then
\[ \Pr\big[|\btau_a^{-1}(\ast)| = qw \pm  w^{1/3}\big] \ge 1- e^{-{\tilde{\Omega}}(w^{1/6})}. \]
\end{lemma}

\begin{proof}
Recalling (\ref{eq:initial}), we have that
\[ \Pr[\btau_{a,i} = \ast] = q \quad \text{independently for all $i \in [w]$}.\]
We shall apply Fact~\ref{fact:chernoff} with
\[ \bS = \bZ_1 + \cdots + \bZ_{w} \quad \text{where $\bZ_i \leftarrow \{ 0_{1-q}, 1_q\}$} \quad \text{(so $\mu = \E[\bS]$ is $qw$)}, \]
and $\gamma$ such that $\gamma\mu = w^{1/3}$.  Observe that since $\mu=qw=\Theta((w\log w)^{1/2})$, we have $\gamma=\Theta(w^{-1/6} (\log w)^{-1/2})$.  Hence by Fact~\ref{fact:chernoff}
 we have that
 \begin{equation*}
  \Pr\big[\big||\btau^{-1}_a(\ast)|- qw\big| > w^{1/3} \big]
  \le \exp\left( -\Omega\big(\gamma^2\mu\big)\right) \\
  = \exp\big(-\tilde{\Omega}\big(w^{1/6}\big)\big). \qedhere \end{equation*}
% \begin{align*}
%  \Pr\big[\big||\btau^{-1}_b(\ast)|- qw_\star\big| > w_\star^{1/3} \big] &\le \exp\left( -\Omega\big(t^2\mu\big)\right) \\
%  &= \exp\left(-\Omega(w_\star^{2/3}/qw_\star)\right) \\
%  &= \exp\big(-\Omega\big(w_\star^{1/6}\big)\big). \end{align*}
\end{proof}

The following observations may help the reader follow the next proof:   Recalling Table~\ref{table:circ-bullet}, since our $\btau$ belongs to $\{0,1,\ast\}^{A_{d-1}}$, we see that $\btau$ corresponds to the second row of the table:  the gates at depth $d-2$ are $\OR$ gates, a $\circ$-value for a coordinate of $\btau$ corresponds to 0, and a $\bullet$-value corresponds to 1.  However, since $\widehat{\btau}$, the lift of $\btau$, is one level higher than $\btau$ in the
$\BalancedSipser_d$ formula (see Figure~\ref{figure:the-figure}), $\widehat{\btau}$ corresponds to the first row of the table; so when Definition \ref{def:lift} specifies a coordinate $\widehat{\btau}_{\alpha,i}$ of $\widehat{\btau}$, a $\circ$-value for $\widehat{\btau}_{\alpha,i}$ corresponds to 1 and a $\bullet$-value corresponds to 0.

\ignore{
\begin{lemma}[Condition (2) of typicality] \label{lemma:condition2-initial}
Fix $\alpha \in A_{d-3}$. Then
\[ \Pr\big[\hat\btau_\alpha \in \{ \ast,1\}^{w_{d-3}} \big] \ge \ignore{1- w e^{-qw} = }1- e^{-{\Omega}(\sqrt{w})}. \]
\end{lemma}

\begin{proof}
Recall from Definition \ref{def:lift} that $\widehat{\btau}_{\alpha,i}=0$ iff $\btau_{\alpha,i} = \{0\}^{w_{d-2}}$
(in order for an $\OR$ to be 0, all its inputs must be 0).  In turn, each coordinate of $\btau_{\alpha,i}$ (we emphasize that $\btau_{\alpha,i}$ is a string of length  $w$) is an $\AND$ of the $w$ coordinates of some $\brho_a$ from  (\ref{eq:initial}), and hence is 0 with probability $1 - \lambda - q$. By independence  we have that
\begin{equation} \Pr[\hat\btau_{\alpha,i} = 0] = \delta := (1-\lambda-q)^{w} \le (1-q)^{w} \le e^{-qw}\label{eq:epsdef}
\end{equation}
holds independently for all $i\in [w_{d-3}]$. Therefore
\[  \Pr\big[\hat\btau_\alpha \in \{ \ast,1\}^{w_{d-3}} \big]
= (1-\delta)^{w_{d-3}}
\ge (1-e^{-qw})^{w_{d-3}}
\ge 1-e^{-\Omega(\sqrt{w\log w})}, \]
and so the lemma holds with room to spare.
\end{proof}
}

\begin{lemma}[Condition ({2}) of typicality] \label{lemma:condition3-initial}
Fix $\alpha \in A_{d-3}$.  Then
\[ \Pr\big[|(\hat\btau_\alpha)^{-1}(\ast)| <  w_{d-3}-w^{4/5} \big] \le e^{-\Omega(\sqrt{w})}.\]
\end{lemma}

\begin{proof}
Recall from Definition \ref{def:lift} that $\widehat{\btau}_{\alpha,i}=0$ iff $\btau_{\alpha,i} = \{0\}^{w_{d-2}}$
(in order for an $\OR$ to be 0, all its inputs must be 0).  In turn, each coordinate of $\btau_{\alpha,i}$ (we emphasize that $\btau_{\alpha,i}$ is a string of length  $w$) is an $\AND$ of the $w$ coordinates of some $\brho_a$ from  (\ref{eq:initial}), and hence is 0 with probability $1 - \lambda - q$. By independence  we have that
\begin{equation} \Pr[\hat\btau_{\alpha,i} = 0] = \delta := (1-\lambda-q)^{w} \le (1-q)^{w} \le e^{-qw}\label{eq:epsdef}
\end{equation}
holds independently for all $i\in [w_{d-3}]$.

We next give an expression for
$ \Pr\big[\hat\btau_{\alpha,i} =1\big].$  From Definition \ref{def:lift} we have that $\widehat{\btau}_{\alpha,i}=1$ iff any of the $w$ coordinates of $\btau_{\alpha,i}$ is 1 (in order for an $\OR$ to be 1, we only need one input to be 1).  As noted above, each coordinate of $\btau_{\alpha,i}$ is an $\AND$ of the $w$ coordinates of some $\brho_a$ from (\ref{eq:initial}); this $\AND$ is 1 iff its input string is $\{1\}^{w}$, so by (\ref{eq:initial}) each coordinate of $\btau_{\alpha,i}$ is not 1 with probability $1-\lambda$.  Hence all $w$ coordinates of $\btau_{\alpha,i}$ are not 1 with probability $(1-\lambda)^{w}$, and $\widehat{\btau}_{\alpha,i}=1$ with probability $1-(1-\lambda)^{w}$.

We thus have that, independently for all $i \in [w_{d-3}],$
\begin{align*} \Pr\big[\hat\btau_{\alpha,i} \in \{ 0,1\}\big]  = \delta + (1-(1-\lambda)^{w})
\le \delta + (1-(1-\lambda w))
\le 2\lambda w = {\frac {2 (\log w)^{3/2}}{w^{1/4}}},
\end{align*}
%and
%\begin{align*} \Pr\big[\hat\btau_{\alpha,i} \in \{ \bullet,\circ\}\big]  &= \eps + (1-(1-\lambda)^w) \\
%&\ge 1-(1-\lfrac1{2}\lambda w) \\
%&= \lfrac1{2} w^{-1/2}
%\end{align*}
where the last inequality holds (with room to spare) by (\ref{eq:epsdef}).
  Applying Fact~\ref{fact:chernoff}, we have that
\[ \Pr\big[ |\hat\btau^{-1}_\alpha(\{0,1\})| > w^{4/5} \big] \le e^{-\Omega(\sqrt{w})} \]
with room to spare.
\end{proof}

\begin{proof}[Proof of Proposition~\ref{prop:initial-typical}]
The proposition follows immediately from Lemmas~\ref{lemma:condition1-initial} and \ref{lemma:condition3-initial} and a union bound over all $a \in A_{d-2}$ and $\alpha \in A_{d-3}$, using the fact that $|A_{d-3}| \leq |A_{d-2}| \leq n \leq w^{O(d)}$ and the bound $d \leq {\frac {c\log w}{\log\log w}}$.\ignore{\rnote{It seems here all we need is that $d \leq c w^{1/6}/\log(w)$, I guess the stronger bound on $d$ is required elsewhere...}}
\end{proof}

\subsubsection{Preserving typicality:  Proof of Proposition \ref{prop:typical-yields-typical}}  \label{sec:pf-typical-yields-typical}

The following numerical lemma relates $q_a$ as defined in (\ref{eq:def-of-qa}) of Definition~\ref{def:main-projection} to $q$ as defined in~(\ref{eq:def-of-lambda-and-q}):

\begin{lemma}
\label{lem:bound-on-qa}
Let $2 \le k \le d-1$ and $S\sse [w_{k-1}]$ be $k$-acceptable (i.e.~$|S| = qw \pm w^{\gradual(k,d)}$), and define
\[ q' = \frac{(1-t_k)^{|S|} - \lambda}{t_{k-1}}. \]
Then $q' = q\cdot  (1 \pm 2t_k w^{\gradual(k,d)})$. (And in particular, by our bounds on $t_k$ in Lemma~\ref{lemma:tk-bound} and the definition of $\gradual(k,d)$, we have that $q' = q \pm o(q)$ for all $k$.)
\end{lemma}

\begin{proof}
For the lower bound, we have the following:
\begin{align*}
q' &\le   \frac{(1-t_{k})^{qw-w^{\gradual({k},d)}} - \lambda}{t_{k-1}}  \\
& = \frac{(1-t_{k})^{qw} - \lambda (1-t_{k})^{w^{\gradual(k,d)}}}{ t_{k-1} (1-t_{k})^{w^{\gradual(k,d)}}}  \\
&\le \frac{(1-t_{k})^{qw} - \lambda }{ t_{k-1} (1-t_{k})^{w^{\gradual(k,d)}}}
+ {\frac {\lambda t_k w^{\gradual(k,d)}}{ t_{k-1} (1-t_{k})^{w^{\gradual(k,d)}}}}\\
& = \frac{t_{k-1} q } { t_{k-1} (1-t_{k})^{w^{\gradual(k,d)}}} + {\frac {\lambda t_k w^{\gradual(k,d)}}{ t_{k-1} (1-t_{k})^{w^{\gradual(k,d)}}}} \qquad \qquad\qquad  \text{(by (\ref{eq:def-of-tk}))} \\
&\le \frac{q}{1-t_{k}w^{\gradual(k,d)}}
+ {\frac {1 + 3q^{0.1}}{1 - t_kw^{\gradual(k,d)}}} \cdot \lambda w^{\gradual(k,d)} \qquad \qquad \qquad \text{(by Lemma~\ref{lemma:tk-bound})}\\
&\le q \cdot (1+ 2t_{k}w^{\gradual(k,d)}),
\end{align*}
where for the last inequality we have used the fact that $q t_k =
\tilde{\Theta}(w^{-1})$ whereas $\lambda = \tilde{\Theta}(w^{-5/4}).$
For the upper bound, we have
\begin{align*}
q' &\ge \frac{(1-t_{k})^{qw+w^{\gradual(k,d)}} - \lambda}{t_{k-1}}  \\
&\ge \frac{(1-t_{k})^{qw}(1-t_{k}w^{\gradual(k,d)}) - \lambda}{t_{k-1}} \\
&\ge q\cdot (1-t_{k}w^{\gradual(k,d)}) - \frac{ \lambda}{t_{k-1}} \qquad \qquad \qquad \qquad \text{(by (\ref{eq:def-of-tk}))} \\
&\ge q \cdot (1-2t_{k}w^{\gradual(k,d)}).
\end{align*}
where the last inequality uses the definition of $\lambda$ in (\ref{eq:def-of-lambda-and-q}) and our bound on $t_{k-1}$ in Lemma~\ref{lemma:tk-bound}.
\end{proof}

Similar to the proof of Proposition~\ref{prop:initial-typical}, Proposition~\ref{prop:typical-yields-typical} follows from Lemmas~\ref{lemma:condition1-yields} \ignore{\ref{lemma:condition2-yields} }and \ref{lemma:condition3-yields} (stated and proved below) and a union bound, again using the fact that each $|A_i| \leq n$ and the bound $d \leq {\frac {c\log w}{\log \log w}}$.
Since Proposition \ref{prop:typical-yields-typical} deals with general values of $k$ which may correspond to either row of Table \ref{table:circ-bullet}, to avoid redundancy we use $\circ,\bullet$ notation in the statements and proofs of the following lemmas.

\begin{lemma} [Condition (1) of typicality] \label{lemma:condition1-yields}
For $2 \le k \le d-2$ let $\tau \in \{ \bullet,\circ,\ast\}^{A_{k+1}}$ be typical and fix $a \in A_{k-1}$. Then
\[ \Prx_{\brho\leftarrow\calR(\tau)}\big[ |(\hat\brho_a)^{-1}(\ast)| = qw \pm w^{\gradual(k,d)}  \big] \ge 1-\exp(-\tilde{\Omega}(w^{2\gradual(k,d)-\frac1{2}})) \ge 1-e^{-\Omega(w^{1/6})}.\ignore{\lnote{The tilde does not appear on the RHS since $\gradual(k,d) > 1/3$ when $k < d$.}}\]
(Recall that from Definition~\ref{def:typical} that $\gradual(k,d) = \frac1{3} + {\frac {d-k-1}{12d}}$).  \end{lemma}

\begin{proof}
Since $\tau \in \{\bullet,\circ,\ast\}^{A_{k+1}}$ is typical, we have that
\begin{equation}
 \hat\tau_a \in \{ \ast,\circ\}^{w} \quad \text{and}\quad |(\hat\tau_a)^{-1}(\ast)| \ge w-w^{4/5} \label{eq:typical1-IH2}
 \end{equation}
by the second and third property of $\tau$ being typical. Furthermore, for every $i\in [w]$ such that $\hat\tau_{a,i} = \ast$, we have that
\begin{equation}
 \tau_{a,i} \in \{ \ast,\circ\}^{w} \quad \quad \text{and} \quad \quad qw - w^{\gradual(k+1,d)} \leq|(\tau_{a,i})^{-1}(\ast)| \leq  qw + w^{\gradual(k+1,d)}, \label{eq:typical1-IH1}
 \end{equation}
by the first property of $\tau$ being typical. Writing $S_{a,i}$ for $(\tau_{a,i})^{-1}(\ast)$ (a subset of $[w]$) and $S_a$ for $(\hat\tau_a)^{-1}(\ast)$ (a subset of $[w]$), it follows from the second branch of (\ref{eq:OW-restriction}) and Definition~\ref{def:lift} that every $i\in S_a$ satisfies
\[  \Prx_{\brho\leftarrow\calR(\tau)}\big[\hat\brho_{a,i} = \ast\big] = q_{a,i} = \frac{(1-t)^{|S_{a,i}|}-\lambda}{t}. \]
Since $S_{a,i}$ is $(k+1)$-acceptable, by the $k+1$ case of Lemma~\ref{lem:bound-on-qa} we have that \[ q_{a,i} = q\cdot (1\pm 2t_{k+1}w^{\gradual(k+1,d)}).\] Since $|S_a| \le w$, we have
\[  \Ex_{\brho\leftarrow\calR(\tau)}\big[ |(\hat\brho_a)^{-1}(\ast)|\big] =  \sum_{i\in S_a} q_{a,i}  \le w\cdot  q(1+2t_{k+1} w^{\gradual(k+1,d)}) \le qw + \tilde{O}(w^{\gradual(k+1,d)}), \]
where the $\tilde{O}$ comes from the fact that $wt_{k+1} q = \Theta(\log w)$ (recalling Lemma~\ref{lemma:tk-bound} we have that $t_{k+1} = q \pm o(q)$).
On the other hand, by (\ref{eq:typical1-IH2}) and similar reasoning we also have the lower bound
\[ \Ex_{\brho\leftarrow\calR(\tau)}\big[ |(\hat\brho_a)^{-1}(\ast)|\big] \ge (w-w^{4/5}) \cdot q(1-2t_{k+1} w^{\gradual(k+1,d)}) \ge qw - \tilde{O}(w^{\gradual(k+1,d)}),\]
where we have taken advantage of the fact that $w^{4/5}q=\tilde{O}(w^{0.3})=o(w^{\gradual(k+1,d)})$.
Since $w^{\gradual(k,d)} = \omega(\polylog (w) \cdot w^{\gradual(k+1,d)})$ (here is where we are using the fact that $d \le \frac{c\log w}{\log \log w}$),
%\rnote{This seems like a significant constraint; it seems to me that basically what we need here is that
%\[
%w^{1/d} \gg \log w,
%\]
%which means $d \ll {\frac {\log w} {\log \log w}}$.  This is a much stronger constraint than the $d \leq {\frac {w^{1/6}}{\log w}}$ that is stated in the condition of
%Theorem~\ref{thm:initial-typical}, which is itself stronger than the $d \leq {\frac {w^{1/2}}{\log w}}$ that I think we really want.
%
%Maybe we can/should have $\gradual(k,d)$ not be linear in $k$, but ``give us more room where we need it''?  (I didn't really think about this...)
%}
it follows from Fact~\ref{fact:chernoff}  that
\begin{align*}
 \Prx_{\brho\leftarrow\calR(\tau)}\big[ |(\hat\brho_a)^{-1}(\ast)| \ne qw \pm w^{\gradual(k,d)}  \big] &\le \exp(-\Omega(w^{2\gradual(k,d)}/qw)\big) \\
 &\le \exp(-\tilde{\Omega}(w^{2\gradual(k,d)-\frac1{2}})).\qedhere
 \end{align*}
\ignore{  \lnote{the $\tilde{\Omega}$ above comes from the fact that $qw$ is now $\sqrt{w\log w}$ instead of $\sqrt{w}$}}
\end{proof}

\begin{lemma}
\label{lemma:haha}
Fix $2 \le k \le d-2$ and let $\tau \in \{\bullet,\circ,\ast\}^{A_{k+1}}$ be typical.  For each $a\in A_{k-1}$ we write $S_a = S_a(\tau)$ to denote $(\hat\tau_a)^{-1}(\ast)$ (note that this is a subset of $[w]$).  Then for $\brho\leftarrow \calR(\tau)$,  we have that $\hat\brho_a$ (which is a string in $\{\bullet,\circ,\ast\}^{w}$) satisfies:
\[  \left\{
\begin{array}{ll}
 \hat\brho_a = \{\circ\}^{w} & \text{with probability~$\prod_{i\in S_a} (1-\lambda-q_{a,i})$}  \\
 (\hat\brho_a)^{-1}(\bullet) \ne \emptyset & \text{with probability~$1-(1-\lambda)^{|S_a|}$}  \\
 \hat\brho_a \in \{ \circ,\ast\}^{w} \setminus \{\circ\}^{w} & \text{otherwise,}
\end{array}
\right.
\]
 independently for all $a \in A_{k-1}$. (Recall that $\hat\tau_a \in \{\ast,\circ\}^{w} \setminus \{\circ\}^{w}$ for all $a\in A_{k-1}$ since $\tau$ is typical.) This implies that
\[ \hat{\hat\brho}_a = \left\{
\begin{array}{ll}
 \bullet & \text{with probability~$\prod_{i\in S_a} (1-\lambda-q_{a,i})$}  \\
\circ& \text{with probability~$1-(1-\lambda)^{|S_a|}$}  \\
\ast & \text{otherwise}
\end{array}
\right.
\]
 independently for all $a \in A_{k-1}$. (Recall that $\hat{\hat\tau}_a = \ast$ for all $a\in A_{k-1}$ since $\tau$ is typical.) \end{lemma}

\begin{proof}
The value of $\hat\brho_{a,i}$ is independent across all $a \in A_{k-1}$ and $i\in [w]$ such that $\hat\tau_{a,i} = \ast$. Fix such a $a\in A_{k-1}$ and $i \in [w]$, and recall that \[ \tau_{a,i} \in \{\ast,\circ\}^{w} \setminus \{\circ\}^{w}. \]
By (\ref{eq:OW-restriction}) and Definition~\ref{def:lift} (the definition of the lift operator), we have that
\[
\hat\brho_{a,i} = \left\{
\begin{array}{ll}
 \bullet & \text{with probability~$\lambda$}  \\
\ast& \text{with probability~$q_{a,i}$}  \\
\circ & \text{otherwise, with probability~$1-\lambda-q_{a,i}$.}
\end{array}
\right.
\]
The lemma then follows by independence.
\end{proof}

\begin{remark}
\label{rem:haha}
 If $\tau \in \{\bullet,\circ,\ast\}^{A_{k+1}}$ is typical then
 (recall that $S_a=(\hat\tau_a)^{-1}(\ast)$ is a subset of $[w]$ and $S_{a,i} = (\tau_{a,i})^{-1}(\ast)$ is a subset of $[w]$) we have
 \[ |S_a| \ge w-w^{4/5} \quad \quad \text{and} \quad \quad qw - w^{\gradual(k+1,d)} \leq |S_{a,i}| \leq qw + w^{\gradual(k+1,d)} \text{ for all $i\in S_a$}. \]
 Therefore we have the estimates
 \begin{align*}
 \Pr\Big[\hat{\hat\brho}_a = \bullet\Big] = \prod_{i\in S_a} (1-\lambda - q_{a,i}) \le (1-q_{a,i})^{w-w^{4/5}}
 \le \left(1-\lfrac{q}{2}\right)^{w-w^{4/5}}
 \le e^{-qw/4} = e^{-\Omega(\sqrt{w\log w})},
 \end{align*}
 where we have used Lemma~\ref{lem:bound-on-qa} for the second inequality,
 \ignore{
 \gray{
 \begin{align*} \Pr\Big[\hat{\hat\brho}_a = \bullet\Big] = \prod_{i\in S_a} (1-\lambda - q_{a,i}) &\le (1-q_{a,i})^{w-w^{4/5}}    \\
 &\le \left( 1- \frac{(1-t)^{qw + w^{\gradual(k+1,d)}}-\lambda}{t} \right)^{w-w^{4/5}} \\
% &= \left( 1- \frac{(1-t)^{qw + w^{\gradual(k+1,d)}}}{t}  + \frac{\lambda}{t} \right)^{w-w^{4/5}} \\
 &= \left( 1- q (1-t)^{w^{\gradual(k+1,d)}}   + \frac{\lambda}{t} \right)^{w-w^{4/5}}  \qquad \text{(by (\ref{eq:def-of-q})}) \\
 &\le \left( 1- q \cdot \big(1-t w^{\gradual(k+1,d)}\big)   + \frac{\lambda}{t} \right)^{w-w^{4/5}}  \\
 &\le \left( 1- q \cdot \big(1-2t w^{\gradual(k+1,d)}\big)  \right)^{w-w^{4/5}}
 \text{(using (\ref{eq:def-of-q}) and (\ref{eq:def-of-lambda-and-t}))}\\
  &= \left( 1- q + 2q t w^{\gradual(k+1,d)}  \right)^{w-w^{4/5}}  \\
 &\le \left( 1- \lfrac{q}{2}  \right)^{w-w^{4/5}}   \le e^{-qw/4} = e^{-\Omega(\sqrt{w})},
 \end{align*}
 }}and
 \[  \Pr\Big[\hat{\hat\brho}_a = \circ\Big] = 1-(1-\lambda)^{|S_a|}
 \le 1-(1-\lambda)^{w}
 \le 1-(1-\lambda w)
 =  \lambda w.  \]
\end{remark}

\begin{lemma} [Condition ({2}) of typicality] \label{lemma:condition3-yields}
For $2\le k \le d-2$  let $\tau \in \{ \bullet,\circ,\ast\}^{A_{k+1}}$ be typical and fix $\alpha \in A_{k-2}$. Then
\[ \Prx_{\brho\leftarrow\calR(\tau)}\big[ \big|\big(\hat{\hat\brho}_\alpha\big)^{-1}(\ast)\big| \geq  w_{k-2} - w^{4/5} \big] = 1- e^{-\Omega(\sqrt{w})}.  \]
\end{lemma}

\begin{proof}
By Lemma~\ref{lemma:haha} and the two estimates of Remark~\ref{rem:haha}, each coordinate of
$(\hat{\hat\brho})_\alpha$ is independently in $\{\bullet,\circ\}$ with probability at most
$e^{-\Omega(\sqrt{w})} + \lambda w = O\big({\frac {(\log w)^{3/2}}{w^{1/4}}}\big).$  Hence the expected size of $\big|\big(\hat{\hat\brho}_\alpha\big)^{-1}(\{\bullet,\circ\})\big|$ is $\tilde{O}(w^{3/4})$, and we may apply Fact~\ref{fact:chernoff} to get that
\[ \Prx_{\brho\leftarrow\calR(\tau)}\big[ \big|\big(\hat{\hat\brho}_\alpha\big)^{-1}(\{\bullet,\circ\})\big| > w^{4/5} \big]  \le e^{-\Omega(\sqrt{w})} \]
with room to spare.
\end{proof}

%\newpage

\subsection{$\BalancedSipser$ survives random projections}
\label{sec:sipser-survives}

In this subsection we prove the main results of Section~\ref{sec:combo}; these are two results which show, in different ways, that the $\BalancedSipser_d$ function ``retains structure'' after being hit with the random projection $\mathbf{\Psi}$.  The first of these results, Proposition \ref{prop:what-happens-to-target}, gives a useful characterization of $\mathbf{\Psi}(\BalancedSipser_d)$ by showing that it is distributed identically to a (suitably randomly restricted) \emph{depth-one} formula.  The second of these results, Proposition \ref{prop:boundexp}, shows that this randomly restricted depth-one formula is very close to perfectly balanced in expectation.  Our later arguments will use both these types of structure.

\subsubsection{$\BalancedSipser_d$ reduces under $\mathbf{\Psi}$ to a random restriction of $\BalancedSipser^{(1)}_d$}

Recalling the definitions of the depth-$k$ $\BalancedSipser_d^{(k)}$ formulas from Definition~\ref{def:truncate-sipser}, we begin with the following observation regarding the effect of projections on the $\BalancedSipser^{(k)}_d$ formulas:

\begin{fact}
\label{fact:project-sipser}
For $2\le k \le d$ we have that
\[ \proj\,\BalancedSipser_d^{(k)} \equiv \BalancedSipser_d^{(k-1)}. \]
\end{fact}

In words, Fact~\ref{fact:project-sipser} says that the projection operator ``wipes out'' the bottom-layer gates of $\BalancedSipser_d^{(k)}$, reducing its depth by exactly one. Fact~\ref{fact:project-sipser} is a straightforward consequence of the definitions of projections and the $\BalancedSipser^{(k)}_d$ formulas (Definitions~\ref{def:projection} and~\ref{def:truncate-sipser} respectively), but is perhaps most easily seen to be true via the equivalently view of projections described in Remark~\ref{rem:equiv-projection}: for every bottom-layer gate $a\in A_k$ of $\BalancedSipser^{(k)}_d$, the projection operator simply replaces every one of its $w_{k-1}$ formal input variables $x_{a,1},\ldots,x_{a,w_{k-1}}$ with the same fresh formal variable $y_a$. Since $\AND(y_a,\ldots,y_a) \equiv \OR(y_a,\ldots,y_a) \equiv y_a$,  the gate simplifies to the single variable~$y_a$.  (Indeed, we defined our projection operators precisely so that they sync up with $\BalancedSipser^{(k)}_d$ this way.)

The same reasoning, along with the definition of lifts (see Definition~\ref{def:lift} and the discussion after), yields the following extension of Fact~\ref{fact:project-sipser}:

\begin{fact}
\label{fact:lift}
For $2\le k \le d$ and $\rho \in \{0,1,\ast\}^{A_k}$ we have
\[ \proj_\rho\, \BalancedSipser^{(k)}_d \equiv \BalancedSipser^{(k-1)}_d \uhr \hat\rho.
\]
\end{fact}

\begin{remark}
\label{rem:typical-restrictions}
With Fact~\ref{fact:lift} in hand we now revisit our definition of typical restrictions (recall Definition~\ref{def:typical} and the discussion thereafter). Recall that the high-level rationale behind this definition is that for $\rho$ such that $\hat\rho$ is typical, the projection $\proj_\rho$ has a ``very limited and well-controlled effect'' on the target $\BalancedSipser_d$. We now make this statement more precise (the reader may find it helpful to refer to the illustration in Figure~\ref{figure:the-figure}).

Fix $\rho \in \supp(\calR_\init)$ such that $\hat\rho$ is typical. By Fact~\ref{fact:lift}, we have that
\[ \proj_\rho\,\BalancedSipser_d \equiv \BalancedSipser_d^{(d-1)} \uhr \hat\rho.\]
Since $\hat\rho$ is typical,
\begin{itemize}
\item The first condition of Definition~\ref{def:typical} implies that $|(\hat\rho_a)^{-1}(\ast)| = \Theta(qw) = \tilde{\Theta}(\sqrt{w})$ for all $a\in A_{d-2}$.  Each such $a\in A_{d-2}$ is the address of an $\OR$ gate, and so if $(\hat\rho_a)^{-1}(1) \ne \emptyset$ the gate is satisfied and evaluates to $1$, and otherwise if $\hat\rho_a \in \{\ast,0\}^w$ the value of the gate remains undetermined (i.e.~it ``evaluates to $\ast$'') and its fan-in becomes $|(\hat\rho_a)^{-1}(\ast)| = \tilde{\Theta}(\sqrt{w})$.
\item The second condition of Definition~\ref{def:typical} tells us that between the two possibilities above, the latter is far more common: for every $\alpha \in A_{d-3}$ specifying a block of $w_{d-3}$ many $\OR$ gates, at most $w^{4/5}$ of these gates evaluate to $1$ and the remaining (vast majority) are undetermined.  Equivalently, \emph{all} the $\AND$ gates at level $d-3$ remain undetermined, and they all have fan-in at least $w_{d-3} - w^{4/5} = w_{d-3}\, (1-o(1))$.
\end{itemize}
The same description holds for $\proj_{\rho^{(k)}}$ and $\BalancedSipser_d^{(k)}$. For $\hat{\rho^{(k)}}$'s that are typical the projection operator $\proj_{\rho^{(k)}}$:
\begin{itemize}
\item ``wipes out'' the bottom-level (level-$k$) gates of $\BalancedSipser^{(k)}_d$,
\item ``trims'' the fan-ins of the level-$(k-1)$ gates from $w$ to $\tilde{\Theta}(\sqrt{w})$,
\item keeps the fan-ins of all level-$(k-2)$ gates at least $w_{k-2} - w^{4/5} = w_{k-2}\,(1-o(1))$.
\end{itemize}
Note in particular that the entire structure of the formula from levels $0$ through $k-3$ is identical to that of $\BalancedSipser_d^{(k)}$, and so $\proj_{\rho^{(k)}}\,\BalancedSipser_d^{(k)}$ ``contains a perfect copy of'' $\BalancedSipser_d^{(k-3)}$.
\end{remark}

Repeated applications of Fact~\ref{fact:lift} gives us the following proposition.  (The  proposition is intuitively very useful since, it tells us that in order to understand the effect of the random projection $\mathbf{\Psi}$ on the (relatively complicated) $\BalancedSipser_d$ function, it suffices to analyze the effect of the random restriction $\widehat{\brho^{(2)}}$ on the (much simpler) $\BalancedSipser^{(1)}_d$ function; we will apply it in the final proof of each of our main lower bounds.)

\begin{proposition}
\label{prop:what-happens-to-target}
Consider $\BalancedSipser_d : \zo^{n} \to\zo$. Then
\[
\mathbf{\Psi}(\BalancedSipser_d) \equiv \BalancedSipser_d^{(1)} \uhr \hat{\brho^{(2)}}.
\]
\end{proposition}

\begin{proof}
By Fact~\ref{fact:lift} we have that
\begin{equation} \proj_{\rho^{(d)}}\, \BalancedSipser_d \equiv \BalancedSipser^{(d-1)}_{d} \uhr \hat{\rho^{(d)}} \label{eq:ccc}
\end{equation}
for all $\rho^{(d)} \in \supp(\calR_\init) \equiv \{0,1,\ast\}^{n}$.  Furthermore for
$\rho^{(k+1)} \in \{0,1,\ast\}^{A_{k+1}}$ and $\rho^{(k)} \in \supp(\calR(\hat{\rho^{(k+1)}})) \sse \{0,1,\ast\}^{A_{k}}$ we have
\begin{eqnarray} \proj_{\rho^{(k)}}\, \Big(\BalancedSipser_d^{(k)} \uhr \hat{\rho^{(k+1)}}\Big) &\equiv& \proj\,\Big(\big(\BalancedSipser_d^{(k)} \uhr \hat{\rho^{(k+1)}}  \big) \uhr \rho^{(k)}\Big)  \nonumber  \\
&\equiv& \proj\,\Big(\BalancedSipser_d^{(k)} \uhr  \rho^{(k)}\Big) \nonumber \\
&\equiv& \BalancedSipser^{(k-1)}_d \uhr \hat{\rho^{(k)}}, \label{eq:ddd}
\end{eqnarray}
where the first equivalence is by the definition of $\rho$-projection (Definition~\ref{def:projection}), the second is by the fact that $\calR(\hat{\rho^{(k+1)}})$ is supported on refinements of $\hat{\rho^{(k+1)}}$ (and in particular, $\rho^{(k)}$ refines $\hat{\rho^{(k+1)}}$), and the last  is Fact~\ref{fact:lift}.  The proposition follows from (\ref{eq:ccc}), repeated application of (\ref{eq:ddd}), and the definition of $\mathbf{\Psi}$ (Definition~\ref{def:boldpsi}).
\end{proof}

\subsubsection{$\BalancedSipser_d$ remains unbiased after random projection by $\mathbf{\Psi}$}

Recall that $\BalancedSipser^{(1)}_d$ denotes the function computed by the top gate of $\BalancedSipser_d$, and in particular, $\BalancedSipser^{(1)}_d$ is a $w_0$-way $\OR$ if $d$ is even, and a $w_0$-way $\AND$ if $d$ is odd (c.f.~Definition~\ref{def:truncate-sipser}).  In this subsubsection we will assume that $d$ is even; the argument for odd values of $d$ follows via a symmetric argument.

To obtain our ultimate results we will need a lower bound on the bias of $\mathbf{\Psi}(\BalancedSipser_d)$ under $\bY$ (or equivalently, by the preceding proposition, on the bias of $\BalancedSipser_d^{(1)} \uhr \hat{\brho^{(2)}}$ where $\brho^{(2)}$ is distributed as described in Definition~\ref{def:boldpsi}). The following lemma will help us establish such a lower bound:

\begin{lemma} \label{lemma:E-min-lower-bound}
Let $\tau \in \{0,1,\ast\}^{A_2}$ be typical.  Then for $\brho\leftarrow \calR(\tau)$ and $\bY \leftarrow \{ 0_{1-t_1},1_{t_1}\}^{w_0}$ we have
\[ \Ex_\brho\Big[ \bias(\BalancedSipser^{(1)}_d \uhr \hat\brho, \bY)\Big] \ge \frac1{2} -  \tilde{O}(w^{-1/12}).
\]
\end{lemma}

\begin{proof}

By our assumption that $d$ is even we may write $\OR_{w_0}$ in place of $\BalancedSipser^{(1)}_d$.  Since $\tau$ is typical, we have by Conditions (2) and (3) of Definition~\ref{def:typical} that
\[ \hat\tau \in \{0,\ast\}^{w_0} \quad \text{and} \quad |(\hat\tau)^{-1}(\ast)| \ge w_0 - w^{4/5}. \]
Furthermore, by (\ref{eq:OW-restriction}) of Definition~\ref{def:main-projection} and Definition~\ref{def:lift} (the definition of the lift operator), we have that
\begin{equation}
\hat\brho_{i} = \left\{
\begin{array}{ll}
 1 & \text{with probability~$\lambda$}  \\
\ast& \text{with probability~$q_{i}$}  \\
0 & \text{otherwise, with probability~$1-\lambda-q_{i}$}
\end{array}
\right. \label{eq:final-projection}
\end{equation}
independently for all $i\in (\hat\tau)^{-1}(\ast) \sse [w_0]$, where
\[ q_i = \frac{(1-t_{{2}})^{|S_i|}  -\lambda}{t_{{1}}} \quad
\]
and $S_i = S_i(\tau) = \tau^{-1}_i(\ast) = \{j \in [w_1]: \tau_{i,j}=\ast\}$ satisfies
$ |S_i| = qw \pm w^{\gradual(2,d)}.$
By a calculation very similar to the one that was employed in the proof of Lemma \ref{lemma:condition1-yields}, we have that
\begin{equation} \Pr\big[ |(\hat\brho)^{-1}(\ast)| = qw_0 \pm w^{\gradual(1,d)}\big] \ge 1-e^{-\Omega(w^{1/6})}. \label{eq:number-of-stars}
\end{equation}
Furthermore, (\ref{eq:final-projection}) also implies that
\begin{equation} \Pr[ \hat\brho \in \{0,\ast\}^{w_0}] = (1-\lambda)^{|\hat{\tau}^{-1}(\ast)|} \ge (1-\lambda)^{w_0} \ge 1-\lambda w_0 = 1-{\tilde{O}}(w^{-1/4}).   \label{eq:not-killed}
\end{equation}
Fix any $\rho \in \supp(\calR(\tau))$ that satisfies the events of both (\ref{eq:number-of-stars}) and (\ref{eq:not-killed}). Writing $S(\hat\rho) \sse [w_0]$ to denote the set $(\hat\rho)^{-1}(\ast)$, we have the bounds
\begin{align*} \Prx_{\bY}[(\OR_{w_0} \uhr \hat\rho)(\bY) = 0 ] &= (1-t_1)^{|S(\hat{\rho})|}  \\
&\ge (1-t_1)^{qw_0 + w^{\gradual(1,d)}} \\
&\ge \left(\frac1{2} - \Theta\left(\frac{\log w}{w}\right)\right) (1-t_1)^{w^{\gradual(1,d)}} \\
&\ge \left(\frac1{2} - \Theta\left(\frac{\log w}{w}\right)\right) (1-t_1w^{\gradual(1,d)}), \\
&\ge {\frac 1 2} - \tilde{O}(w^{-1/12}),
\end{align*}
where the second inequality crucially uses the definition (\ref{eq:def-of-w0}) of $w_0$
and its corollary (\ref{eq:upper lower}).
Similarly,
\begin{align*} \Prx_{\bY}[(\OR_{w_0} \uhr \hat\rho)(\bY) = 0 ] &= (1-t_1)^{|S(\hat{{\rho}})|}  \\
&\le (1-t_1)^{qw_0 - w^{\gradual(1,d)}} \\
&\le \frac1{2}\cdot  (1-t_1)^{-w^{\gradual(1,d)}} \\
&\le {\frac 1 2} + \tilde{O}(w^{-1/12}),
\end{align*}
which establishes the lemma.
\end{proof}

Now we are ready to lower bound the expected bias of $\mathbf{\Psi}(\BalancedSipser_d)$ (or equivalently,  of $\BalancedSipser_d^{(1)} \uhr \hat{\brho^{(2)}}$) under $\bY$:\ignore{\rnote{I am somehow unsure of where we should be making the `` $\mathbf{\Psi}(\BalancedSipser_d)$ to  $\BalancedSipser_d^{(1)} \uhr \hat{\brho^{(2)}}$'' switch.  We could have the statement of this claim be about $\mathbf{\Psi}(\BalancedSipser_d)$ rather than $\BalancedSipser_d^{(1)} \uhr \hat{\brho^{(2)}}$  -- right now I am not clear on the pros and cons of doing this.  Let's discuss this at some point?}}

\begin{proposition} \label{prop:boundexp}
For $\mathbf{\Psi}$ as defined in Definition \ref{def:boldpsi},
\[ \mathbf{\Psi}(f) \equiv \proj_{\brho^{(2)}}\,\proj_{\brho^{(3)}}\cdots \proj_{\brho^{(d-1)}}\,\proj_{\brho^{(d)}}\, f \]
where   $\brho^{(d)}\leftarrow\calR_{\init}$ and $\brho^{(k)}\leftarrow\calR(\hat{\brho^{(k+1)}})$ for all $2 \le k\le d-1$, and for $\bY\leftarrow \{ 0_{1-t_1},1_{t_1}\}^{w_0}$, we have that
\[
\Ex_{\mathbf{\Psi}}\bigg[ \bias(\BalancedSipser_d^{(1)} \uhr \hat{\brho^{(2)}}, \bY)\bigg] \geq {\frac 1 2} - \tilde{O}(w^{-1/12}).\]
\end{proposition}

\begin{proof}
By Proposition~\ref{prop:initial-typical} and $d-3$ successive applications of Proposition~\ref{prop:typical-yields-typical}, we have that
\[ \Pr\big[\widehat{\brho^{(d)}},\dots,\widehat{\brho^{(3)}}\text{ are all typical}\big] \geq 1 -  d \cdot e^{-{\tilde{\Omega}}(w^{1/6})}.\]  For every typical $\widehat{\rho^{(3)}} \in \{0,1,\ast\}^{A_2}$,
Lemma \ref{lemma:E-min-lower-bound} gives that
\[ \Ex_{\brho^{(2)}\leftarrow\calR(\hat{\rho^{(3)}})}\Big[  \bias(\BalancedSipser^{(1)}_d \uhr \widehat{\brho^{(2)}},\bY) \Big] \ge \frac1{2} -  \tilde{O}(w^{-1/12}),
\]
which together with the preceding inequality gives the proposition.
\end{proof}

\begin{remark}
\label{rem:sipser-is-balanced}
We note that combining Proposition~\ref{prop:what-happens-to-target} and Proposition~\ref{prop:boundexp},
for $\bY\leftarrow \{ 0_{1-t_1},1_{t_1}\}^{w_0}$
we have that
\[
\Ex_{\mathbf{\Psi}}\big[\bias(\mathbf{\Psi}(\BalancedSipser_d),\bY)\big] \geq {\frac 1 2} - \tilde{O}(w^{-1/12}),\]
which we may rewrite as
\[
\Pr[(\mathbf{\Psi}(\BalancedSipser_d))(\bY)=0] =
\Ex_{\mathbf{\Psi}}\Big[\Prx_{\bY} (\mathbf{\Psi}(\BalancedSipser_d))(\bY)=0]\Big]
= \frac1{2} \pm \tilde{O}(w^{-1/12}).
\]
Applying Proposition~\ref{prop:complete-to-uniform}, we get that for $\bX \leftarrow \{0_{1/2},1_{1/2}\}^n$ we have
\[
\Pr[\BalancedSipser_d(\bX)=1] = \frac1{2} \pm \tilde{O}(w^{-1/12}). 
\]
verifying (\ref{eq:sipser-is-balanced}) in Section~\ref{sec:sipser}:  the $\BalancedSipser_d$ function is indeed (essentially) balanced.
\end{remark}

\section{Proofs of main theorems} \label{sec:puttogether}

Recall that $\BalancedSipser^{(1)}_d$ denotes the function computed by the top gate of $\BalancedSipser_d$, and in particular, $\BalancedSipser^{(1)}_d$ is a $w_0$-way $\OR$ if $d$ is even, and a $w_0$-way $\AND$ if $d$ is odd (c.f.~Definition~\ref{def:truncate-sipser}).  Throughout this section we will assume that $d$ is even; the argument for odd values of $d$ follows via a symmetric argument. For conciseness we will sometimes write $\OR_{w_0}$ in place of $\BalancedSipser_d^{(1)}$ in the arguments below; we stress that these are the same function.

\subsection{``Bottoming out'' the argument}

As we will see in the proofs of Theorems~\ref{thm:smallbottomfanin} and~\ref{thm:alternationpattern}, the machinery we have developed enables us to relate the correlation between $\BalancedSipser_d$ and the circuits $C$ against which we are proving lower bounds, to the correlation between $\BalancedSipser_d^{(1)} \uhr \widehat{\brho^{(2)}}$ (obtained by hitting $\BalancedSipser_d$ with the random projection $\mathbf{\Psi}$) and bounded-width CNFs (that are similarly obtained by hitting $C$ with $\mathbf{\Psi}$).  To finish the argument, we need to bound the correlation between $\BalancedSipser_d^{(1)} \uhr \tau$ (for suitable restrictions $\tau$) and such CNFs.   The following proposition, which is a slight extension of Lemma 4.1 of~\cite{OW07}, enables us to do this, by relating the correlation between  $\BalancedSipser_d^{(1)} \uhr \tau$ and such CNFs to the bias of $\BalancedSipser_d^{(1)} \uhr \tau$.

\begin{proposition}
\label{prop:base-case}
Let $F: \zo^{w_0} \to \zo$ be a width-$r$ CNF and $\tau \in \{0,\ast\}^{w_0} \setminus \{0\}^{w_0}$.    Then for $\bY \leftarrow \{0_{1-{t_1}},1_{t_1}\}^{w_0}$,
\[ \Pr [(\OR_{w_0} \uhr \tau)(\bY) \ne F(\bY)]\big] \ge \bias(\OR_{w_0} \uhr \tau,\bY) - r{t_1}. \]
\end{proposition}

\begin{proof}
Writing $S = S(\tau) \sse [w_0]$ to denote the set $\tau^{-1}(\ast)$, we have that $\OR_{w_0} \uhr \tau$ computes the $|S|$-way $\OR$ of variables with indices in $S$ (note that $S\ne \emptyset$ since $\tau \in \{0,\ast\}^{w_0}\setminus \{0\}^{w_0}$); for notational brevity we will write $\OR_S$ instead of $\OR_{w_0} \uhr \tau$.

We begin with the claim that  there exists a CNF $F' : \zo^{w_0}\to \zo$ of size and width at most that of $F$, depending only on the variables in $S$, such that
\begin{equation}  \Pr[\OR_S(\bY) \ne F(\bY)] \ge \Pr[\OR_S(\bY) \ne F'(\bY)]. \label{eq:relevant-vars-only}
\end{equation}
This holds because
\begin{align*}
\Pr[\OR_S(\bY) \ne F(\bY)] &= \Ex_{\brho\leftarrow \{ 0_{1-{t_1}},1_{t_1}\}^{[w_0]\setminus S}} \Big[ \Pr[(\OR_S\uhr \brho)(\bY) \ne (F \uhr \brho)(\bY)] \Big] \\
&=  \Ex_{\brho\leftarrow \{ 0_{1-{t_1}},1_{t_1}\}^{[w_0]\setminus S}} \Big[ \Pr[\OR_S(\bY) \ne (F \uhr \brho)(\bY)] \Big],
\end{align*}
and so certainly there exists $\rho \in \zo^{[w_0]\setminus S}$ such that $F' := F\uhr \rho$ satisfies (\ref{eq:relevant-vars-only}).
Next, writing $\{ y_i \}_{i\in S}$ to denote the formal variables that both $\OR_S$ and $F'$ depend on, we consider two possible cases:
\begin{enumerate}
\item For every clause $T$ in $F'$ there exists $i\in S$ such that $\overline{y}_i$ occurs in $T$. In this case we note that $F'(0^S) = 1$ (whereas $\OR_S(0^S) = 0$), and so
\[ \Pr[\OR_S(\bY) \ne F'(\bY)] \ge \Pr[\text{$\bY_i = 0$ for all $i\in S$}] =  \Pr[\OR_S(\bY) = 0]. \]
\item Otherwise, there must exist a monotone clause $T$ in $F'$ (one containing only positive occurrences of variables) since $F'$ depends only on the variables in $S$.  In this case, since each unnegated literal is true with probability ${t_1}$ (recall that $\bY \leftarrow \{ 0_{1-{t_1}},1_{t_1}\}^{w_0}$) and $T$ has width at most $r$, by a union bound we have that
\[ \Pr[F'(\bY) = 1] \le \Pr[T(\bY) = 1] \le r{t_1}, \]
and so
\[ \Pr[\OR_S(\bY) \ne F'(\bY)] \ge \Pr[\OR_S(\bY) = 1] - \Pr[F'(\bY) = 1] \ge \Pr[\OR_S(\bY) = 1] - r{t_1}. \]
\end{enumerate}
Together, theses two cases give us the lower bound
\begin{align*} \Pr[\OR_S(\bY) \ne F'(\bY)] &\ge \min \big\{ \Pr[\OR_S(\bY) = 1], \Pr[\OR_S(\bY) = 0] - r{t_1} \big\} \\
&\ge \min\big\{ \Pr[\OR_S(\bY) = 1], \Pr[\OR_S(\bY) = 0]\big\} - r{t_1},
\end{align*}
which along with (\ref{eq:relevant-vars-only}) completes the proof.
\end{proof}

\subsection{Approximators with small bottom fan-in}

The pieces are in place to prove the first of our two main theorems, showing that $\BalancedSipser_d$ cannot be approximated by depth-$d$ size-$S$ circuits with bounded bottom fan-in:

\begin{reptheorem}{thm:smallbottomfanin}
\ignore{For infinitely many $n$, f}For $2 \leq d \leq {\frac {c\sqrt{\log n}}{\log \log n}}$, the $n$-variable $\BalancedSipser_d$ function has the following property:  Let $C : \zo^n \to \zo$ be any depth-$d$ circuit of size $S = 2^{n^{{\frac 1 {6(d-1)}}}}$ and bottom fan-in ${\frac {\log n}{10(d-1)}}$.  Then for a uniform random input $\bX \leftarrow \{ 0_{1/2}, 1_{1/2}\}^{n}$, we have
\[ \Pr[\BalancedSipser_d(\bX) \ne C(\bX)] \ge \frac1{2}  - {\frac 1 {n^{\Omega(1/d)}}}.\]
\end{reptheorem}

\begin{proof}
Let $\bY \leftarrow \{ 0_{1-t_1}, 1_{t_1}\}^{w_0}$. We successively apply Proposition~\ref{prop:complete-to-uniform} and Proposition~\ref{prop:what-happens-to-target} to obtain
\begin{align*}
 \Pr[\BalancedSipser_d(\bX) \ne C(\bX)]  &= \Ex_{\mathbf{\Psi}}\bigg[\Prx_{\bY}[(\mathbf{\Psi}(\BalancedSipser_d))(\bY) \ne (\mathbf{\Psi}(C))(\bY)] \bigg] \\
 &= \Ex_{\mathbf{\Psi}}\bigg[\Prx_{\bY}[(\OR_{w_0}\uhr \hat{\brho^{(2)}})(\bY) \ne (\mathbf{\Psi}(C))(\bY)] \bigg]
 \end{align*}
(for the second equality, recall that $\BalancedSipser^{(1)}_d$ is simply $\OR_{w_0}$, by our assumption from the start of the section that $d$ is even).
For every possible outcome $\Psi$ of $\mathbf{\Psi}$ (corresponding to successive outcomes of $\rho^{(d)}$ for $\brho^{(d)}$, $\dots,$ $\rho^{(2)}$ for $\brho^{(2)}$) and every $r\in \N$, we have the bound
\begin{align*} &\Prx_{\bY}[(\OR_{w_0}\uhr \hat{\rho^{(2)}})(\bY) \ne ({\Psi}(C))(\bY)] \\
& \ge
\Prx_{\bY}[(\OR_{w_0}\uhr \hat{\rho^{(2)}})(\bY) \ne ({\Psi}(C))(\bY) \mid \text{$\Psi(C)$ is a depth-$r$ DT} ]  - \ind[\text{$\Psi(C)$ is not a depth-$r$ DT} ] \\
& \ge \bias( \OR_{w_0} \uhr \hat{\rho^{(2)}}, \bY) - rt_1 - \ind[\text{$\Psi(C)$ is not a depth-$r$ DT} ],
\end{align*}
where the final inequality is by Proposition~\ref{prop:base-case} along with the fact that every depth-$r$ DT can be expressed as either a width-$r$ CNF or a width-$r$ DNF.  Setting $r = n^{{\frac 1 {4(d-1)}}}$ and taking expectation with respect to $\mathbf{\Psi}$, we conclude that
\begin{align*}  \Ex_{\mathbf{\Psi}}\bigg[\Prx_{\bY}[(\OR_{w_0}\uhr \hat{\brho^{(2)}})(\bY) \ne ({\Psi}(C))(\bY)]\bigg]
&\ge \Ex_{\mathbf{\Psi}}\bigg[ \bias(\OR_{w_0} \uhr \hat{\brho^{(2)}},\bY) \bigg] - rt_1 - \Prx_{\mathbf{\Psi}}[\text{$\mathbf{\Psi}(C)$ is not a depth-$r$ DT} ] \\
& \ge
\frac1{2} - \tilde{O}(w^{-1/12})  -  rt_1 - \exp\left(-\Omega(n^{\frac 1 {6(d-1)}})\right)\\
& \ge {\frac 1 2} - {\frac 1 {n^{\Omega(1/d)}}},
\end{align*}
where the second-to-last inequality  uses both Proposition~\ref{prop:boundexp} and Theorem \ref{thm:bounded-bottom-fan-in}, and the last claim follows by simple substitution, recalling the values of $r,t_1$ and $w$ in terms of $n$ and $d$.
\end{proof}

\subsection{Approximators with the opposite alternation pattern}

Our second main theorem states that $\BalancedSipser_d$ cannot be approximated by depth-$d$ size-$S$ circuits with the opposite alternation pattern to $\BalancedSipser_d$:

\begin{reptheorem}{thm:alternationpattern}
\ignore{For infinitely many $n$, f}For $2 \leq d \leq  {\frac {c\sqrt{\log n}}{\log \log n}}$, the $n$-variable $\BalancedSipser_d$ function has the following property:  Let $C : \zo^n \to \zo$ be any depth-$d$ circuit of size $S = 2^{n^{{\frac 1 {6(d-1)}}}}$ and the opposite alternation pattern to $\BalancedSipser_d,$ (i.e.~its top-level gate is $\OR$ if $\BalancedSipser_d$'s is $\AND$ and vice versa).
Then for a uniform random input $\bX \leftarrow \{ 0_{1/2}, 1_{1/2}\}^{n}$, we have
\[ \Pr[\BalancedSipser_d(\bX) \ne C(\bX)] \ge \frac1{2}  - {\frac 1 {n^{\Omega(1/d)}}}.\]
\end{reptheorem}

\begin{proof}
By our assumption that $d$ is even, we have that the top gate of $\BalancedSipser_d$ is a $w_0$-way $\OR$, whereas the top gate of $C$ is an $\AND$.  Let $\bY \leftarrow \{ 0_{1-{t_1}}, 1_{t_1}\}^{w_0}$.  As in the proof of Theorem~\ref{thm:smallbottomfanin}, we successively apply Proposition~\ref{prop:complete-to-uniform} and Proposition~\ref{prop:what-happens-to-target} to obtain
\begin{align*}
 \Pr[\BalancedSipser_d(\bX) \ne C(\bX)]  &= \Ex_{\mathbf{\Psi}}\bigg[\Prx_{\bY}[(\mathbf{\Psi}(\BalancedSipser_d))(\bY) \ne (\mathbf{\Psi}(C))(\bY)] \bigg] \\
 &= \Ex_{\mathbf{\Psi}}\bigg[\Prx_{\bY}[(\OR_{w_0}\uhr \hat{\brho^{(2)}})(\bY) \ne (\mathbf{\Psi}(C))(\bY)] \bigg].
 \end{align*}
For every possible outcome $\Psi = \rho^{(d)},\dots,\rho^{(2)}$ of $\mathbf{\Psi}$ and every $r\in \N$ we have the bound
\begin{align*} &\Prx_{\bY}[(\OR_{w_0}  \uhr \hat{\rho^{(2)}})(\bY) \ne ({\Psi}(C))(\bY)] \\ & \ge
\Prx_{\bY}[(\OR_{w_0}\uhr \hat{\rho^{(2)}})(\bY) \ne ({\Psi}(C))(\bY) \mid \text{$\Psi(C)$ is $(1/S)$-close to a width-$r$ CNF} ] \\
& \quad \quad\  - \ind[\text{$\Psi(C)$ is not $(1/S)$-close to a width-$r$ CNF} ] \\
&\ge \bias( \OR_{w_0} \uhr \hat{\rho^{(2)}}, \bY) - r{t_1} - (1/S)
 - \ind[\text{$\Psi(C)$ is not $(1/S)$-close to a width-$r$ CNF} ],
\end{align*}
where the final inequality is by Proposition~\ref{prop:base-case}.  As in the proof of Theorem~\ref{thm:smallbottomfanin}, setting $r = {n^{{\frac 1 {4(d-1)}}}}$ and taking expectation with respect to $\mathbf{\Psi}$, we conclude that
\begin{align*} & \Ex_{\mathbf{\Psi}}\bigg[\Prx_{\bY}[(\OR_{w_0} \uhr \hat{\brho^{(2)}})(\bY) \ne ({\Psi}(C))(\bY)]\bigg] \\
& \ge \Ex_{\mathbf{\Psi}}\bigg[\bias(\OR_{w_0} \uhr \hat{\brho^{(2)}},\bY) ]\big\}\bigg]
  - (1/S) - rt_1 - \Prx_{\mathbf{\Psi}}[\text{$\mathbf{\Psi}(C)$ is not $(1/S)$-close to a width-$r$ CNF} ] \\
& \ge \frac1{2} -\tilde{O}(w^{-1/12}) - (1/S) - rt_1 - \exp \left(-\Omega(n^{{\frac 1 {6(d-1)}}})\right)\\
& \ge {\frac 1 2} - {\frac 1 {n^{\Omega(1/d)}}},
\end{align*}
where the second-to-last inequality  uses both Proposition~\ref{prop:boundexp} and Theorem \ref{thm:different-alternation}, and the last claim follows by simple substitution, recalling the values of $r,{t_1},w$ and $S$ in terms of $n$ and $d.$
\end{proof}

%\newpage

\bibliography{allrefs}{}
\bibliographystyle{alpha}

%\newpage

\appendix

\section{Proof of Lemma~\ref{lemma:tk-bound}}
\label{ap:tk-bound}

\begin{replemma}{lemma:tk-bound}
There is a universal constant $c>0$ such that for $2 \leq d \leq {\frac {cm}{\log m}}$, we have that $t_k = q \pm q^{1.1}$ for all $k \in [d-1]$.
\end{replemma}

\begin{proof}
We shall establish the following bound, for $k=d-1,\dots,1$, by downward induction on $k$:
\begin{equation} \label{eq:tkq-close-to-p}
|t_k q - p| \leq (2m)^{d-1-k} \lambda.
\end{equation}
Lemma~\ref{lemma:tk-bound} follows directly from (\ref{eq:tkq-close-to-p}), using (\ref{eq:def-of-lambda-and-q}), (\ref{eq:def-of-w}) and the fact that $p = \Theta({\frac {\log w} w}).$

The base case $k=d-1$ of (\ref{eq:tkq-close-to-p}) holds with equality since (\ref{eq:def-of-tk}) gives us that $|t_{d-1}q-p| = \lambda.$
For the inductive step suppose that (\ref{eq:tkq-close-to-p}) holds for some value $k=\ell+1.$ By
(\ref{eq:def-of-tk}) we have that $t_\ell q = (1-t_{\ell+1})^{qw}-\lambda$, so our goal is to put upper and lower bounds on $(1-t_{\ell+1})^{qw}-\lambda$ that are close to $p$.  For the upper bound, we have
\begin{align*}
(1-t_{\ell+1})^{qw}  - \lambda &= \left( (1-t_{\ell+1})^{{\frac 1 {t_{\ell+1}}}}\right)^{qwt_{\ell+1}} - \lambda&\\
&\leq \exp \left(-q w t_{\ell+1}\right) - \lambda & \text{(by Fact~\ref{fact:approx})} \\
&\leq \exp \left(-w \left(p - (2m)^{d-\ell-2} \lambda\right)\right) - \lambda & \text{(by the inductive hypothesis)} \\
&\leq \exp \left(-\left({\frac {m 2^m}{\log e}} -1 \right) \cdot \left(2^{-m} - (2m)^{d-\ell-2} \lambda\right)\right) - \lambda\ & \text{(by (\ref{eq:def-of-w}))} \\
&= 2^{-m} \cdot \exp\left( 2^{-m} + \left({\frac {m2^m}{\log e}}-1\right) \cdot (2m)^{d-\ell-2} \lambda \right) - \lambda&\\
&\leq p \cdot \left(1 + 2^{-m+1} + {\frac {m2^{m+1}}{\log e}} \cdot (2m)^{d-\ell-2}\lambda\right) - \lambda & \text{(by Fact~\ref{fact:approx})} \\
&\leq p + 2^{-2m+1} + {\frac {2m}{\log e}} (2m)^{d-\ell-2} \lambda - \lambda&\\
&\leq p + (2m)^{d-\ell-1}\lambda, &
\end{align*}
where in the last inequality we have used the fact that $\lambda = \tilde{\Theta}(2^{-5m/4}).$

For the lower bound we proceed similarly:

\begin{align*}
(1-t_{\ell+1})^{qw} - \lambda &= \left( (1-t_{\ell+1})^{{\frac 1 {t_{\ell+1}}}}\right)^{qwt_{\ell+1}} - \lambda &\\
&\geq \exp \left(-q w t_{\ell+1}\right) \cdot (1-t_{\ell+1})^{qwt_{\ell+1}} - \lambda & \text{(by Fact~\ref{fact:approx})} \\
&\geq \exp \left(-w \left(p + (2m)^{d-\ell-2} \lambda\right)\right) \cdot (1 - qw(t_{\ell+1})^2) - \lambda & \text{(by the i.h. \& Fact \ref{fact:totally-standard})} \\
&\geq \exp \left(-{\frac {m 2^m}{\log e}} \cdot \left(2^{-m} + (2m)^{d-\ell-2} \lambda\right)\right)  \cdot (1 - qw(t_{\ell+1})^2) - \lambda  &\\
&= 2^{-m} \cdot \exp\left( - {\frac {m2^m}{\log e}} \cdot (2m)^{d-\ell-2} \lambda \right) \cdot (1 - qw(t_{\ell+1})^2)- \lambda &\\
&\geq 2^{-m} \cdot \left(1 - {\frac {m2^m}{\log e}} \cdot (2m)^{d-\ell-2} \lambda - qw(t_{\ell+1})^2\right) -\lambda & \text{(using Fact \ref{fact:totally-standard})} \\
&\geq 2^{-m} \cdot \left(1 - {\frac {m2^m}{\log e}} \cdot (2m)^{d-\ell-2} \lambda - {\frac w q} \cdot \left(p + (2m)^{d-\ell-2}\lambda\right)^2 \right) -\lambda & \text{(by the i.h.)} \\
&\geq 2^{-m} \cdot \left(1 - {\frac {m2^m}{\log e}} \cdot (2m)^{d-\ell-2} \lambda - {\frac {4wp^2} q}\right) -\lambda & \text{(by the bound on $d$)} \\
&= p - {\frac {m}{\log e}} \cdot (2m)^{d-\ell-2} \lambda - {\frac {4wp^3} q} - \lambda &\\
&\geq p - (2m)^{d-\ell-1} \lambda. & \qedhere
\end{align*}
\end{proof}

\end{document}